\documentclass{theoretics}

\usepackage{graphicx} % support the \includegraphics command and options
\usepackage{array} % for better arrays (eg matrices) in maths

\usepackage{amsmath, amssymb, amsfonts, verbatim}
\usepackage{hyphenat,epsfig,subcaption,multirow}
\usepackage{nicefrac}
\usepackage{mathtools} % For PairedDelimiter
\usepackage{mleftright}

\usepackage[font=footnotesize,labelfont=bf]{caption}

\DeclareFontFamily{U}{mathx}{\hyphenchar\font45}
\DeclareFontShape{U}{mathx}{m}{n}{
      <5> <6> <7> <8> <9> <10>
      <10.95> <12> <14.4> <17.28> <20.74> <24.88>
      mathx10
      }{}
\DeclareSymbolFont{mathx}{U}{mathx}{m}{n}
\usepackage{tcolorbox}
\tcbuselibrary{skins,breakable}
\tcbset{enhanced jigsaw}

\usepackage[normalem]{ulem}

\definecolor{DarkRed}{rgb}{0.5,0.1,0.1}
\definecolor{DarkBlue}{rgb}{0.1,0.1,0.5}

\usepackage{nameref}
\definecolor{ForestGreen}{rgb}{0.1333,0.5451,0.1333}
\definecolor{Red}{rgb}{0.9,0,0}
\usepackage[noabbrev,nameinlink]{cleveref}
\crefname{property}{property}{properties}
\crefname{fact}{fact}{facts}
\crefname{claim}{claim}{claims}
\crefname{observation}{observation}{observations}
\crefname{sidefigure}{sidefigure}{sidefigures}
\creflabelformat{property}{(#1)#2#3}
\crefname{equation}{eq.}{eqs.}
\creflabelformat{equation}{(#1)#2#3}

\usepackage{bm}
\usepackage{url}
\usepackage{xspace}
\usepackage[mathscr]{euscript}

\usepackage{tikz}
\usetikzlibrary{arrows}
\usetikzlibrary{arrows.meta}
\usetikzlibrary{shapes}
\usetikzlibrary{backgrounds}
\usetikzlibrary{positioning}
\usetikzlibrary{decorations.markings}
\usetikzlibrary{patterns}
\usetikzlibrary{calc}
\usetikzlibrary{fit}
\usetikzlibrary{snakes}

\usepackage{mdframed}

\usepackage[noend]{algpseudocode}
\makeatletter
\def\BState{\State\hskip-\ALG@thistlm}
\makeatother

\LinesNotNumbered

\SetAlFnt{\mathversion{sansmath}\fontsize{14.4}{20}\selectfont} % no
 \DeclareMathOperator*{\argmax}{arg\,max}

 \ThCSnewtheoita{lesson}
\ThCSnewtheoita{problem}

\renewcommand{\leq}{\leqslant}

\renewcommand{\geq}{\geqslant}

\newcommand{\Ot}{\ensuremath{\widetilde{O}}}
\newcommand{\eps}{\ensuremath{\varepsilon}}
\newcommand{\Paren}[1]{\Big(#1\Big)}

\newcommand{\bracket}[1]{\mleft[#1\mright]}
\newcommand{\paren}[1]{\ensuremath{\left(#1\right)}\xspace}
\newcommand{\card}[1]{\mleft\vert{#1}\mright\vert}

\DeclarePairedDelimiter{\cbrac}{\{}{\}}

\newcommand{\IN}{\ensuremath{\mathbb{N}}}

\newcommand{\ceil}[1]{{\mleft\lceil{#1}\mright\rceil}}

\newcommand{\prob}[1]{\Pr\paren{#1}}
\newcommand{\expect}[1]{\Exp\bracket{#1}}

\newcommand{\set}[1]{\ensuremath{\mleft\{ #1 \mright\}}}
\newcommand{\poly}{\mbox{\rm poly}}

\DeclareMathOperator*{\Exp}{\ensuremath{{\mathbb{E}}}}
\DeclareMathOperator*{\Prob}{\ensuremath{\textnormal{Pr}}}
\renewcommand{\Pr}{\Prob}

\newenvironment{tbox}{\begin{tcolorbox}[
		enlarge top by=5pt,
		enlarge bottom by=5pt,
		 breakable,
		 boxsep=0pt,
                  left=4pt,
                  right=4pt,
                  top=10pt,
                  arc=0pt,
                  boxrule=1pt,toprule=1pt,
                  colback=white
                  ]%%
	}
{\end{tcolorbox}}

\newcommand{\II}{\ensuremath{\mathbb{I}}}

\newcommand{\mireal}[1][]{
  \ifx\relax#1\relax%
    \II(\mione \,; \mitwo)%
  \else%
    \II(\mione \,; \mitwo\mid #1)%
  \fi
}

\newcommand{\MM}{\ensuremath{{M}}}

\renewcommand{\SS}{\ensuremath{\mathcal{S}}}

\newcommand{\Gbase}{\ensuremath{\mathcal{G}_{\textnormal{\texttt{Base}}}}\xspace}
\newcommand{\Ebase}{\ensuremath{\mathcal{E}_{\textnormal{\texttt{Base}}}}\xspace}

\newcommand{\Gsample}{\ensuremath{\mathcal{G}_{\textnormal{\texttt{Sample}}}}\xspace}
\newcommand{\Esample}{\ensuremath{\mathcal{E}_{\textnormal{\texttt{Sample}}}}\xspace}

\newcommand{\F}{\mathbb{F}}

\newcommand{\inner}[2]{\langle #1, #2 \rangle}

\newcommand{\cL}{\ensuremath{\mathcal{L}}}
\newcommand{\cR}{\ensuremath{\mathcal{R}}}

\newcommand{\OR}{\textnormal{\textsc{or}}}
\newcommand{\ORone}{\textnormal{\textsc{or-one}}}
\newcommand{\AndORone}{\textnormal{\textsc{and-or-one}}}
\newcommand{\tribes}{\textnormal{\textsc{tribes}}}
\newcommand{\disjoint}{\textnormal{\textsc{disj}}}
\newcommand{\clique}{\textnormal{\textsc{clique}}}

\ThCSauthor[rutgers,waterloo]{Sepehr Assadi}{sepehr@assadi.info}[0009-0006-8914-5995]
\ThCSauthor[charles]{Pankaj Kumar}{pankaj@kam.mff.cuni.cz}
\ThCSauthor[rutgers]{Parth Mittal}{parth.mittal@rutgers.edu}
\ThCSaffil[rutgers]{Department of Computer Science, Rutgers
  University, USA}
  \ThCSaffil[waterloo]{Cheriton School of Computer Science, University of Waterloo, Canada}
\ThCSaffil[charles]{Department of Applied
Mathematics, Faculty of Mathematics and Physics, 
Charles University, Prague, \v{C}zech Republic}

\ThCSyear{2023}
\ThCSarticlenum{9}
\ThCSreceived{Sep 11, 2022}
\ThCSaccepted{Apr 13, 2023}
\ThCSpublished{August 3, 2023}
\ThCSkeywords{Semi-streaming algorithms, Graph coloring, Brooks' theorem}
\ThCSshortnames{S.\ Assadi, P.\ Kumar, P.\ Mittal} 

\title{Brooks' Theorem in Graph Streams:  A Single-Pass
  Semi-Streaming Algorithm for \texorpdfstring{$\Delta$}{Delta}-Coloring}
\ThCSthanks{An extended abstract of this paper appeared in ACM
  Symposium on Theory of Computing (STOC'22)~\cite{AssadiKM22}.

Sepehr Assadi was supported in part by an NSF CAREER Grant
CCF-2047061, a Sloan Research Fellowship, a Google Research gift, and
a Fulcrum award from Rutgers Research Council. Parth Mittal was
supported in part by the NSF CAREER grant CCF-2047061.  Part of this research
was carried out while Pankaj Kumar and Parth Mittal  participated
 in the 2020 DIMACS REU program, supported by CoSP, a project funded by
 European Union’s Horizon 2020 research and innovation program, grant
 agreement No. 823748.}

\ThCSshorttitle{Brooks' Theorem in Graph Streams}

\begin{document}
\maketitle

%\pagenumbering{roman}

\begin{abstract}
	Every graph with maximum degree $\Delta$ can be colored with $(\Delta+1)$ colors using a simple greedy algorithm. 
	Remarkably, recent work has shown that one can find such a coloring even in the semi-streaming model: there exists a randomized algorithm
	that with high probability finds a $(\Delta+1)$-coloring of the input graph in only $O(n\cdot\poly\!\log{n})$ \emph{space} assuming a single pass over 
	the edges of the graph in any arbitrary order. But, in reality, one almost never needs $(\Delta+1)$ colors to properly color a graph. Indeed, the celebrated \textbf{Brooks' theorem} states that every (connected) graph beside cliques and odd cycles can 
	be  colored with $\Delta$ colors. Can we find a $\Delta$-coloring in the semi-streaming model as well? 
	
	\medskip
	
	We settle this key question in the affirmative by designing a randomized semi-streaming algorithm that given any graph, with high probability, either correctly declares that the graph is not $\Delta$-colorable 
	or outputs a $\Delta$-coloring of the graph. 
	
	\medskip
	
	The proof of this result starts with a detour. We first (provably) identify the extent to which the previous approaches for streaming coloring fail for $\Delta$-coloring: for instance, all these prior approaches can 
	handle streams with repeated edges and they can run in $o(n^2)$ time, whereas
  prove that neither of these tasks is possible for $\Delta$-coloring. These impossibility results however 
	pinpoint exactly what is missing from prior approaches when it comes to $\Delta$-coloring. 
	
	\medskip
	
	We  build on these insights to design a semi-streaming algorithm that uses $(i)$ a novel sparse-recovery approach based on sparse-dense decompositions 
	to (partially) recover the ``problematic'' subgraphs of the input---the ones that form the basis of our impossibility results---and $(ii)$ a new coloring approach
	for these subgraphs that allows for recoloring of other vertices in a controlled way without relying on local explorations or finding ``augmenting paths'' that 
	are generally impossible for semi-streaming algorithms. We believe both these techniques can be of independent interest. 
	
\end{abstract}

%%

% !TeX root = main.tex 
%!TEX root = main.tex

\section{Introduction}\label{sec:intro}

Graph coloring problems are ubiquitous in graph theory and computer science. Given a graph $G=(V,E)$, a 
proper $c$-coloring of $G$ is any assignment of colors from the palette $\set{1,\ldots,c}$ to the vertices so that no edge receives 
the same color on both its endpoints. Recent years have witnessed a flurry of results for graph coloring in the graph streaming model~\cite{RadhakrishnanS11,BeraG18,AssadiCK19a,BehnezhadDHKS19,BeraCG19,AlonA20,BhattacharyaBMU21,AssadiCS22,ChakrabartiGS22,HalldorssonKNT22}. In this model, the edges of the input graph arrive one by one in an arbitrarily ordered stream and 
the algorithm needs to process these edges sequentially using a limited space, much smaller than the input size. Of particular interest are \emph{semi-streaming} algorithms, introduced by~\cite{FeigenbaumKMSZ05}, 
that use only $\Ot(n) := O(n \cdot \poly\!\log{n})$ space\footnote{Throughout, we use $\Ot(f) := O(f \cdot \poly\log{f})$ to suppress $\poly\!\log{(f)}$ factors.} on $n$-vertex graphs which is proportional to the output size. 
We focus on this model in this paper. 

One of the simplest forms of graph coloring problems is $(\Delta+1)$-coloring of graphs with maximum degree $\Delta$. 
Not only does every graph admits a $(\Delta+1)$-coloring, one can in fact find one quite easily via a greedy algorithm: iterate over the vertices and color each one from any of 
$(\Delta+1)$ colors that has not appeared in any of its at most $\Delta$ colored neighbors. Yet, despite its utter simplicity, this algorithm does not easily lend itself to a semi-streaming algorithm 
as the arbitrary arrival of edges prohibits us from coloring vertices one at a time. 

Nonetheless, a breakthrough of~\cite{AssadiCK19a} showed that $(\Delta+1)$ 
coloring is still possible in the semi-streaming model, albeit via a randomized algorithm that employs a ``non-greedy'' approach. In particular,~\cite{AssadiCK19a} proved the following \emph{palette sparsification theorem}: if we sample 
$O(\log{n})$ colors from $\set{1,\ldots,\Delta+1}$ for each vertex independently, then with high probability, the entire graph can be  colored by coloring each vertex from its own sampled colors. 
This result immediately leads to a semi-streaming algorithm for $(\Delta+1)$-coloring: sample these colors for each vertex and store any edge in the stream that can potentially become monochromatic 
under any coloring of vertices from their sampled list. A simple probabilistic analysis bounds the number of stored edges by $O(n\log^2{n})$ with high probability, and the palette sparsification theorem guarantees that one 
can find a $(\Delta+1)$-coloring of the graph at the end of the stream. 

Going back to  existential results, it is easy to see that there are graphs that do need $\Delta+1$ colors for proper coloring, for instance $(\Delta+1)$-cliques or odd cycles (where $\Delta=2$). The celebrated \textbf{Brooks' theorem}~\cite{Brooks41} states that these two are the only examples: any (connected) graph besides cliques and odd cycles can be colored with $\Delta$ colors (see also~\cite{MelnikovV69} and~\cite{Lovasz75}  for other classical proofs
 of this result by Melnikov and Vizing, and by Lov\'asz, respectively). Unlike existence of $(\Delta+1)$-colorings which is rather a triviality, Brooks' theorem turned out to be a fundamental result in graph coloring~\cite{MolloyR13,StiebitzT15} with numerous proofs discovered for it over the years; see, e.g., \cite{StiebitzT15,CranstonR15,Rabern14,Rabern22} and references therein. The algorithmic aspects of Brooks' theorem 
 have also been studied extensively in classical algorithms~\cite{Lovasz75,Skulrattanakulchai02,BaetzW14}, PRAM algorithms~\cite{Karloff89,PanconesiS95,KarchmerN88,HajnalS90}, or LOCAL algorithms~\cite{PanconesiS92,BrandtFHKLRSU16,GhaffariHKM18}. 
 
 Given the key role Brooks' theorem plays in graph coloring literature on one hand, and the recent advances on streaming coloring algorithms on the other hand, it is thus quite natural to ask: 
 \begin{quote}
 	\emph{Does there exist a \textbf{``semi-streaming Brooks' theorem''}, namely, a semi-streaming algorithm that colors any given graph, besides cliques and odd cycles, with $\Delta$ colors?} 
 \end{quote}
This is precisely the question addressed in this paper. We emphasize that our interest in this question is not in ``shaving off'' a single color from $(\Delta+1)$-coloring to $\Delta$-coloring in practice, but rather as a source of insights and ideas (as is the case, say, in graph theory or classical algorithms where $(\Delta+1)$-coloring is just a triviality). In fact, $\Delta$-coloring appears to be just beyond the reach of our current techniques. For instance, previous streaming coloring algorithms in~\cite{AssadiCK19a,BeraCG19,AlonA20} can all be obtained via palette sparsification (see~\cite{AlonA20} for details). Yet, it was already observed in~\cite{AssadiCK19a} that 
palette sparsification cannot handle $\Delta$-coloring (we elaborate on this later). More generally, while $(\Delta+1)$-coloring  has a strong ``greedy nature'', all existential/algorithmic proofs of $\Delta$-coloring are based on ``exploring” the graph for certain structures, say cut vertices or certain spanning trees~\cite{Lovasz75}, Kempe Chains~\cite{MelnikovV69}, Rubin's Block Lemma~\cite{ErdosRT79,AlonT92}, or ``augmenting paths''~\cite{PanconesiS95} to name a few (we refer
the interested reader to~\cite{StiebitzT15} for an excellent overview of various proofs of Brooks' theorem). These (local) exploration
tasks however tend to be generally impossible in the semi-streaming model.%
\footnote{For instance, while computing all neighbors of a given vertex is trivial via a semi-streaming algorithm (by storing edges of the vertex), 
it is even impossible to discover the $2$-hop neighborhood of a given vertex~\cite{FeigenbaumKMSZ08}.}

\subsection{Our Contributions}\label{sec:results} 

We start with studying the \emph{limitations} of the current approaches in streaming graph coloring for solving $\Delta$-coloring. To do so, 
we focus on two common characteristics of all prior algorithms in~\cite{AssadiCK19a,BeraCG19,AlonA20}: they all also naturally lead to $(i)$ \emph{sublinear-time} algorithms for the corresponding coloring problems that run in $(n^{3/2+o(1)})$ time, 
and $(ii)$ semi-streaming algorithms that can  handle \emph{repeated-edge} streams wherein the same edge may appear more than once. We prove that obtaining either type of algorithms is provably \emph{impossible} for $\Delta$-coloring: 
\begin{quote}
  \begin{itemize}%[leftmargin=15pt]
  \item \textbf{Sublinear-time algorithms (\Cref{sec:sub-time}):} Any
    algorithm that, given access to adjacency lists and adjacency
    matrix of a graph with known maximum degree $\Delta$, can output a
    $\Delta$-coloring with large constant probability requires
    $\Omega(n\Delta)$ queries to input and time.
	
  \item \textbf{Repeated-edge streams (\Cref{sec:or-streams}):} Any
    algorithm that, given the edges of a graph with known maximum
    degree $\Delta$ in a \emph{repeated-edge} stream, can output a
    $\Delta$-coloring with large constant probability requires
    $\Omega(n\Delta)$ space.
	
  \end{itemize}
\end{quote}

These impossibility results already demonstrate how different $\Delta$-coloring is compared to prior graph coloring problems studied in the semi-streaming model. But, as we shall elaborate later, these results play a much more important role 
for us: they pinpoint what is missing from prior approaches when it comes to the $\Delta$-coloring problem and act as an excellent guide for addressing
our  motivating question. This brings us to the {main contribution} of our work. 

\begin{theorem}[Semi-Streaming Brooks' Theorem]\label{res:main}
	There exists a randomized semi-streaming algorithm that given any connected graph $G=(V,E)$ with maximum degree $\Delta$, which is not a clique nor an odd-cycle, 
	with high probability, outputs a $\Delta$-coloring of $G$. 
\end{theorem}
Consequently, despite the fact that prior approaches inherently fail for $\Delta$-coloring in fundamental ways and that $\Delta$-coloring is provably 
intractable in closely related models, we can still obtain a semi-streaming Brooks' theorem and settle our motivating question in the affirmative. 
It is also worth mentioning that randomness in~\Cref{res:main} is crucial: a very recent result of~\cite{AssadiCS22} shows that deterministic semi-streaming algorithms 
cannot even find an $\exp\paren{\Delta^{o(1)}}$-coloring. Our~\Cref{res:main} thus fully settles the complexity of the $\Delta$-coloring problem
in the semi-streaming model. 

\Cref{res:main} can be stated more generally as an algorithm that either decides whether the input graph is $\Delta$-colorable or not, and if yes, outputs the coloring. This is because checking whether 
a graph is $\Delta$-colorable can be done by simply storing a spanning forest of the input (see, e.g.,~\cite{FeigenbaumKMSZ05}) and maintaining the degrees of vertices; this allows us to check whether any of the connected 
components in the graph is a $(\Delta+1)$-clique or an odd-cycle. If not, applying~\Cref{res:main} to each connected component of the graph (in parallel in a single pass) gives us the desired $\Delta$-coloring (the algorithm does not 
even require the prior knowledge of $\Delta$ using a standard trick observed in~\cite{AssadiCK19a}; see \Cref{rem:no-know-delta}). However, we find the statement of~\Cref{res:main} to best capture the most interesting part of the result
and thus opted to present it in this form. 

 \paragraph{Our Techniques.} We shall go over our techniques in detail in the streamlined overview of our approach in~\Cref{sec:tech}. For now, we only mention the three main technical ingredients of our work: 
% \begin{quote}
   \begin{enumerate}[label=$\roman*).$]%,leftmargin=20pt]
   \item A thorough understanding of the powers and limitations of the
     palette sparsification approach of~\cite{AssadiCK19a} for
     $\Delta$-coloring via a rough characterization of which
     (sub)graphs it still applies to;
   \item An algorithm for \emph{implicitly} identifying and storing
     ``problematic'' subgraphs of the input graph---the ones that
     cannot be handled by palette sparsification approach of previous
     step---via a novel sparse recovery approach that relies on
     \emph{algorithmic} use of sparse-dense decompositions (see
     \Cref{sec:decomposition}) in place of their \emph{analytic} use
     in prior streaming algorithms~\cite{AssadiCK19a,AlonA20};
   \item A new coloring procedure that combines simple graph theoretic
     ideas with probabilistic analysis of palette sparsification using
     a notion of \emph{helper structures}; these are simple subgraphs
     of the input that can be recovered via our semi-streaming
     algorithms from the previous part and does not rely on local
     exploration steps of prior proofs of Brooks' theorem mentioned
     earlier.
   \end{enumerate}
% \end{quote}

\paragraph{Other Sublinear Algorithms Models.} Prior semi-streaming algorithms for graph coloring  also naturally lead to a series of algorithmic results for the respective problems 
in other models. Our first impossibility result already rules out this possibility for $\Delta$-coloring when it comes to \emph{sublinear-time} algorithms. Nevertheless, our approach in~\Cref{res:main} is still 
quite flexible and thus allows for extension of this algorithm to many other models. In particular, the algorithm is implemented via a \emph{linear sketch} (see~\cite{McGregor14}), which immediately 
implies the following two results as well: 
\begin{itemize}
	\item \textbf{Dynamic streams:} There exists a (single-pass) randomized semi-streaming algorithm for $\Delta$-coloring on the streams that contain insertion and deletion of edges. 
	\item \textbf{Massively parallel computation (MPC):} There exist a one round randomized MPC algorithm for $\Delta$-coloring on machines of memory $\Ot(n)$ with only $\Ot(n)$ extra global memory. 
\end{itemize}
As this is not the focus of the paper, we omit the definition and details of the models and instead refer the interested to~\cite{AhnGM12,McGregor14} and~\cite{KarloffSV10,BeameKS17} for each model, respectively. 

\subsection{Related Work}\label{sec:related}  

Similar to the classical setting, it is known that approximating the minimum number of colors for proper coloring, namely, the \emph{chromatic number}, is  intractable in the semi-streaming model~\cite{HalldorssonSSW12,AbboudCKP19,CormodeDK19}. 
Thus, recent work has focused instead on ``combinatorially optimal'' bounds---termed by~\cite{HalldorssonHLS16}---for streaming coloring problems.
On this front, we already discussed the $(\Delta+1)$-coloring result of~\cite{AssadiCK19a}. 
Independently and concurrently,~\cite{BeraG18} obtained a semi-streaming algorithm for $O(\Delta)$ colorings. These results were followed by semi-streaming algorithms for other 
 coloring problems such as degeneracy coloring~\cite{BeraCG19}, coloring locally sparse graphs and $(\deg+1)$-coloring~\cite{AlonA20}, $(\deg+1)$-{list} coloring~\cite{HalldorssonKNT22}, adversarially robust coloring~\cite{ChakrabartiGS22}, 
edge-coloring (in W-streams)~\cite{BehnezhadDHKS19}, deterministic lower bounds and (multi-pass) algorithms~\cite{AssadiCS22}, and coloring verification problems~\cite{BhattacharyaBMU21}. 
Moreover,~\cite{AlonA20} studied
various aspects of palette sparsification technique of~\cite{AssadiCK19a} and showed that other semi-streaming coloring algorithms such as~\cite{BeraG18,BeraCG19} can also be obtained via this technique. 

Many of these work on streaming algorithms for graph coloring also extend to other models such as sublinear-time and massively parallel computation (MPC) algorithms. For instance, 
for $(\Delta+1)$-coloring, there are randomized sublinear-time algorithms in $\Ot(n^{3/2})$ time~\cite{AssadiCK19a} or deterministic MPC algorithms with $O(1)$ rounds and $O(n)$ per-machine memory~\cite{CzumajDP20} (see also~\cite{CzumajDP21b}). 
Moreover, subsequent to the conference publication of this paper in~\cite{AssadiKM22}, some of the ideas in our work was also used in~\cite{FischerHM23} in designing distributed LOCAL algorithms for $\Delta$-coloring.  

Numerous beautiful algorithmic results are known for $\Delta$-coloring problem in various other models such as classical algorithms~\cite{Lovasz75,Skulrattanakulchai02,BaetzW14}, PRAM algorithms~\cite{Karloff89,PanconesiS95,KarchmerN88,HajnalS90}, or LOCAL algorithms~\cite{PanconesiS92,BrandtFHKLRSU16,GhaffariHKM18,BalliuBKO22}. For instance, a remarkable ``distributed Brooks' theorem'' of~\cite{PanconesiS95} 
proves that any partial $\Delta$-coloring of all but one vertex of the graph, can be turned into a proper $\Delta$-coloring of the entire graph by recoloring a single ``augmenting path'' of  $O(\log_{\Delta}{\!n})$ length. 
Finally, it is worth mentioning that Brooks' theorem 
is part of a more general phenomenon in graph theory: as the maximum clique size in $G$ moves further away from $\Delta+1$, so does its chromatic number. For instance,~\cite{Reed99b} proves that 
for sufficiently large $\Delta$, if a graph does not contain a $\Delta$-clique, then it is in fact always $(\Delta-1)$-colorable; see, e.g.,~\cite{BorodinK77,Reed98,Reed99b,KostochkaRS12,MolloyR13,MolloyR14} and references therein for various other 
examples. 

% !TeX root = main.tex 
%!TEX root = main.tex

\section{Technical Overview}\label{sec:tech} 

We now give a streamlined overview of our approach. While~\Cref{res:main} is by far the main contribution of our work, we find it illuminating 
to first  talk about our impossibility results for $\Delta$-coloring as they, despite their  simplicity, played a crucial role for us in obtaining~\Cref{res:main} and we believe they can shed more light into
different components of our final algorithm. 

\subsection{A Detour: Impossibility Results, Barriers, and Lessons Along the Way}\label{sec:tech-lower}

\textbf{Palette sparsification.} Let us start by  reviewing the palette sparsification theorem of~\cite{AssadiCK19a}: if we sample $O(\log{n})$ colors from $\set{1,\ldots,\Delta+1}$ for each vertex independently, then with high probability, 
we can still color the graph by coloring each vertex from its sampled palette. The proof of this result in~\cite{AssadiCK19a} uses a variant of \emph{sparse-dense decomposition}~\cite{Reed98} that partitions the graph into \emph{``(locally) 
sparse''} vertices and a collection of \emph{``almost-clique''} subgraphs that can be turned into $(\Delta+1)$-cliques by changing $o(1)$ fraction of their vertices and edges (\Cref{fig:decomposition}). The sparse vertices are then colored \emph{one at a time} 
from their sampled lists using a standard greedy coloring argument originally introduced in~\cite{MolloyR97}.  The main part of the proof is to handle almost-cliques by going over them one by one and coloring each one \emph{entirely}, using the sampled lists 
of \emph{all} its vertices at the same time, even assuming the outside vertices are colored \emph{adversarially}. 

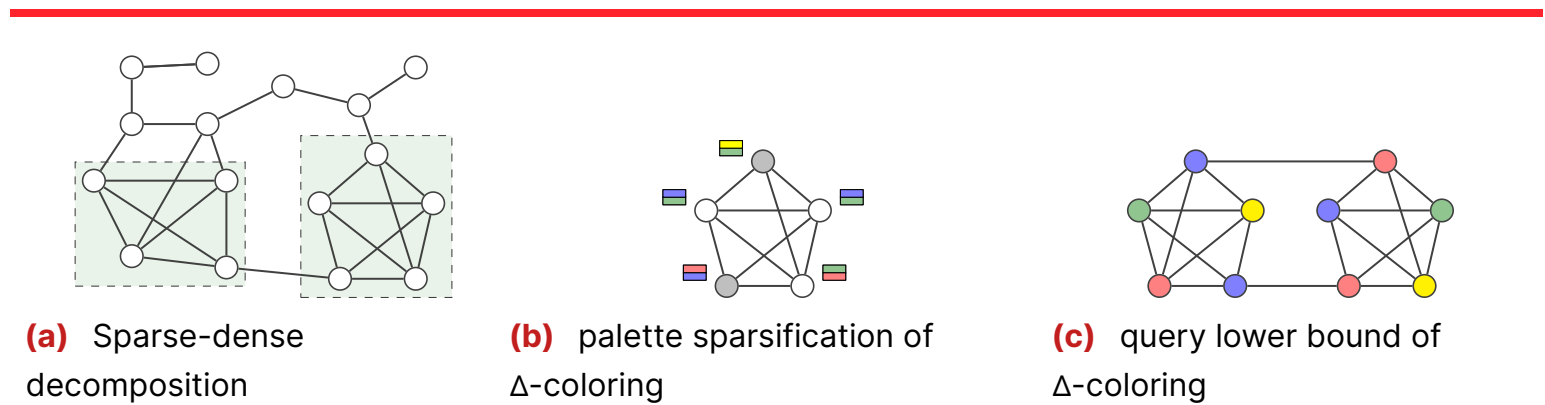
\begin{figure}[h!]
\centering
\subcaptionbox{\footnotesize Sparse-dense decomposition \label{fig:decomposition}}%
  [.31\linewidth]{\definecolor {processblue}{cmyk}{0.96,0,0,0}
\tikzset{myptr/.style={decoration={markings,mark=at position 1 with %
    {\arrow[scale=3,>=stealth]{>}}},postaction={decorate}}}

\newcommand{\colorgrid}{}

\newcommand{\nodenum}[1]{}%{{\fontsize{1pt}{1pt}\selectfont #1}}

\begin {tikzpicture}[ auto ,node distance =1cm and 1cm , on grid , semithick, inner sep=3pt]

\node[circle, black!75, fill=white, draw]  (V1){\nodenum{1}};

\node[circle, black!75, fill=white, draw]    (V2)[above right=0.05cm and 1cm of V1]{\nodenum{ 2}};

\node[circle, black!75, fill=white, draw]    (V3)[below=0.75cm of V1]{\nodenum{ 3}};

\node[circle, black!75, fill=white, draw]    (V4)[right=of V3]{\nodenum{ 4}};

\node[circle, black!75,   fill=white, draw]  (V5)[below left=0.75cm and 0.5cm of V3]{\nodenum{ 5}};

\node[circle, black!75,  fill=white, draw]  (V6)[right=1.75cm of V5]{\nodenum{ 6}};

\node[circle, black!75,  fill=white, draw]  (V7)[below right=1cm and 0.5cm of V5]{\nodenum{ 7}};

\node[circle, black!75,  fill=white, draw]  (V8)[below right=0.15cm and 1.25cm of V7]{\nodenum{ 8}};

\node[circle, black!75, fill=white, draw]   (V11)[above right =0.5cm and 1cm  of V4]{\nodenum{ 11}};

\node[circle, black!75, fill=white, draw]   (V10)[below right=0.25cm and 1cm of V11]{\nodenum{ 10}};

\node[circle, black!75, fill=white, draw]   (V13)[above right=0.5cm and 0.75cm of V10]{\nodenum{ 13}};

\node[circle, black!75,  fill=white, draw]  (V14)[below left=2.3cm and 0.25cm of V10]{\nodenum{ 14}};

\node[circle, black!75,  fill=white, draw]  (V15)[below right=0cm and 1cm of V14]{\nodenum{ 15}};

\node[circle, black!75,  fill=white, draw]  (V16)[above left=1cm and 0.27cm of V14]{\nodenum{ 16}};

\node[circle, black!75,  fill=white, draw]  (V17)[above right=0.65cm and 0.75cm of V16]{\nodenum{ 17}};

\node[circle, black!75,  fill=white, draw]  (V22)[below right=0.65cm and 0.75cm of V17]{\nodenum{ 22}};

%%\node[circle, blue, fill=white, draw,  line width=1pt]   (V18)[below right=0.75cm and 0.25 cm of V15]{\nodenum{ 18}};
%%\node (C18)[above right=5pt and 12pt of V18]{\colorgrid};
%%
%%\node[circle, blue, fill=white, draw,  line width=1pt]   (V19)[below right=1.75cm and 0.25 cm of V18]{\nodenum{ 19}};
%%\node (C19)[above right=5pt and 12pt of V19]{\colorgrid};
%%
%%\node[circle, blue, fill=white, draw,  line width=1pt]    (V20)[above left=0.25cm and 1.25 cm of V19]{\nodenum{ 20}};
%%\node (C20)[below left=5pt and 12pt of V20]{\colorgrid};
%%
%%\node[circle, blue, fill=white, draw,  line width=1pt]   (V21)[below right=0.5cm and 0.5 cm of V20]{\nodenum{ 21}};
%%\node (C21)[below left=5pt and 12pt of V21]{\colorgrid};
%%

\begin{scope}[on background layer]
\node[black!75, draw, fill=ForestGreen!10, inner sep=2.5pt, fit=(V5) (V6) (V7) (V8), dashed]{};
\node[black!75, draw, fill=ForestGreen!10, inner sep=2.5pt, fit=(V14) (V15) (V16) (V17) (V22), dashed]{};
\end{scope}

\draw[color=black!75, line width=0.75pt](V1) to (V2);
\draw[color=black!75, line width=0.75pt](V4) to (V7);
\draw[color=black!75, line width=0.75pt](V3) to (V4);
\draw[color=black!75, line width=0.75pt](V1) to (V2);
\draw[color=black!75, line width=0.75pt](V1) to (V3);
\draw[color=black!75, line width=0.75pt](V3) to (V5);
\draw[color=black!75, line width=0.75pt](V4) to (V11);
\draw[color=black!75, line width=0.75pt](V4) to (V6);
\draw[color=black!75, line width=0.75pt](V5) to (V7);
\draw[color=black!75, line width=0.75pt](V7) to (V6);
\draw[color=black!75, line width=0.75pt](V7) to (V8);
\draw[color=black!75, line width=0.75pt](V6) to (V8);
\draw[color=black!75, line width=0.75pt](V6) to (V5);
\draw[color=black!75, line width=0.75pt](V8) to (V5);

\draw[color=black!75, line width=0.75pt](V10) to (V11);
\draw[color=black!75, line width=0.75pt](V14) to (V8);
\draw[color=black!75, line width=0.75pt](V10) to (V17);
\draw[color=black!75, line width=0.75pt](V10) to (V13);

\draw[color=black!75, line width=0.75pt](V15) to (V22);
\draw[color=black!75, line width=0.75pt](V22) to (V14);
\draw[color=black!75, line width=0.75pt](V16) to (V14);
\draw[color=black!75, line width=0.75pt](V16) to (V22);
\draw[color=black!75, line width=0.75pt](V16) to (V17);
\draw[color=black!75, line width=0.75pt](V17) to (V15);
\draw[color=black!75, line width=0.75pt](V16) to (V15);
\draw[color=black!75, line width=0.75pt](V22) to (V17);
\draw[color=black!75, line width=0.75pt](V14) to (V15);

\end{tikzpicture}}
  \subcaptionbox{\footnotesize palette sparsification of $\Delta$-coloring \label{fig:ps-fail}}%
  [.33\linewidth]{\definecolor {processblue}{cmyk}{0.96,0,0,0}
\tikzset{myptr/.style={decoration={markings,mark=at position 1 with %
    {\arrow[scale=3,>=stealth]{>}}},postaction={decorate}}}

\newcommand{\nodenum}[1]{}%{{\fontsize{1pt}{1pt}\selectfont #1}}

\newcommand{\colorgridbr}{\tikz{
\draw[xstep=0.3cm,ystep=0.1cm] (0,0)  grid (0.3,0.2); 
\filldraw[fill=blue!50] (0,0) rectangle (0.3,0.1);
\filldraw[fill=red!50] (0,0.1) rectangle (0.3,0.2);
}}

\newcommand{\colorgridrf}{\tikz{
\draw[xstep=0.3cm,ystep=0.1cm] (0,0)  grid (0.3,0.2); 
\filldraw[fill=red!50] (0,0) rectangle (0.3,0.1);
\filldraw[fill=ForestGreen!50] (0,0.1) rectangle (0.3,0.2);
}}

\newcommand{\colorgridbg}{\tikz{
\draw[xstep=0.3cm,ystep=0.1cm] (0,0)  grid (0.3,0.2); 
\filldraw[fill=blue!75] (0,0) rectangle (0.3,0.1);
\filldraw[fill=Gray] (0,0.1) rectangle (0.3,0.2);
}}

\newcommand{\colorgridfy}{\tikz{
\draw[xstep=0.3cm,ystep=0.1cm] (0,0)  grid (0.3,0.2); 
\filldraw[fill=ForestGreen!50] (0,0) rectangle (0.3,0.1);
\filldraw[fill=Yellow] (0,0.1) rectangle (0.3,0.2);
}}

\newcommand{\colorgridry}{\tikz{
\draw[xstep=0.3cm,ystep=0.1cm] (0,0)  grid (0.3,0.2); 
\filldraw[fill=red!50] (0,0) rectangle (0.3,0.1);
\filldraw[fill=Yellow] (0,0.1) rectangle (0.3,0.2);
}}

\newcommand{\colorgridrg}{\tikz{
\draw[xstep=0.3cm,ystep=0.1cm] (0,0)  grid (0.3,0.2); 
\filldraw[fill=red!75] (0,0) rectangle (0.3,0.1);
\filldraw[fill=Gray] (0,0.1) rectangle (0.3,0.2);
}}

\newcommand{\colorgridrp}{\tikz{
\draw[xstep=0.3cm,ystep=0.1cm] (0,0)  grid (0.3,0.2); 
\filldraw[fill=red!50] (0,0) rectangle (0.3,0.1);
\filldraw[fill=ForestGreen!50] (0,0.1) rectangle (0.3,0.2);
}}

\newcommand{\colorgridfp}{\tikz{
\draw[xstep=0.3cm,ystep=0.1cm] (0,0)  grid (0.3,0.2); 
\filldraw[fill=ForestGreen!50] (0,0) rectangle (0.3,0.1);
\filldraw[fill=blue!50] (0,0.1) rectangle (0.3,0.2);
}}
\begin {tikzpicture}[ auto ,node distance =1cm and 1cm , on grid , semithick, inner sep=3pt]

\node[circle, black!75,  fill=black!25, draw]  (V14){};%[below left=2.3cm and 0.25cm of V10]{\nodenum{ 14}};
\node (C14)[above left=5pt and 12pt of V14]{\colorgridbr};

\node[circle, black!75,  fill=white, draw]  (V15)[below right=0cm and 1cm of V14]{\nodenum{ 15}};
\node (C15)[above right=5pt and 12pt of V15]{\colorgridrp};

\node[circle, black!75,  fill=white, draw]  (V16)[above left=1cm and 0.27cm of V14]{\nodenum{ 16}};
\node (C16)[above left=5pt and 12pt of V16]{\colorgridfp};

\node[circle, black!75,  fill=black!25, draw]  (V17)[above right=0.65cm and 0.75cm of V16]{\nodenum{ 17}};
\node (C17)[above left=5pt and 12pt of V17]{\colorgridfy};

\node[circle, black!75,  fill=white, draw]  (V22)[below right=0.65cm and 0.75cm of V17]{\nodenum{ 22}};
\node (C22)[above right=5pt and 12pt of V22]{\colorgridfp};

%%\begin{scope}[on background layer]
%%\node[black!75, draw, fill=ForestGreen!10, inner sep=2.5pt, fit=(V14) (V15) (V16) (V17) (V22), dashed]{};
%%%%\end{scope}
%%
%%\draw[color=black!75, line width=0.75pt](V14) to (V8);
%%\draw[color=black!75, line width=0.75pt](V10) to (V17);

\draw[color=black!75, line width=0.75pt](V15) to (V22);
\draw[color=black!75, line width=0.75pt](V22) to (V14);
\draw[color=black!75, line width=0.75pt](V16) to (V14);
\draw[color=black!75, line width=0.75pt](V16) to (V22);
\draw[color=black!75, line width=0.75pt](V16) to (V17);
\draw[color=black!75, line width=0.75pt](V17) to (V15);
\draw[color=black!75, line width=0.75pt](V16) to (V15);
\draw[color=black!75, line width=0.75pt](V22) to (V17);
\draw[color=black!75, line width=0.75pt](V14) to (V15);

%%\draw[color=red, line width=1pt, dashed] (V14) to (V17); 

\end{tikzpicture}} \hspace{0.25cm} 
   \subcaptionbox{\footnotesize query lower bound of $\Delta$-coloring \label{fig:query}}%
  [.31\linewidth]{\definecolor {processblue}{cmyk}{0.96,0,0,0}
\tikzset{myptr/.style={decoration={markings,mark=at position 1 with %
    {\arrow[scale=3,>=stealth]{>}}},postaction={decorate}}}

\newcommand{\nodenum}[1]{}%{{\fontsize{1pt}{1pt}\selectfont #1}}

\begin {tikzpicture}[ auto ,node distance =1cm and 1cm , on grid , semithick, inner sep=3pt]

\node[circle, black!75,  fill=red!50, draw]  (V14){};%[below left=2.3cm and 0.25cm of V10]{\nodenum{ 14}};

\node[circle, black!75,  fill=yellow, draw]  (V15)[below right=0cm and 1cm of V14]{\nodenum{ 15}};

\node[circle, black!75,  fill=blue!50, draw]  (V16)[above left=1cm and 0.27cm of V14]{\nodenum{ 16}};

\node[circle, black!75,  fill=red!50, draw]  (V17)[above right=0.65cm and 0.75cm of V16]{\nodenum{ 17}};

\node[circle, black!75,  fill=ForestGreen!50, draw]  (V22)[below right=0.65cm and 0.75cm of V17]{\nodenum{ 22}};

\node[circle, black!75,  fill=red!50, draw]  (Z14)[left=2.5cm of V14]{};%[below left=2.3cm and 0.25cm of V10]{\nodenum{ 14}};

\node[circle, black!75,  fill=blue!50, draw]  (Z15)[below right=0cm and 1cm of Z14]{\nodenum{ 15}};

\node[circle, black!75,  fill=ForestGreen!50, draw]  (Z16)[above left=1cm and 0.27cm of Z14]{\nodenum{ 16}};

\node[circle, black!75,  fill=blue!50, draw]  (Z17)[above right=0.65cm and 0.75cm of Z16]{\nodenum{ 17}};

\node[circle, black!75,  fill=yellow, draw]  (Z22)[below right=0.65cm and 0.75cm of Z17]{\nodenum{ 22}};

%%\begin{scope}[on background layer]
%%\node[black!75, draw, fill=ForestGreen!10, inner sep=2.5pt, fit=(V14) (V15) (V16) (V17) (V22), dashed]{};
%%%%\end{scope}
%%
%%\draw[color=black!75, line width=0.75pt](V14) to (V8);
%%\draw[color=black!75, line width=0.75pt](V10) to (V17);

\draw[color=black!75, line width=0.75pt](V15) to (V22);
\draw[color=black!75, line width=0.75pt](V22) to (V14);
\draw[color=black!75, line width=0.75pt](V16) to (V14);
\draw[color=black!75, line width=0.75pt](V16) to (V22);
\draw[color=black!75, line width=0.75pt](V16) to (V17);
\draw[color=black!75, line width=0.75pt](V17) to (V15);
\draw[color=black!75, line width=0.75pt](V16) to (V15);
\draw[color=black!75, line width=0.75pt](V22) to (V17);
\draw[color=black!75, line width=0.75pt](V14) to (V15);

%%\draw[color=red, line width=1pt, dashed] (V14) to (V17); 

\draw[color=black!75, line width=0.75pt](Z15) to (Z22);
\draw[color=black!75, line width=0.75pt](Z22) to (Z14);
\draw[color=black!75, line width=0.75pt](Z16) to (Z14);
\draw[color=black!75, line width=0.75pt](Z17) to (Z14);
\draw[color=black!75, line width=0.75pt](Z16) to (Z22);
\draw[color=black!75, line width=0.75pt](Z16) to (Z17);
\draw[color=black!75, line width=0.75pt](Z16) to (Z15);
\draw[color=black!75, line width=0.75pt](Z22) to (Z17);
\draw[color=black!75, line width=0.75pt](Z14) to (Z15);

\draw[color=black!75, line width=0.75pt](V14) to (Z15);
\draw[color=black!75, line width=0.75pt](V17) to (Z17);
%\draw[color=black!75, line width=0.75pt](V15) to (V22);

\end{tikzpicture}}
\caption{ A graph with maximum degree $\Delta=4$ and its sparse-dense decomposition in \textbf{(a)} (each box denotes an almost-clique and remaining vertices are sparse). 
Part \textbf{(b)} is an illustration of why palette sparsification fails for $\Delta$-coloring: the only way to $\Delta$-color this graph is to color the marked vertices the same, which cannot be done with these sampled lists.
Part \textbf{(c)} shows a similar construction can be used to prove a query lower bound for $\Delta$-coloring. (The actual instance is obtained from $\Theta(n/\Delta)$ copies of such pairs.)} 
\label{fig:color}
\end{figure}

As expected, the hard part in extending palette sparsification theorem of~\cite{AssadiCK19a} to $\Delta$-coloring involves the argument for almost-cliques. Indeed, this is not just a matter of analysis; 
as already observed by~\cite{AssadiCK19a}, this theorem fails for $\Delta$-coloring (\Cref{fig:ps-fail}): consider a $(\Delta+1)$-clique minus
a single edge $(u,v)$; the only way we can find a $\Delta$-coloring of this graph is if we color both $u$ and $v$ the same, which requires their sampled lists to intersect; by the birthday paradox this only happens 
when size of each list is $\Omega(\sqrt{\Delta})$ which in turn implies that the  algorithm has to store $\Omega(n\Delta)$ edges from the stream---this is effectively the same as storing the input itself!

\paragraph{Sublinear time (query) algorithms.} Consider a graph $G$ which is a collection of $\Theta(n/\Delta)$ \emph{pairs} of $(\Delta+1)$-cliques. For each pair, randomly pick two vertices $(u_1,v_1)$ and $(u_2,v_2)$ from its first and second clique, respectively. Remove the edges $(u_1,v_1)$ and $(u_2,v_2)$ and instead include the edges $(u_1,v_2)$ and $(u_2,v_1)$ in the graph. See~\Cref{fig:query}  for an illustration. 
It is easy to see that the only way to $\Delta$-color this graph is to find the ``switched'' edges in each copy and color their endpoints the same inside each (now) almost-clique. Yet, it is an easy exercise 
to use the linear lower bound on the query complexity of $\OR$ function~\cite{BuhrmanW02} to prove that this requires making $\Omega(\Delta^2)$ queries to the adjacency lists or matrix of the graph for each pair, and thus $\Omega(n\Delta)$ queries overall. 
This lower bound now leaves us with the following lesson. 

\begin{lesson}\label{lesson1}
	Any semi-streaming algorithm for $\Delta$-coloring should explicitly \emph{look at} all but a tiny fraction of edges of the graph presented in the stream. 
\end{lesson}

\Cref{lesson1} may sound trivial at first. After all, the semi-streaming model allows all algorithms to look at all edges of the graph. Yet, note that 
numerous semi-streaming algorithms, say, sampling algorithms, 
including all prior streaming coloring algorithms in~\cite{AssadiCK19a,BeraCG19,AlonA20,HalldorssonKNT22},  
do not use this power---\Cref{lesson1} implies that these algorithms cannot solve $\Delta$-coloring.

\paragraph{Semi-streaming algorithms on repeated-edge streams.} What if we take~\Cref{lesson1} to the extreme and give the algorithm
each edge (potentially) multiple times in the stream---this should surely helps us even more, no?
It turns out however that this is not really the case. 

\begin{sidefigure}%{r}{0.45\textwidth}
\centering
\definecolor {processblue}{cmyk}{0.96,0,0,0}
\tikzset{myptr/.style={decoration={markings,mark=at position 1 with %
    {\arrow[scale=3,>=stealth]{>}}},postaction={decorate}}}

\newcommand{\nodenum}[1]{}%{{\fontsize{1pt}{1pt}\selectfont #1}}

\begin {tikzpicture}[ auto ,node distance =1cm and 1cm , on grid , semithick, inner sep=3pt]

\node[circle, black!75,  fill=white, draw]  (V14){};%[below left=2.3cm and 0.25cm of V10]{\nodenum{ 14}};

\node[circle, black!75,  fill=white, draw]  (V15)[below right=0cm and 1cm of V14]{\nodenum{ 15}};

\node[circle, black!75,  fill=white, draw]  (V16)[above left=1cm and 0.27cm of V14]{\nodenum{ 16}};

\node[circle, black!75,  fill=white, draw]  (V17)[above right=0.65cm and 0.75cm of V16]{\nodenum{ 17}};

\node[circle, black!75,  fill=white, draw]  (V22)[below right=0.65cm and 0.75cm of V17]{\nodenum{ 22}};

\node  (U)[above left=0.8cm and 0.75cm of V14]{\huge $+$};%[below left=2.3cm and 0.25cm of V10]{\nodenum{ 14}};

\node[circle, black!75,  fill=white, draw]  (Z14)[left=2.5cm of V14]{};%[below left=2.3cm and 0.25cm of V10]{\nodenum{ 14}};

\node[circle, black!75,  fill=white, draw]  (Z15)[below right=0cm and 1cm of Z14]{\nodenum{ 15}};

\node[circle, black!75,  fill=white, draw]  (Z16)[above left=1cm and 0.27cm of Z14]{\nodenum{ 16}};

\node[circle, black!75,  fill=white, draw]  (Z17)[above right=0.65cm and 0.75cm of Z16]{\nodenum{ 17}};

\node[circle, black!75,  fill=white, draw]  (Z22)[below right=0.65cm and 0.75cm of Z17]{\nodenum{ 22}};

\node  (N)[right=2.75cm of U]{\huge $=$};%[below left=2.3cm and 0.25cm of V10]{\nodenum{ 14}};

\node[circle, black!75,  fill=black!25, draw]  (W14)[ right=3cm of V14]{};%[below left=2.3cm and 0.25cm of V10]{\nodenum{ 14}};

\node[circle, black!75,  fill=white, draw]  (W15)[below right=0cm and 1cm of W14]{\nodenum{ 15}};

\node[circle, black!75,  fill=white, draw]  (W16)[above left=1cm and 0.27cm of W14]{\nodenum{ 16}};

\node[circle, black!75,  fill=black!25, draw]  (W17)[above right=0.65cm and 0.75cm of W16]{\nodenum{ 17}};

\node[circle, black!75,  fill=white, draw]  (W22)[below right=0.65cm and 0.75cm of W17]{\nodenum{ 22}};

%%\begin{scope}[on background layer]
%%\node[black!75, draw, fill=ForestGreen!10, inner sep=2.5pt, fit=(V14) (V15) (V16) (V17) (V22), dashed]{};
%%%%\end{scope}
%%
%%\draw[color=black!75, line width=0.75pt](V14) to (V8);
%%\draw[color=black!75, line width=0.75pt](V10) to (V17);

%\draw[color=black!75, line width=0.75pt](V15) to (V22);
\draw[color=red!75, line width=1.5pt]
(V22) to (V14)
(V16) to (V14)
(V16) to (V17)
(V17) to (V15)
(V22) to (V17)
(V14) to (V15);

%%\draw[color=red, line width=1pt, dashed] (V14) to (V17); 

\draw[color=blue!75, line width=1.5pt]
(Z15) to (Z22)
(Z16) to (Z22)
(Z16) to (Z17)
(Z16) to (Z15)
(Z22) to (Z17)
(Z14) to (Z15);

\draw[color=black!75, line width=1pt]
(W22) to (W14)
(W16) to (W14)
(W16) to (W17)
(W17) to (W15)
(W22) to (W17)
(W14) to (W15)
(W15) to (W22)
(W16) to (W22)
(W16) to (W15);

%%\draw[color=black!75, line width=0.75pt](V14) to (Z15);
%%\draw[color=black!75, line width=0.75pt](V17) to (Z17);
%\draw[color=black!75, line width=0.75pt](V15) to (V22);

\end{tikzpicture}
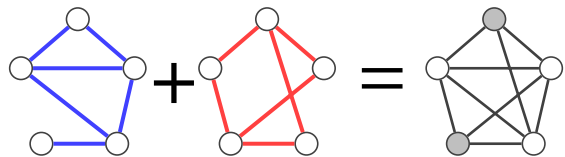
\caption{An illustration of the hard instances for the repeated-edge stream lower bound---the actual instance is obtained from $\Theta(n/\Delta)$ copies of these graphs. The only possible $\Delta$-coloring is to color both endpoints of marked
vertices the same.}
\label{fig:or-stream}
\end{sidefigure}

Suppose that we have a graph $G$ on a collection of $\Theta(n/\Delta)$ disjoint sets of vertices of size $\Delta+1$ each. For each set of $\Delta+1$ vertices $U$, consider a stream of edges that in the first part, provides a subset $E_1$ of 
edges over $U$ and in the second part, provides another subset $E_2$---the repeated-edge stream allows these subsets to be overlapping and we shall choose them so that $E_1 \cup E_2$ leaves 
precisely one pair of vertices $(u,v)$ among all pairs in $U$ without an edge. %See~\Cref{fig:or-stream} for an illustration. 
As before, the only way to $\Delta$-color this graph is to color vertices $u$, $v$ in each of the $\Theta(n/\Delta)$ pieces the same. 
But, given that the edges between $E_1$ and $E_2$ may overlap, one can prove that finding all these pairs requires $\Omega(n\Delta)$ space. This is by  
a reduction from communication complexity lower bounds of the \emph{Tribes} function~\cite{JayramKS03} (a slightly less well-known cousin of the famous set disjointness problem). 
This brings us to the next lesson. 

\begin{lesson}\label{lesson2}
	Any semi-streaming algorithm for $\Delta$-coloring should crucially use the fact that each edge of the graph arrives \emph{exactly once} in the stream. 
\end{lesson}

Again, while the semi-streaming model only allows for presenting each edge  once in the stream, many algorithms are entirely oblivious to this feature. This includes all previous semi-streaming 
coloring algorithms in~\cite{AssadiCK19a,BeraG18,BeraCG19,AlonA20,HalldorssonKNT22}, as well as various other ones for spanning trees~\cite{FeigenbaumKMSZ05}, sparsifiers~\cite{McGregor14}, spanners~\cite{FeigenbaumKMSZ05,FeigenbaumKMSZ08} and maximal matchings~\cite{FeigenbaumKMSZ05}. 
 \Cref{lesson2} says that any potential $\Delta$-coloring algorithm
cannot be of this type. 

\paragraph{A natural algorithm or a barrier result?} Finally, let us conclude this part by considering a natural semi-streaming algorithm for $\Delta$-coloring: Sample $O(n\log{n}/\Delta)$ vertices $S$ uniformly at random, and partition
the input graph into two subgraphs $G_S$ consisting of all edges incident on $S$, and $G_{-S}$ consisting of all remaining edges. We can  easily detect, for each arriving edge in the stream, which subgraph it belongs to. Moreover, it is easy to see that $G_S$ contains $O(n\log{n})$ edges and $G_{-S}$, with high probability, has maximum degree at most $\Delta-1$. We can thus store $G_S$ explicitly via a semi-streaming algorithm and run the algorithm of~\cite{AssadiCK19a} on $G_{-S}$ 
to color it with $(\Delta(G_{-S})+1) = \Delta$ colors. So, we have a $\Delta$-coloring of $G_{-S}$ and all edges of $G_S$. 

Surely, now that we know \emph{all}  of~$G_S$, we should be able to extend the $\Delta$-coloring of $G_{-S}$ to~$G_S$, no? 
The answer however turns out to be \emph{no}: unlike the case of $(\Delta+1)$-coloring, not every partial coloring of a graph can be extended directly to a proper $\Delta$-coloring of the entire graph.  
But perhaps this is only an abstract worry and we should just find the right way of analyzing this algorithm? The answer is yet again \emph{no}: the algorithm we just proposed in fact neglects both~\Cref{lesson1} and~\Cref{lesson2} and thus is doomed to fail completely.%
\footnote{An added bonus of those impossibility results is to allow for quickly checking viability of potential algorithms.}
But this also leaves us with the following lesson. 

\begin{lesson}\label{lesson3}
	Any semi-streaming algorithm for $\Delta$-coloring that colors the graph by extending a partial coloring, subgraph by subgraph, should either provide a \emph{stronger guarantee} than solely an arbitrary $\Delta$-coloring for each subgraph, 
	or allow for \emph{recoloring} of an already colored subgraph.  
\end{lesson}

While perhaps less concrete than our two previous lessons,~\Cref{lesson3} has a profound impact in the design of our semi-streaming algorithm that also colors the graph one subgraph at a time; 
in particular, our algorithm is going to adhere to \emph{both} restrictions imposed by~\Cref{lesson3} simultaneously. 

\subsection{The High-Level Overview of Our Algorithm}\label{sec:tech-alg} 

After this long detour, we are now ready to go over our algorithm in~\Cref{res:main}. As stated earlier, the three main ingredients of our algorithm are: $(i)$ a variant of palette sparsification for $\Delta$-coloring that can color all but some
{problematic} almost-cliques in the input (such as~\Cref{fig:ps-fail}), $(ii)$ a sparse recovery approach for (partially) recovering these problematic subgraphs, and $(iii)$ a new coloring procedure that allows for extending the partial coloring of part $(i)$ 
to the remaining subgraphs partially recovered in part $(ii)$ to obtain a proper $\Delta$-coloring of the entire graph. We will go over each part separately in the following.

\subsubsection*{Part One: Powers and Limitations of Palette Sparsification for $\Delta$-Coloring} 

While we already discussed in~\Cref{sec:tech-lower} that palette sparsification fails for $\Delta$-coloring, we are still going to employ its ideas crucially in our work. 
The goal of this step is to identify to what extent this approach fails for $\Delta$-coloring. 
We  develop a \textbf{classification of almost-cliques} (\Cref{sec:classification}) based on the following three criteria: 

\begin{itemize}%[leftmargin=15pt]
	\item \textbf{Size}---number of vertices: \emph{small} for $< \Delta+1$, \emph{critical} for $\Delta+1$, and \emph{large} for $> \Delta+1$ vertices. 
	\item \textbf{Inner density}---number of non-edges inside the almost-clique: \emph{holey} for $\Omega(\Delta)$ non-edges (or ``holes'' in the almost-clique), and \emph{unholey} otherwise. 
	\item \textbf{Outer connections}---a measure of how ``tightly'' the almost-clique is connected to outside; we postpone the technical details of this part to the actual definition in~\Cref{sec:classification}. 
\end{itemize}

Among these, size and inner density are perhaps usual suspects. For instance, we already saw in~\Cref{sec:tech-lower} that palette sparsification entirely fails for almost-cliques of size $\Delta+1$ with exactly one non-edge---in our classification, 
these correspond to critical unholey almost-cliques. The third criterion is more technical and is motivated by~\Cref{lesson3};
as we are still going to color the graph one almost-clique at a time, we would like 
to be able to reason about the partial coloring of outside vertices and possibility of its extension to the almost-clique.
This is particularly relevant for small almost-cliques which can be actually a true clique inside and hence would definitely be in trouble if the same exact set of colors 
is ``blocked'' for all their vertices from outside. See~\Cref{fig:outer}. 

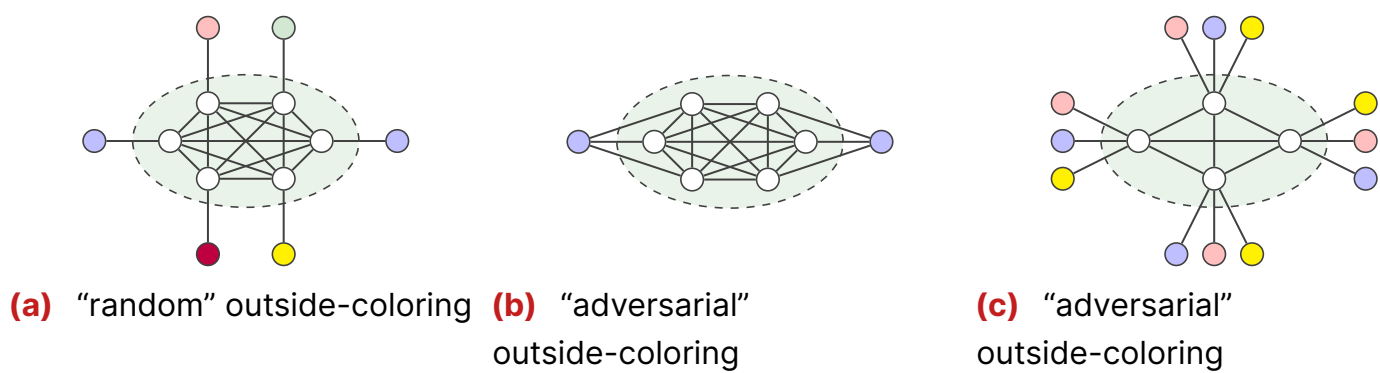
\begin{figure}[h!]
\centering
\subcaptionbox{\footnotesize ``random'' outside-coloring \label{fig:out1}}%
  [.31\linewidth]{\definecolor {processblue}{cmyk}{0.96,0,0,0}
\tikzset{myptr/.style={decoration={markings,mark=at position 1 with %
    {\arrow[scale=3,>=stealth]{>}}},postaction={decorate}}}

\newcommand{\colorgrid}{}

\newcommand{\nodenum}[1]{}%{{\fontsize{1pt}{1pt}\selectfont #1}}

\begin {tikzpicture}[ auto ,node distance =1cm and 1cm , on grid , semithick, inner sep=3pt]

%%

%\draw [help lines, step=0.5cm] (0,0) grid (5,5);

\node[ellipse, black!75, fill=ForestGreen!10, draw, dashed, minimum width=85pt, minimum height=50pt] (K) at (2,1) {};

\node[circle, black!75, fill=white, draw] (V1) at (1,1) {};
\node[circle, black!75, fill=blue!25, draw] (Z1) at (0,1) {};
\node[circle, black!75, fill=white, draw] (V2) at (1.5,1.5) {};
\node[circle, black!75, fill=red!25, draw] (Z2) at (1.5,2.5) {};
\node[circle, black!75, fill=white, draw] (V3) at (2.5,1.5) {};
\node[circle, black!75, fill=ForestGreen!20, draw] (Z3) at (2.5,2.5) {};
\node[circle, black!75, fill=white, draw] (V4) at (3,1) {};
\node[circle, black!75, fill=blue!25, draw] (Z4) at (4,1) {};
\node[circle, black!75, fill=white, draw] (V5) at (2.5,0.5) {};
\node[circle, black!75, fill=yellow, draw] (Z5) at (2.5,-0.5) {};
\node[circle, black!75, fill=white, draw] (V6) at (1.5,0.5) {};
\node[circle, black!75, fill=purple, draw] (Z6) at (1.5,-0.5) {};

\draw[color=black!75, line width=0.75pt]
(V1) to (Z1)
(V2) to (Z2)
(V3) to (Z3)
(V4) to (Z4)
(V5) to (Z5)
(V6) to (Z6); 

\foreach \i in {1,...,5}
{
  \pgfmathtruncatemacro{\ip}{(\i + 1)};
\foreach \j in {\ip,...,6}
{
\draw[color=black!75, line width=0.75pt]
(V\i) to (V\j);	
}
}

\end{tikzpicture}}
  \subcaptionbox{\footnotesize ``adversarial'' outside-coloring \label{fig:out2}}%
  [.31\linewidth]{\definecolor {processblue}{cmyk}{0.96,0,0,0}
\tikzset{myptr/.style={decoration={markings,mark=at position 1 with %
    {\arrow[scale=3,>=stealth]{>}}},postaction={decorate}}}

\newcommand{\colorgrid}{}

\newcommand{\nodenum}[1]{}%{{\fontsize{1pt}{1pt}\selectfont #1}}

\begin {tikzpicture}[ auto ,node distance =1cm and 1cm , on grid , semithick, inner sep=3pt]

%%

%\draw [help lines, step=0.5cm] (0,0) grid (5,5);

\node[ellipse, black!75, fill=ForestGreen!10, draw, dashed, minimum width=85pt, minimum height=50pt] (K) at (2,1) {};

\node[circle, black!75, fill=white, draw] (V1) at (1,1) {};
\node[circle, black!75, fill=blue!25, draw] (Z1) at (0,1) {};
\node[circle, black!75, fill=white, draw] (V2) at (1.5,1.5) {};
\node[circle, black!75, fill=white] (Z2) at (1.5,2.5) {};
\node[circle, black!75, fill=white, draw] (V3) at (2.5,1.5) {};
\node[circle, black!75, fill=white] (Z3) at (2.5,2.5) {};
\node[circle, black!75, fill=white, draw] (V4) at (3,1) {};
\node[circle, black!75, fill=blue!25, draw] (Z4) at (4,1) {};
\node[circle, black!75, fill=white, draw] (V5) at (2.5,0.5) {};
\node[circle, black!75, fill=white] (Z5) at (2.5,-0.5) {};
\node[circle, black!75, fill=white, draw] (V6) at (1.5,0.5) {};
\node[circle, black!75, fill=white] (Z6) at (1.5,-0.5) {};

\draw[color=black!75, line width=0.75pt]
(V1) to (Z1)
(V2) to (Z1)
(V6) to (Z1)
(V4) to (Z4)
(V5) to (Z4)
(V3) to (Z4); 

\foreach \i in {1,...,5}
{
  \pgfmathtruncatemacro{\ip}{(\i + 1)};
\foreach \j in {\ip,...,6}
{
\draw[color=black!75, line width=0.75pt]
(V\i) to (V\j);	
}
}

\end{tikzpicture}}
  \subcaptionbox{\footnotesize ``adversarial'' outside-coloring \label{fig:out3}}%
  [.31\linewidth]{\definecolor {processblue}{cmyk}{0.96,0,0,0}
\tikzset{myptr/.style={decoration={markings,mark=at position 1 with %
    {\arrow[scale=3,>=stealth]{>}}},postaction={decorate}}}

\newcommand{\colorgrid}{}

\newcommand{\nodenum}[1]{}%{{\fontsize{1pt}{1pt}\selectfont #1}}

\begin {tikzpicture}[ auto ,node distance =1cm and 1cm , on grid , semithick, inner sep=3pt]

%%

%\draw [help lines, step=0.5cm] (0,0) grid (5,5);

\node[ellipse, black!75, fill=ForestGreen!10, draw, dashed, minimum width=85pt, minimum height=50pt] (K) at (2,1) {};

\node[circle, black!75, fill=white, draw] (V1) at (1,1) {};
\node[circle, black!75, fill=blue!25, draw] (Z1) at (0,1) {};
\node[circle, black!75, fill=yellow, draw] (Z2) at (0,0.5) {};
\node[circle, black!75, fill=red!25, draw] (Z3) at (0,1.5) {};
\node[circle, black!75, fill=white, draw] (V2) at (2,1.5) {};
\node[circle, black!75, fill=red!25, draw] (Z4) at (1.5,2.5) {};
\node[circle, black!75, fill=blue!25, draw] (Z5) at (2,2.5) {};
\node[circle, black!75, fill=yellow, draw] (Z6) at (2.5,2.5) {};
\node[circle, black!75, fill=white, draw] (V3) at (3,1) {};
\node[circle, black!75, fill=yellow, draw] (Z7) at (4,1.5) {};
\node[circle, black!75, fill=red!25, draw] (Z8) at (4,1) {};
\node[circle, black!75, fill=blue!25, draw] (Z9) at (4,0.5) {};
\node[circle, black!75, fill=white, draw] (V4) at (2,0.5) {};
\node[circle, black!75, fill=red!25, draw] (Z10) at (2,-0.5) {};
\node[circle, black!75, fill=yellow, draw] (Z11) at (2.5,-0.5) {};
\node[circle, black!75, fill=blue!25, draw] (Z12) at (1.5,-0.5) {};

\draw[color=black!75, line width=0.75pt]
(V1) to (Z1)
(V1) to (Z2)
(V1) to (Z3)
(V2) to (Z4)
(V2) to (Z5)
(V2) to (Z6)
(V3) to (Z7)
(V3) to (Z8)
(V3) to (Z9)
(V4) to (Z10)
(V4) to (Z11)
(V4) to (Z12); 
\foreach \i in {1,...,3}
{
  \pgfmathtruncatemacro{\ip}{(\i + 1)};
\foreach \j in {\ip,...,4}
{
\draw[color=black!75, line width=0.75pt]
(V\i) to (V\j);	
}
}

\end{tikzpicture}}
\caption{An illustration of three possible types of outer connections on a graph with maximum degree $\Delta=6$. The almost-clique in part \textbf{(a)} has a ``right'' type of 
outside connection and is going to receive a more ``random'' coloring on its neighbors, 
compared to the almost-clique in part \textbf{(b)} with ``few'' outside neighbors and part \textbf{(c)} with ``too many'' ones. In particular, the latter almost-cliques now cannot be $\Delta$-colored without changing the color of outside vertices as the same colors are blocked for all vertices of the inner (actual) cliques.  
}
\label{fig:outer}
\end{figure}

We then consider palette sparsification (on steroids!) wherein each vertex samples $\poly\!\log{(n)}$ colors from $\set{1,\ldots,\Delta}$ and  \emph{characterize} which families of almost-cliques in our classification can still be colored using only the sampled colors. We show that all holey almost-cliques (regardless of their size or outer connections) can be still colored from their sampled colors using a similar argument as in~\cite{AssadiCK19a}. More interestingly, we show that 
even unholey small almost-cliques that have the ``right'' type of outside connections can be colored at this step. 
Our analysis in this part deviates significantly from~\cite{AssadiCK19a} and in particular crucially establishes certain \emph{randomness} properties on the coloring of vertices \emph{outside} of an almost-clique when trying to color the almost-clique
itself (recall~\Cref{lesson3}). Thus, what remains are unholey critical almost-cliques (regardless of their outer connections) and unholey small almost-cliques with ``problematic'' outside connections. 
We delegate coloring of these almost-cliques to the next steps of the algorithm. 

Let us now briefly discuss the effect of outer connections  in coloring a small almost-clique. As stated earlier, the main problem with small almost-cliques occurs when exactly the same  set of colors is used 
to color all outside vertices, thus blocking these colors entirely for the almost-clique. While this event is basically unavoidable for almost-cliques with only a couple of outside neighbors (\Cref{fig:out2}), 
it becomes less and less likely as the number of outside neighbors increases (\Cref{fig:out1}). After all, for a color to be used on all these vertices, 
it should be sampled by every single one of them in the first place. We will however run into problem again in cases when we have ``too many'' outside neighbors for every single vertex of the almost-clique (\Cref{fig:out3}). 
The almost-cliques of~\Cref{fig:out3} are particularly problematic as we have no knowledge of the neighborhood of their outside vertices (for~\Cref{fig:out2}, we at least know that each of them have many neighbors in the almost-clique, which is used
crucially by our latter algorithms). Thus, we should basically avoid ending up in a situation that we have to color such almost-cliques after having colored their outside neighbors. 

Fortunately, the almost-cliques of~\Cref{fig:out3} can only happen for small enough almost-cliques; this in turn makes the neighborhood of these almost-cliques sufficiently sparse. Thus, 
we can instead handle them similar to sparse vertices by \emph{increasing} the size of sampled palettes for vertices. There is however a subtle issue with this approach.
Increasing the size of sampled palettes means that even a fewer number of outside vertices can make a problem for us, hence requiring us to send even more almost-cliques to sparse vertices to handle, leading to a chicken-and-egg problem.
A key idea in this part is a way to break this dependency cycle by careful sequencing the order of processing of vertices in a way that ensures sufficient ``randomness'' exist in the coloring of neighborhood of 
all small almost-cliques, except for the ones with very few outside neighbors (the type in \Cref{fig:out2}). We postpone the discussion of this ``dependency-breaking'' step
to~\Cref{sec:coloring} and~\Cref{rem:why-interleave}. 

All in all, this step  effectively establishes that palette sparsification achieves a ``weak'' streaming Brooks' theorem by $\Delta$-coloring graphs that do not contain certain \emph{forbidden} subgraphs 
such as $(\Delta+1)$-cliques minus few edges or $\Delta$-cliques that have few neighbors outside (for Brooks' theorem itself, the only forbidden subgraph is a $(\Delta+1)$-clique).

\subsubsection*{Part Two: Sparse Recovery for Remaining Almost-Cliques, and Helper Structures} 

Our next step is a way of \emph{finding edges of} the almost-cliques left uncolored by the previous step so that we can color them using a different approach. Let us bring up an obvious point here: these left out 
almost-cliques are precisely the same family of instances that were at the core of our impossibility results in~\Cref{sec:tech-lower} (and their natural relaxations). Consequently, to handle them, 
our algorithm needs to take into account the recipe put forward by~\Cref{lesson1} and~\Cref{lesson2}: it should look at \emph{all} edges of the stream \emph{exactly once}. 
This rather uniquely points us toward a canonical technique in the streaming model: \textbf{sparse recovery} (via linear sketching). 

\begin{figure}[h!]
\centering
\subcaptionbox{\footnotesize neighbors-vector of $v=3$ \label{fig:sr1}}%
  [.31\linewidth]{\definecolor {processblue}{cmyk}{0.96,0,0,0}
\tikzset{myptr/.style={decoration={markings,mark=at position 1 with %
    {\arrow[scale=3,>=stealth]{>}}},postaction={decorate}}}

\newcommand{\colorgrid}{}

\newcommand{\nodenum}[1]{{\fontsize{1pt}{1pt}\selectfont #1}}

\begin {tikzpicture}[ auto ,node distance =1cm and 1cm , on grid , semithick, inner sep=3pt]

\node[circle, black!75, fill=white!20, draw]  (V8){\nodenum{1}};

\node[circle, black!75, fill=white!20, draw]   (V10)[above right=2.1cm and 1.5cm of V8]{\nodenum{2}};

\node[circle, black!75,  fill=black!25, draw]  (V14)[below left=2.3cm and 0.25cm of V10]{\nodenum{3}};

\node (n) [below right=0.75cm and 0.3cm of V14]{{\scriptsize $\vec{N}(3) := [\textcolor{blue}{1},0,0,\textcolor{blue}{1},\textcolor{blue}{1},0,\textcolor{blue}{1}]$}};

\node[circle, black!75,  fill=white, draw]  (V15)[below right=0cm and 1cm of V14]{\nodenum{4}};

\node[circle, black!75,  fill=white, draw]  (V16)[above left=1cm and 0.27cm of V14]{\nodenum{5}};

\node (nn) [above left=1.75cm and 0.3cm of V14]{{\scriptsize $K$}};

\node[circle, black!75,  fill=white!25, draw]  (V17)[above right=0.65cm and 0.75cm of V16]{\nodenum{6}};

\node[circle, black!75,  fill=white, draw]  (V22)[below right=0.65cm and 0.75cm of V17]{\nodenum{7}};

\begin{scope}[on background layer]
\node[black!75, draw, fill=ForestGreen!10, inner sep=2.5pt, fit=(V14) (V15) (V16) (V17) (V22), dashed]{};
\end{scope}

\draw[color=blue!75, line width=1pt](V14) to (V8);
\draw[color=black!75, line width=0.75pt](V10) to (V17);

\draw[color=black!75, line width=0.75pt](V15) to (V22);
\draw[color=blue!75, line width=1pt](V22) to (V14);
\draw[color=blue!75, line width=1pt](V16) to (V14);
\draw[color=black!75, line width=0.75pt](V16) to (V22);
\draw[color=black!75, line width=0.75pt](V16) to (V17);
\draw[color=black!75, line width=0.75pt](V17) to (V15);
\draw[color=black!75, line width=0.75pt](V16) to (V15);
\draw[color=black!75, line width=0.75pt](V22) to (V17);
\draw[color=blue!75, line width=1pt](V14) to (V15);

%\draw[color=red, line width=1pt, dashed] (V14) to (V17); 

\end{tikzpicture}}
  \subcaptionbox{\footnotesize characteristic-vector of $K$  \label{fig:sr2}}%
  [.31\linewidth]{\definecolor {processblue}{cmyk}{0.96,0,0,0}
\tikzset{myptr/.style={decoration={markings,mark=at position 1 with %
    {\arrow[scale=3,>=stealth]{>}}},postaction={decorate}}}

\newcommand{\colorgrid}{}

\newcommand{\nodenum}[1]{{\fontsize{1pt}{1pt}\selectfont #1}}

\begin {tikzpicture}[ auto ,node distance =1cm and 1cm , on grid , semithick, inner sep=3pt]

\node[circle, black!75, fill=white!25, draw]  (V8){\nodenum{1}};

\node[circle, black!75, fill=white!25, draw]   (V10)[above right=2.1cm and 1.5cm of V8]{\nodenum{2}};

\node[circle, black!75,  fill=black!25, draw]  (V14)[below left=2.3cm and 0.25cm of V10]{\nodenum{3}};

\node (n) [below right=0.75cm and 0.5cm of V14]{{\scriptsize ${\chi(K)} := [0,0,1,1,1,1,1]$}};

\node[circle, black!75,  fill=black!25, draw]  (V15)[below right=0cm and 1cm of V14]{\nodenum{4}};

\node[circle, black!75,  fill=black!25, draw]  (V16)[above left=1cm and 0.27cm of V14]{\nodenum{5}};
\node[circle, black!75,  fill=black!25, draw]  (V17)[above right=0.65cm and 0.75cm of V16]{\nodenum{6}};

\node (nn) [above left=1.75cm and 0.3cm of V14]{{\scriptsize $K$}};

\node[circle, black!75,  fill=black!25, draw]  (V22)[below right=0.65cm and 0.75cm of V17]{\nodenum{7}};

\begin{scope}[on background layer]
\node[black!75, draw, fill=ForestGreen!10, inner sep=2.5pt, fit=(V14) (V15) (V16) (V17) (V22), dashed]{};
\end{scope}

%\draw[color=black!75, line width=0.75pt](V14) to (V8);
%\draw[color=black!75, line width=0.75pt](V10) to (V17);

%%\draw[color=black!75, line width=0.75pt](V15) to (V22);
%%\draw[color=black!75, line width=0.75pt](V22) to (V14);
%%\draw[color=black!75, line width=0.75pt](V16) to (V14);
%%\draw[color=black!75, line width=0.75pt](V16) to (V22);
%%\draw[color=black!75, line width=0.75pt](V16) to (V17);
%%\draw[color=black!75, line width=0.75pt](V17) to (V15);
%%\draw[color=black!75, line width=0.75pt](V16) to (V15);
%%\draw[color=black!75, line width=0.75pt](V22) to (V17);
%%\draw[color=black!75, line width=0.75pt](V14) to (V15);

%\draw[color=red, line width=1pt, dashed] (V14) to (V17); 

\end{tikzpicture}}
  \subcaptionbox{\footnotesize  recovery-vector for neighbors of 3 \label{fig:sr3}}%
  [.32\linewidth]{\definecolor {processblue}{cmyk}{0.96,0,0,0}
\tikzset{myptr/.style={decoration={markings,mark=at position 1 with %
    {\arrow[scale=3,>=stealth]{>}}},postaction={decorate}}}

\newcommand{\colorgrid}{}

\newcommand{\nodenum}[1]{{\fontsize{1pt}{1pt}\selectfont #1}}

\begin {tikzpicture}[ auto ,node distance =1cm and 1cm , on grid , semithick, inner sep=3pt]

\node[circle, black!75, fill=white!25, draw]  (V8){\nodenum{1}};

\node[circle, black!75, fill=white!25, draw]   (V10)[above right=2.1cm and 1.5cm of V8]{\nodenum{2}};

\node[circle, black!75,  fill=black!25, draw]  (V14)[below left=2.3cm and 0.25cm of V10]{\nodenum{3}};

\node (n) [below right=0.75cm and 0.5cm of V14]{{\scriptsize $\vec{N}(3)-{\chi(K)} := [\textcolor{blue}{1},0,0,0,0,\textcolor{red}{-1},0]$}};

\node[circle, black!75,  fill=white!25, draw]  (V15)[below right=0cm and 1cm of V14]{\nodenum{4}};

\node[circle, black!75,  fill=white!25, draw]  (V16)[above left=1cm and 0.27cm of V14]{\nodenum{5}};
\node[circle, black!75,  fill=white!25, draw]  (V17)[above right=0.65cm and 0.75cm of V16]{\nodenum{6}};

\node (nn) [above left=1.75cm and 0.3cm of V14]{{\scriptsize $K$}};

\node[circle, black!75,  fill=white!25, draw]  (V22)[below right=0.65cm and 0.75cm of V17]{\nodenum{7}};

\begin{scope}[on background layer]
\node[black!75, draw, fill=ForestGreen!10, inner sep=2.5pt, fit=(V14) (V15) (V16) (V17) (V22), dashed]{};
\end{scope}

\draw[color=blue!75, line width=1pt](V14) to (V8);
\draw[color=red, line width=1pt, dashed] (V14) to (V17);

\end{tikzpicture}}
\caption{ 
An illustration of sparse-recovery on the neighborhood of each vertex, plus an algorithm that finds the identity of vertices in each almost-clique, allows for recovering all edges  ``highly-dense'' almost-cliques. Our actual algorithm is considerably more involved as it needs to \emph{partially} recover ``not-too-dense'' almost-cliques  also. 
}
\label{fig:sr}
\end{figure}
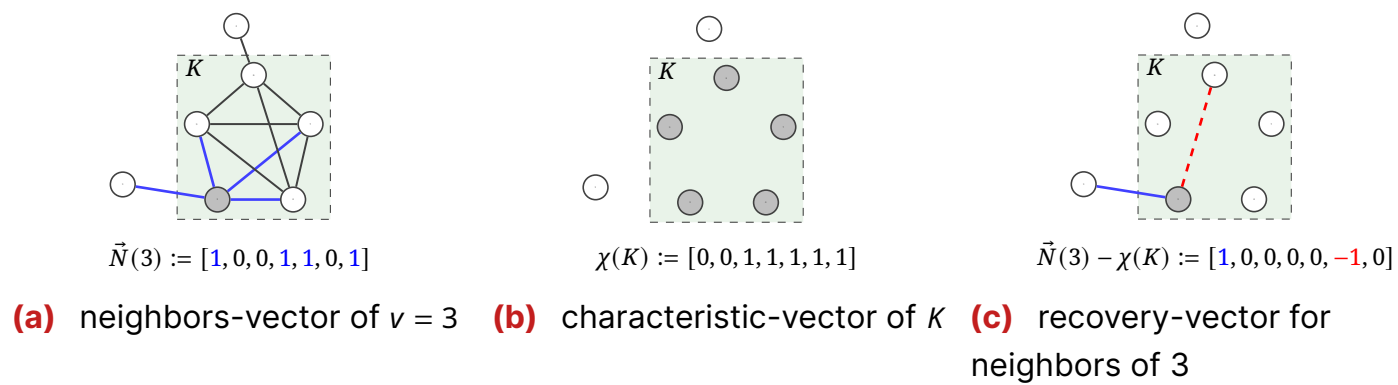

Consider an almost-clique which is a $(\Delta+1)$-clique minus an edge (\Cref{fig:ps-fail}). On the surface, recovering all edges of this almost-clique is problematic as these subgraphs are actually quite \emph{dense}; for instance, if the 
graph consists of only copies of such almost-cliques, we will need $\Omega(n\Delta)$ space to store all of them. But what saves us at this stage is  the fact that these subgraphs are actually {too dense}! 
Informally speaking, this reduces their ``entropy'' dramatically {conditioned} on our knowledge of the sparse-dense decomposition. Thus, we can 
recover them \emph{implicitly} using a novel {sparse recovery approach} that uses sparse-dense decompositions {algorithmically} and not only analytically.%
\footnote{For the main results in~\cite{AssadiCK19a} 
one only needs to know the \emph{existence} of the decomposition and does not need to compute it. That being said,~\cite{AssadiCK19a} also gave algorithms for finding the decomposition from the stream (which is needed
to run their algorithms in polynomial time)---we use an extension of their algorithm by~\cite{AssadiC22} in this paper.}
%See~\Cref{fig:sr} for an illustration. 

Although in the examples we discussed, we can  hope to recover the entire almost-clique in question implicitly, this will not be the case for all almost-cliques left 
uncolored by the first part, e.g., for a $(\Delta+1)$-clique minus a $\sqrt{\Delta}$-size inner clique (applying this method to a graph consisting of only such almost-cliques requires $\Omega(n\sqrt{\Delta})$ space).  
 As a result, our semi-streaming algorithm  settles for recovering certain \textbf{helper structures} from these almost-cliques instead. These are subgraphs of the input that are sufficiently simple to be recoverable via a combination 
of sparse recovery, sampling, and some basic graph theory arguments. At the same time, they are structured enough to give us enough flexibility for the final coloring step.  Given the technicality of their definitions (\Cref{def:helper-friendly,def:helper-critical} and~\Cref{fig:helper}), we postpone further details  to~\Cref{sec:classification}. 

We  point out that sparse recovery and linear sketching have been a staple of graph streaming algorithms since the seminal work of~\cite{AhnGM12}. But, these tools have been almost exclusively used to handle 
edge deletions in \emph{dynamic} streams. Their applications for us, on \emph{insertion-only} streams, as a way of (implicitly) sparsifying a graph using outside information (i.e., the sparse-dense decomposition), 
is quite different and can form a tool of independent interest.  

\subsubsection*{Part Three: The Final Coloring Procedure} 

The final step of our approach is then to color these remaining almost-cliques, given the extra information we recovered for them in the previous step. For intuition, let us consider two inherently different 
types of almost-cliques left uncolored by the approach in the first part.  

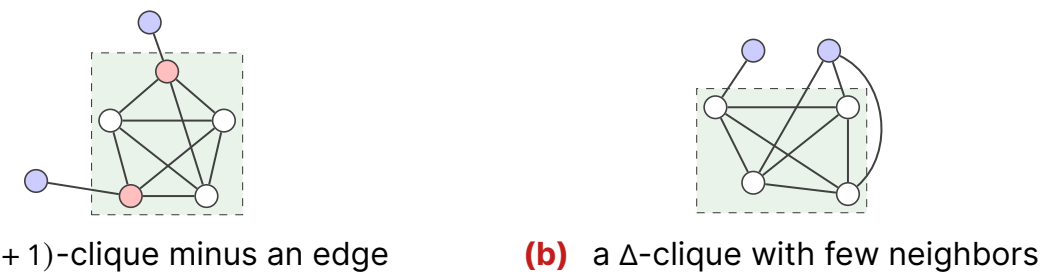
\begin{figure}[h!]
\centering
  \subcaptionbox{\footnotesize a $(\Delta+1)$-clique minus an edge \label{fig:easy}}%
  [.42\linewidth]{\definecolor {processblue}{cmyk}{0.96,0,0,0}
\tikzset{myptr/.style={decoration={markings,mark=at position 1 with %
    {\arrow[scale=3,>=stealth]{>}}},postaction={decorate}}}

\newcommand{\colorgrid}{}

\newcommand{\nodenum}[1]{}%{{\fontsize{1pt}{1pt}\selectfont #1}}

\begin {tikzpicture}[ auto ,node distance =1cm and 1cm , on grid , semithick, inner sep=3pt]

\node[circle, black!75, fill=blue!20, draw]  (V8){\nodenum{ 8}};

\node[circle, black!75, fill=blue!20, draw]   (V10)[above right=2.1cm and 1.5cm of V8]{\nodenum{ 10}};

\node[circle, black!75,  fill=red!25, draw]  (V14)[below left=2.3cm and 0.25cm of V10]{\nodenum{ 14}};

\node[circle, black!75,  fill=white, draw]  (V15)[below right=0cm and 1cm of V14]{\nodenum{ 15}};

\node[circle, black!75,  fill=white, draw]  (V16)[above left=1cm and 0.27cm of V14]{\nodenum{ 16}};
\node[circle, black!75,  fill=red!25, draw]  (V17)[above right=0.65cm and 0.75cm of V16]{\nodenum{ 17}};

\node[circle, black!75,  fill=white, draw]  (V22)[below right=0.65cm and 0.75cm of V17]{\nodenum{ 22}};

\begin{scope}[on background layer]
\node[black!75, draw, fill=ForestGreen!10, inner sep=2.5pt, fit=(V14) (V15) (V16) (V17) (V22), dashed]{};
\end{scope}

\draw[color=black!75, line width=0.75pt](V14) to (V8);
\draw[color=black!75, line width=0.75pt](V10) to (V17);

\draw[color=black!75, line width=0.75pt](V15) to (V22);
\draw[color=black!75, line width=0.75pt](V22) to (V14);
\draw[color=black!75, line width=0.75pt](V16) to (V14);
\draw[color=black!75, line width=0.75pt](V16) to (V22);
\draw[color=black!75, line width=0.75pt](V16) to (V17);
\draw[color=black!75, line width=0.75pt](V17) to (V15);
\draw[color=black!75, line width=0.75pt](V16) to (V15);
\draw[color=black!75, line width=0.75pt](V22) to (V17);
\draw[color=black!75, line width=0.75pt](V14) to (V15);

%%\draw[color=red, line width=1pt, dashed] (V14) to (V17); 

\end{tikzpicture}}
  \subcaptionbox{\footnotesize  a $\Delta$-clique with few neighbors  \label{fig:hard}}%
  [.42\linewidth]{\definecolor {processblue}{cmyk}{0.96,0,0,0}
\tikzset{myptr/.style={decoration={markings,mark=at position 1 with %
    {\arrow[scale=3,>=stealth]{>}}},postaction={decorate}}}

\newcommand{\colorgrid}{}

\newcommand{\nodenum}[1]{}%{{\fontsize{1pt}{1pt}\selectfont #1}}

\begin {tikzpicture}[ auto ,node distance =1cm and 1cm , on grid , semithick, inner sep=3pt]

\node[circle, black!75, fill=blue!20, draw]    (V3){\nodenum{ 3}};
\node (C3)[above left=5pt and 12pt of V3]{\colorgrid};

\node[circle, black!75, fill=blue!20, draw]    (V4)[right=of V3]{\nodenum{ 4}};
\node (C4)[above right=14pt and 8pt of V4]{\colorgrid};

\node[circle, black!75,   fill=white, draw]  (V5)[below left=0.75cm and 0.5cm of V3]{\nodenum{ 5}};
\node (C5)[above left=5pt and 12pt of V5]{\colorgrid};

\node[circle, black!75,  fill=white, draw]  (V6)[right=1.75cm of V5]{\nodenum{ 6}};
\node (C6)[above right=5pt and 12pt of V6]{\colorgrid};

\node[circle, black!75,  fill=white, draw]  (V7)[below right=1cm and 0.5cm of V5]{\nodenum{ 7}};
\node (C7)[below left=5pt and 12pt of V7]{\colorgrid};

\node[circle, black!75,  fill=white, draw]  (V8)[below right=0.15cm and 1.25cm of V7]{\nodenum{ 8}};
\node (C8)[below right=5pt and 12pt of V8]{\colorgrid};

%%\node[circle, black!75, fill=blue!20, draw]  (V14)[below right=0.2cm and 1.25cm of V8]{\nodenum{ 14}};
%%\node (C14)[above left=5pt and 12pt of V14]{\colorgrid};

%%\node[circle, blue, fill=white, draw,  line width=1pt]   (V18)[below right=0.75cm and 0.25 cm of V15]{\nodenum{ 18}};
%%\node (C18)[above right=5pt and 12pt of V18]{\colorgrid};
%%
%%\node[circle, blue, fill=white, draw,  line width=1pt]   (V19)[below right=1.75cm and 0.25 cm of V18]{\nodenum{ 19}};
%%\node (C19)[above right=5pt and 12pt of V19]{\colorgrid};
%%
%%\node[circle, blue, fill=white, draw,  line width=1pt]    (V20)[above left=0.25cm and 1.25 cm of V19]{\nodenum{ 20}};
%%\node (C20)[below left=5pt and 12pt of V20]{\colorgrid};
%%
%%\node[circle, blue, fill=white, draw,  line width=1pt]   (V21)[below right=0.5cm and 0.5 cm of V20]{\nodenum{ 21}};
%%\node (C21)[below left=5pt and 12pt of V21]{\colorgrid};
%%

\begin{scope}[on background layer]
\node[black!75, draw, fill=ForestGreen!10, inner sep=2.5pt, fit=(V5) (V6) (V7) (V8), dashed]{};
\end{scope}

\draw[color=black!75, line width=0.75pt](V4) to (V7);
\draw[color=black!75, line width=0.75pt](V4) to (V6);
%\draw[color=black!75, line width=0.75pt](V3) to (V4);
\draw[color=black!75, line width=0.75pt](V3) to (V5);
\draw[color=black!75, line width=0.75pt](V5) to (V7);
\draw[color=black!75, line width=0.75pt](V7) to (V6);
\draw[color=black!75, line width=0.75pt](V7) to (V8);
\draw[color=black!75, line width=0.75pt](V6) to (V8);
\draw[color=black!75, line width=0.75pt](V6) to (V5);
\draw[color=black!75, line width=0.75pt](V8) to (V5);

\draw[color=black!75, line width=0.75pt, bend left=2cm](V4) to (V8);

\end{tikzpicture}}
\caption{Two problematic almost-cliques in a graph with maximum degree $\Delta=4$. Both almost-cliques are \emph{hard} for palette sparsification. The only way part \textbf{(a)} can be $\Delta$-colored
is if the vertices incident on the non-edge have intersecting lists. Part \textbf{(b)} is \emph{not} $\Delta$-colorable if the outside (marked) vertices are all colored the same. 
}
%\label{fig:color} Unused, anyway
\end{figure}

\paragraph{$(\Delta+1)$-clique minus an edge $(u,v)$ (\Cref{fig:easy}):} Suppose we know all edges of such an almost-clique. 
We can color both $u$ and $v$ the same using a color that does not appear in their outside neighbors (which is possible because vertices of almost-cliques have few edges out); 
the standard greedy algorithm now actually manages to $\Delta$-color this almost-clique (recall that we assumed the knowledge of \emph{all} edges of the almost-clique for now). 
The argument is simply the following: Pick a common neighbor $z$ of $u$ and $v$, which will exist in 
an almost-clique, and greedily color vertices from a color not used in their neighborhood, waiting for $z$ to be colored last; at this point, since two neighbors of $z$ have the same color, there is still a choice for $z$ to be colored with in the  algorithm.

\paragraph{$\Delta$-clique with few neighbors (\Cref{fig:hard}):} These are more problematic cases if we have ended up coloring all their outside neighbors the same. Even if we know all edges of this almost-clique, 
there is no way we can color a $\Delta$-clique with $(\Delta-1)$ colors (the same one color is ``blocked'' for all vertices). So, the next ingredient of our algorithm is a \textbf{recoloring step} that allows 
for handling these almost-cliques (again, recall~\Cref{lesson3}); we show that our strengthened palette sparsification  is  flexible enough that, given the edges of a new almost-clique, 
allows for altering some of its \emph{past} decisions on outside vertices of this almost-clique. 
This step involves reasoning about a probabilistic process (possibility of having a ``good'' color to change for an outside vertex) \emph{after} already viewing
the outcome of the process (having ended up with a ``blocked'' color). This requires a careful analysis which is handled by partitioning the randomness of the process into multiple \emph{phases} and using our classification of 
almost-cliques to limit the amount of ``fresh randomness'' we need for this step across these  phases. We discuss this in more details in~\Cref{sec:coloring}. 

The discussion above oversimplified many details. Most importantly, we actually do not have such a ``clean'' picture as above for the remaining almost-cliques we need to color. Instead,  
the  algorithm needs to handle almost-cliques that are not fully recoverable by sparse recovery (as they are not sufficiently dense), 
using their helper structures described earlier.%
\footnote{For a $(\Delta+1)$-clique minus an edge $(u,v)$, the helper structure is  the subgraph consisting of vertices $u$ and $v$, 
and all edges incident on at least one of them. The helper structure of the other example is more tricky; see~\Cref{fig:helper}.}
This coloring is thus done via a combination of the greedy arguments of the above type on the helper 
structures, combined with palette sparsification ideas for the remaining vertices of the almost-clique. This in turns requires using some \textbf{out of (sampled) palette coloring} of vertices, which is in conflict 
with what the palette sparsification does and  needs to be handled carefully; see~\Cref{sec:unholey-friendly,sec:unholey-critical}.

% !TeX root = main.tex
%!TEX root = main.tex

\newcommand{\Vsparse}{\ensuremath{V_{\textnormal{sparse}}}}
\newcommand{\NS}[1]{\ensuremath{N_{\textnormal{sample}}(#1)}}
\newcommand{\sample}{\ensuremath{\textnormal{\textsf{SAMPLE}}}\xspace}

\newcommand{\myC}{\ensuremath{C}}

\newcommand{\CN}[2]{\ensuremath{\mathsf{ColNei}_{#2}(#1)}}
\newcommand{\cn}[2]{\ensuremath{\mathsf{coldeg}_{#2}(#1)}}

\newcommand{\Avail}[2]{\ensuremath{\mathsf{Avail}_{#2}(#1)}}
\newcommand{\avail}[2]{\ensuremath{\mathsf{avail}_{#2}(#1)}}

\newcommand{\Uncol}[2]{\ensuremath{\mathsf{Uncol}_{#2}(#1)}}
\newcommand{\uncol}[2]{\ensuremath{\mathsf{uncol}_{#2}(#1)}}

\newcommand{\RemCol}[2]{\ensuremath{\mathsf{RemCol}_{#2}(#1)}}
\newcommand{\remcol}[2]{\ensuremath{\mathsf{remcol}_{#2}(#1)}}

\newcommand{\PhiV}{\ensuremath{\Phi}^{\textnormal{V}}}
\newcommand{\PhiR}{\ensuremath{\Phi}^{\textnormal{R}}}

\newcommand{\OneShot}{\textnormal{\texttt{one-shot-coloring}}\xspace}
\newcommand{\AC}[2]{\ensuremath{ \mathcal{A}_{#1}(#2)}}
\newcommand{\UN}[2]{\ensuremath{ \mathcal{U}_{#1}(#2)}}
\newcommand{\drem}[2]{\ensuremath{ d_{#1}(#2) }}
\newcommand{\Xprime}{\ensuremath{X^{'}}}
\newcommand{\Tprime}{\ensuremath{T^{'}}}
\newcommand{\Kset}{\ensuremath{ \mathcal{K} }}

\renewcommand{\MM}{\mathcal{M}}

\section{Preliminaries}\label{sec:prelim}

\paragraph{Notation.}  For an integer $t \geq 1$, we define
$[t] := \set{1,2,\ldots,t}$.
For a graph $G=(V,E)$ and a vertex $v \in V$, we use $N_G(v)$ to denote the
neighbors of $v$ in $G$, and $N(v)$ when the graph is clear from the context.
Further, we define degree of $v$ by $\deg_G(v) := \card{ N_G(v) }$ (and similarly $\deg(v)$).
We refer to a pair $(u,v)$ of vertices in $V$ as a \textbf{non-edge} when there
is no edge between $u$ and $v$ in $G$. Similarly, we sometimes say that $u$ is a \textbf{non-neighbor} of $v$ and vice versa.

For a graph $G=(V,E)$ and integer $q \geq 1$, we refer to any function $\myC: V \rightarrow [q]$ as a \textbf{$q$-coloring} and call it a \emph{proper} coloring iff there is no edge $(u,v)$ in $G$ with $\myC(u)=\myC(v)$. 
We further refer to a function $\myC: V \rightarrow [q] \cup \set{\perp}$ as a \textbf{partial $q$-coloring} and call the vertices $v \in V$ with $\myC(v) = \perp$ as \textbf{uncolored} vertices by $\myC$. The edges $(u,v)$ in $G$ 
with $\myC(u) = \myC(v) = \perp$ are \emph{not} considered monochromatic, and thus we consider a partial $q$-coloring {proper} iff there is no edge $(u,v)$ in $G$ with $\myC(u) = \myC(v) \neq \perp$. 
Finally, for any proper partial $q$-coloring, and any vertex $v \in V$, we define: 
\begin{itemize}
	\item $\CN{v}{\myC} := \set{u \in N(v) \mid \myC(u) \neq \perp}$: the neighbors of $v$ that are assigned a color by $\myC$; we further define $\cn{v}{\myC} := \card{\CN{v}{\myC}}$. 
	\item $\Avail{v}{\myC} := \set{c \in [q] \mid \myC(u) \neq c \text{~for all~} u \in N(v)}$: the colors in $[q]$ that have not been assigned to any neighbor of $v$ by $\myC$, i.e., are \emph{available} to $v$; we  define
	 $\avail{v}{\myC} := \card{\Avail{v}{\myC}}$. 
\end{itemize}
Finally, we say that a coloring $\myC_1$ is an \textbf{extension} of coloring $\myC_2$ iff for every $v \in V$ with $\myC_2(v) \neq \perp$, $\myC_1(v) = \myC_2(v)$; in other words, only uncolored vertices in $\myC_2$ may receive a different color in $\myC_1$. 

For a bipartite graph $H=(L,R,E)$, a matching $M$ is any collection of \emph{vertex-disjoint} edges. We say that a matching $M$ is an \textbf{$L$-perfect matching} iff it matches all vertices in $L$. We use the following presentation of 
the well-known Hall's marriage theorem~\cite{Hall35}. 

\begin{fact}[cf.~\cite{Hall35}]\label{fact:halls-theorem}
	Suppose $H=(L,R,E)$ is a bipartite graph such that for any set $A \subseteq L$, we have $\card{N(A)} \geq \card{A}$; then $H$ has an $L$-perfect matching. 
\end{fact}

\medskip

\paragraph{Concentration results.}
We use the following standard form of Chernoff bound in our proofs.
 \begin{proposition}[Chernoff bound; c.f.~\cite{DubhashiP09}]\label{prop:chernoff}
 	Suppose $X_1,\ldots,X_m$ are $m$ independent random variables with range $[0,b]$ each. Let $X := \sum_{i=1}^m X_i$ and $\mu_L \leq \expect{X} \leq \mu_H$. Then, for any $\delta > 0$,
 	\[
 	\Pr\paren{X >  (1+\delta) \cdot \mu_H} \leq \exp\paren{-\frac{\delta^2 \cdot \mu_H}{(3+\delta) \cdot b}} \quad \textnormal{and} \quad \Pr\paren{X <  (1-\delta) \cdot \mu_L} \leq \exp\paren{-\frac{\delta^2 \cdot \mu_L}{(2+\delta) \cdot b}}.
 	\]
 \end{proposition}

\smallskip

Throughout, we say that an event happens \emph{``with high probability''}, or \emph{``w.h.p.''} for short, to mean that it happens with probability at least $1-1/\poly{(n)}$ for some large polynomial (the degree can be arbitrarily large
without changing the asymptotic performance of the algorithms). Moreover, we pick this degree to be large enough to allow us
to do a union bound over the polynomially many events considered and we do not explicitly mention this union bound each time (but in certain places that we need to do a union bound over exponentially many events, we will be more explicit).

 \subsection{A Sparse-Dense Decomposition}\label{sec:decomposition}

We use a simple corollary of known streaming sparse-dense decompositions in~\cite{AssadiCK19a,AssadiC22}, which have their origin in the classical work in graph theory~\cite{Reed98,MolloyR13} and have subsequently been used
extensively in distributed algorithms as well~\cite{HarrisSS16,ChangLP18,Parter18,HalldorssonKMT21}.

The decomposition is based on partitioning the vertices into \emph{``(locally) sparse''} vertices with many non-edges among their neighbors and \emph{``almost-clique''} vertices that are part of a subgraph which is close to being a clique. We formally define these as follows.

\begin{definition}\label{def:sparse}
	For a graph $G=(V,E)$ and parameter $\eps > 0$, a vertex $v \in V$ is \textnormal{\textbf{$\eps$-sparse}} iff there are at least $\eps^2 \cdot \Delta^2/2$ \underline{non-edges} between the neighbors of $v$.
\end{definition}

\begin{definition}\label{def:almost-clique}
	For a graph $G=(V,E)$ and parameter $\eps > 0$, a subset of vertices $K \subseteq V$ is an \textnormal{\textbf{$\eps$-almost-clique}} iff $K$ has the following properties:
		\begin{enumerate}[label=$\roman*)$.]
	\item\label{dec:size} Size of $K$ satisfies $(1-5\eps)\cdot \Delta \leq |K| \leq (1+5\eps)\cdot \Delta$.
	\item\label{dec:non-neighbors}  Every vertex $v\in K$ has $\leq 10\eps\Delta$ \emph{non-neighbors} inside $K$;
\item\label{dec:neighbors} Every vertex $v\in K$ has $\leq 10\eps\Delta$ \emph{neighbors} outside $K$;
\item\label{dec:outside} Every vertex $u \notin K$ has $\geq 10\eps\Delta$ \emph{non-neighbors} inside $K$.
\end{enumerate}
\end{definition}

We note that property~\ref{dec:outside} of \Cref{def:almost-clique} does not typically appear in the definition of almost-cliques in prior work but it is crucial for our proofs. However, 
this property follows immediately from the proof of~\cite{AssadiC22}. 
We have the following sparse-dense decomposition.
\begin{proposition}[Sparse-Dense Decomposition; cf.~\cite{MolloyR13,AssadiCK19a,AssadiC22}]\label{prop:decomposition}
	There is a constant $\eps_0 > 0$ such that the following holds. For any $0 < \eps < \eps_0$, vertices of any  graph $G=(V,E)$ can be partitioned into
	\textnormal{\textbf{sparse vertices}} $\Vsparse$ that are $\eps$-sparse (\Cref{def:sparse}) and \textnormal{\textbf{dense vertices}} partitioned into a collection of disjoint $\eps$-almost-cliques $K_1,\ldots,K_k$ (\Cref{def:almost-clique}).

Moreover, there is an absolute constant $\gamma > 0$ and an algorithm that given access to only the following information about $G$, with high probability, computes this decomposition of $G$:
\begin{itemize}
	\item \textnormal{\textbf{Random edge samples:}} A collection of sets $\NS{v}$ of $(\gamma \cdot \eps^{-2} \cdot \log{n})$
	neighbors of every vertex $v \in V$ chosen independently and uniformly at random (with repetition);
	\item \textnormal{\textbf{Random vertex samples:}} A set $\sample$~of vertices wherein each $v \in V$ is included independently with probability $(\gamma \cdot \log{n}/\Delta)$,
	together with all the neighborhood $N(v)$ of each sampled vertex $v \in \sample$.
\end{itemize}
\end{proposition}

We note that a dense vertex is \emph{not} necessarily not $\eps$-sparse.
Or in other words, the set $\Vsparse$ may not include all the $\eps$-sparse
vertices of $G$, as some $\eps$-sparse vertices may still be be included in some $\eps$-almost-clique.

\subsection{Sparse Recovery}\label{sec:sr}

We also use the following standard variant of {sparse recovery} in our proofs.
We note that the specific recovery matrix below, the Vandermonde matrix, is not
necessary for our proofs (i.e., can be replaced with any other standard
construction) and is only mentioned explicitly for concreteness.

 \begin{proposition}[cf.~\cite{DasV13}]\label{prop:sr}
 	Let $n,k \geq 1$ be arbitrary integers and $p \geq n$ be a prime number.
	Consider the $(2k \times n)$-dimensional Vandermonde matrix over $\mathbb{F}_p$:
	\[
		\PhiV :=
		\begin{bmatrix}
		1 & 1 & 1 & \cdots & 1 \\
		1 & 2 & 3 & \cdots & n \\
		1 & 2^2 & 3^2 & \cdots & n^2 \\
		\vdots & \vdots & \vdots & \ddots & \vdots \\
		1 & 2^{2k-1} & 3^{2k-1} & \cdots & n^{2k-1}
		\end{bmatrix}, \quad \textnormal{or for all $i \in [2k]$ and $j \in [n]$:} \quad \PhiV_{i,j} = j^{i-1} \!\!\!\!\mod p.
	\]
	Then, for any $k$-sparse vector $x \in \mathbb{F}_p^n$, one can uniquely recover $x$ from $\PhiV \cdot x$ in polynomial time.
 \end{proposition}
 
The proof that a $k$-sparse vector $x$ can be recovered from $\PhiV \cdot x$ is
simply the following: for $\PhiV \cdot x = \PhiV \cdot y$ for two $k$-sparse
vectors $x \neq y$ (i.e. recovery is impossible), 
 we need $\PhiV \cdot (x-y) = 0$. But then this means that the (at most) $2k$ columns in the matrix $\PhiV$ corresponding to the support of $x-y$ have a non-trivial kernel; the latter is a contradiction as any 
 $2k$ columns of $\PhiV$ are independent (polynomial time recovery also can be obtained via syndrome decoding from coding theory). 
 
 In order to safely use sparse recovery (in case when we mistakenly run it on a non-sparse vector), we also need the following standard result that allows us to test whether the output of the recovery is indeed correct or not. 
 This result is also standard and is proven for completeness. 
 \begin{proposition}\label{prop:eq-test}
 	Let $n,t \geq 1$ be arbitrary integers and $p \geq n$ be a prime number. Consider the $(t \times n)$-dimensional random matrix $\PhiR$ over $\mathbb{F}_p$ chosen uniformly from all matrices in $\mathbb{F}_p^{t \times n}$. 
	Then, for any two different vectors $x \neq y \in \mathbb{F}_p^n$, we have, 
	\[
		\Pr_{\PhiR}\paren{\PhiR \cdot x = \PhiR \cdot y} = p^{-t}. 
	\]
 \end{proposition}
 \begin{proof}
 	Let $z := x-y$ and note that there exists at least one index $i \in [n]$ with $z_i \neq 0$. We need to calculate the probability of $\PhiR \cdot z = 0$. Let $r$ be any row of matrix $\PhiR$. We have, 
	\[
		\Pr\Paren{\inner{r}{z}=0} = \Pr\paren{r_i \cdot z_i = \sum_{j\neq i} r_j \cdot z_j} = p^{-1},
	\]
	because $r_i \cdot z_i$ is going to be any element of the field $\mathbb{F}_p$ with equal probability (as $z_i$ is non-zero) and choice of $r_i$ is independent of $\set{r_j \mid j \neq i}$. As all the $t$ rows of $\PhiR$ 
	are independent, we get the final bound immediately. 
 \end{proof}
 
 We can now combine~\Cref{prop:sr} and~\Cref{prop:eq-test} to have a ``safe'' recovery w.h.p. as follows: For the (unknown) vector $x \in \mathbb{F}_p^n$, we compute $\PhiV \cdot x$ and $\PhiR \cdot x$ in parallel. We first use $\PhiV \cdot x$ 
in~\Cref{prop:sr} to recover a vector $y \in \mathbb{F}_p^n$; then, since we know $\PhiR$, we can also compute $\PhiR \cdot y$ and use~\Cref{prop:eq-test} to check whether $\PhiR \cdot y = \PhiR \cdot x$: if yes, we output $y$ and otherwise output `fail'. It is easy to see that if $x$ is indeed $k$-sparse, this scheme always recovers $x$ correctly, and in any other case, w.h.p., it does not recover a wrong vector (but may output `fail').

\subsection{Palette Graphs, Matching View of Coloring, and Random Graph Theory}\label{sec:palette-graph}

We also borrow a key technique from~\cite{AssadiCK19a} for coloring almost-cliques in the decomposition. In the following, we shall follow the presentation of~\cite{AssadiCK19a} as specified in the notes by~\cite{PSnotes} which is conceptually identical but notation-wise slightly different from the original presentation.  

We note that this subsection might be rather too technical and not intuitive enough at this stage and can be skipped by the reader on the first read of the paper---we will get to these topics 
only starting from~\Cref{sec:lonely} of our algorithm once we start with the final coloring step of our algorithm, and by that time we have set the stage more for these definitions.

\begin{definition}\label{def:palette-graph}
	Let $G = (V, E)$ be a graph, $q \geq 1$ be an integer, and $\myC$ be a proper partial $q$-coloring of $G$. 
	Let $K$ be any almost-clique in $G$. 
	We define the \textnormal{\textbf{base palette graph}} of $K$ and $\myC$ as the following bipartite graph $\Gbase := (\cL,\cR,\Ebase)$: 
	\begin{itemize}
		\item \emph{Vertex-set:} $\cL$ consists of all vertices in $K$ that are uncolored by $\myC$ and $\cR$ consists of all the colors in $[q]$ that are \underline{not} assigned to any vertex in $K$. To avoid ambiguity, we use \emph{nodes} to refer to elements of $\cL$ and $\cR$ (as opposed to vertices), and call nodes in $\cL$ as \emph{vertex-nodes}, and nodes in $\cR$ as \emph{color-nodes}.
		\item \emph{Edge-set:} there is an edge between any pair of vertex-node $v \in \cL$ and color-node $c \in \cR$ iff $c \in \Avail{v}{\myC}$, i.e., $c$ does not appear in the neighborhood of $v$ in the partial coloring $\myC$. 
	\end{itemize}
\Cref{fig:palette-graph} gives an illustration of this definition. 
\end{definition}

The base palette graph gives us a different graph theoretic way of looking at graph coloring. Suppose we start with a proper partial $q$-coloring $\myC$ of $G$ and manage to find an \textbf{$\cL$-perfect matching} $\MM$ in $\Gbase$;%
\footnote{Clearly, this will not  be always possible, for instance when size of $K$ is  larger than $q$---in general, one needs 
some preprocessing steps before being able to apply this idea.}
this will allow us to find an extension of $\myC$ that colors all vertices of $K$: simply color 
each vertex $v \in K$ with the color $c$ which corresponds to the matched pair of $\MM(v)$ (in $\Gbase$). This ``matching view'' of the coloring problem turns out to be quite helpful when analyzing almost-cliques in~\cite{AssadiCK19a}, and we shall 
use and considerably generalize this idea in this paper as well.  

There is however an obvious obstacle in using $\Gbase$ when coloring the graph: we may not have access to all of $\Gbase$ (when using a semi-streaming algorithm due to space limitations). 
This motivates the next definition. 

\begin{definition}\label{def:sampled-palette-graph}
	Let $G = (V, E)$ be a graph, $q \geq 1$ be an integer, and $\myC$ be a proper partial $q$-coloring of $G$. 
	Let $K$ be any almost-clique in $G$ and $\SS:= \set{S(v) \subseteq [q] \mid v \in K}$ be a collection of sampled colors.%
\footnote{Think of $S(v)$ as being a small set of colors sampled for each vertex. For instance, in the context of 
	palette sparsification theorem of~\cite{AssadiCK19a}, $S(v)$ is the set of $O(\log{n})$ colors sampled for each vertex.}
	We define the \textnormal{\textbf{sampled palette graph}} of $K$, $\myC$, and $\SS$, denoted by $\Gsample = (\cL,\cR,\Esample)$ as the spanning subgraph of the base palette graph $\Gbase$ obtained by letting $\Esample$ to be the edges $(v,c)$ from 
	$\Ebase$ such that $c \in S(v)$. 
\end{definition}

%\Cref{fig:palette-graph} gives an illustration of this definition. 

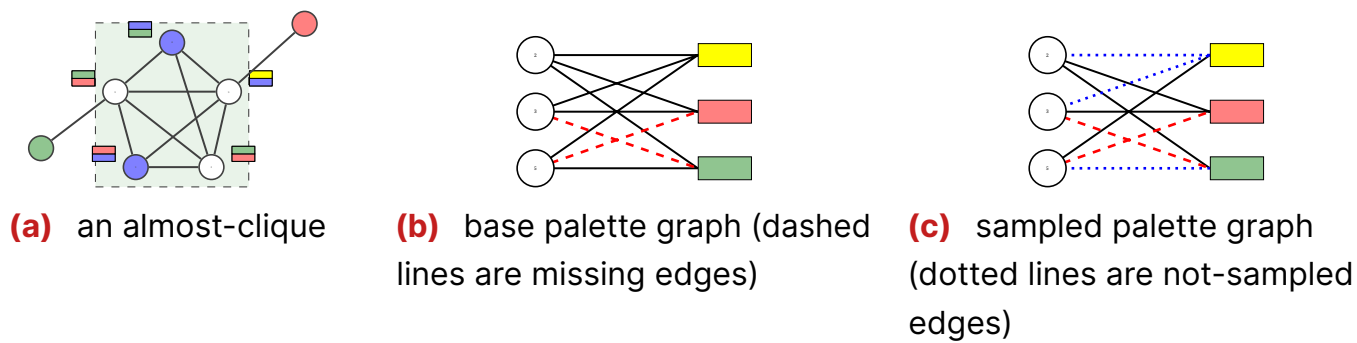
\begin{figure}[t!]
\centering
  \subcaptionbox{\footnotesize an almost-clique \label{fig:cliquey}}%
  [.28\linewidth]{\definecolor {processblue}{cmyk}{0.96,0,0,0}
\tikzset{myptr/.style={decoration={markings,mark=at position 1 with %
    {\arrow[scale=3,>=stealth]{>}}},postaction={decorate}}}

\newcommand{\nodenum}[1]{{\fontsize{1pt}{1pt}\selectfont #1}}

\newcommand{\colorgridbr}{\tikz{
\draw[xstep=0.3cm,ystep=0.1cm] (0,0)  grid (0.3,0.2); 
\filldraw[fill=blue!50] (0,0) rectangle (0.3,0.1);
\filldraw[fill=red!50] (0,0.1) rectangle (0.3,0.2);
}}

\newcommand{\colorgridrf}{\tikz{
\draw[xstep=0.3cm,ystep=0.1cm] (0,0)  grid (0.3,0.2); 
\filldraw[fill=red!50] (0,0) rectangle (0.3,0.1);
\filldraw[fill=ForestGreen!50] (0,0.1) rectangle (0.3,0.2);
}}

\newcommand{\colorgridbg}{\tikz{
\draw[xstep=0.3cm,ystep=0.1cm] (0,0)  grid (0.3,0.2); 
\filldraw[fill=blue!75] (0,0) rectangle (0.3,0.1);
\filldraw[fill=Gray] (0,0.1) rectangle (0.3,0.2);
}}

\newcommand{\colorgridfy}{\tikz{
\draw[xstep=0.3cm,ystep=0.1cm] (0,0)  grid (0.3,0.2); 
\filldraw[fill=blue!50] (0,0) rectangle (0.3,0.1);
\filldraw[fill=Yellow] (0,0.1) rectangle (0.3,0.2);
}}

\newcommand{\colorgridry}{\tikz{
\draw[xstep=0.3cm,ystep=0.1cm] (0,0)  grid (0.3,0.2); 
\filldraw[fill=red!50] (0,0) rectangle (0.3,0.1);
\filldraw[fill=ForestGreen!50] (0,0.1) rectangle (0.3,0.2);
}}

\newcommand{\colorgridrg}{\tikz{
\draw[xstep=0.3cm,ystep=0.1cm] (0,0)  grid (0.3,0.2); 
\filldraw[fill=red!75] (0,0) rectangle (0.3,0.1);
\filldraw[fill=Gray] (0,0.1) rectangle (0.3,0.2);
}}

\newcommand{\colorgridrp}{\tikz{
\draw[xstep=0.3cm,ystep=0.1cm] (0,0)  grid (0.3,0.2); 
\filldraw[fill=red!50] (0,0) rectangle (0.3,0.1);
\filldraw[fill=ForestGreen!50] (0,0.1) rectangle (0.3,0.2);
}}

\newcommand{\colorgridfp}{\tikz{
\draw[xstep=0.3cm,ystep=0.1cm] (0,0)  grid (0.3,0.2); 
\filldraw[fill=ForestGreen!50] (0,0) rectangle (0.3,0.1);
\filldraw[fill=blue!50] (0,0.1) rectangle (0.3,0.2);
}}
\begin {tikzpicture}[ auto ,node distance =1cm and 1cm , on grid , semithick, inner sep=3pt]

\node[circle, black!75,  fill=blue!50, draw]  (V14){\nodenum{1}};%[below left=2.3cm and 0.25cm of V10]{\nodenum{ 14}};
\node (C14)[above left=5pt and 12pt of V14]{\colorgridbr};

\node[circle, black!75,  fill=white!50, draw]  (V15)[below right=0cm and 1cm of V14]{\nodenum{2}};
\node (C15)[above right=5pt and 12pt of V15]{\colorgridrp};

\node[circle, black!75,  fill=white, draw]  (V16)[above left=1cm and 0.27cm of V14]{\nodenum{3}};
\node (C16)[above left=5pt and 12pt of V16]{\colorgridry};

\node[circle, black!75,  fill=blue!50, draw]  (V17)[above right=0.65cm and 0.75cm of V16]{\nodenum{4}};
\node (C17)[above left=5pt and 12pt of V17]{\colorgridfp};

\node[circle, black!75,  fill=white, draw]  (V22)[below right=0.65cm and 0.75cm of V17]{\nodenum{5}};
\node (C22)[above right=5pt and 12pt of V22]{\colorgridfy};

\node[circle, black!75,  fill=red!50, draw]  (Z1)[above right=0.25cm and 1.75cm of V17]{\nodenum{\textcolor{red!50}{1}}};
\node[circle, black!75,  fill=ForestGreen!50, draw]  (Z2)[below left=-0.25cm and 1.25cm of V14]{\nodenum{\textcolor{ForestGreen!50}{1}}};

\begin{scope}[on background layer]
\node[black!75, draw, fill=ForestGreen!10, inner sep=2.5pt, fit=(V14) (V15) (V16) (V17) (V22), dashed]{};
\end{scope}
%%
%%\draw[color=black!75, line width=0.75pt](V14) to (V8);
%%\draw[color=black!75, line width=0.75pt](V10) to (V17);

\draw[color=black!75, line width=0.75pt](V15) to (V22);
\draw[color=black!75, line width=0.75pt](V22) to (V14);
\draw[color=black!75, line width=0.75pt](V16) to (V14);
\draw[color=black!75, line width=0.75pt](V16) to (V22);
\draw[color=black!75, line width=0.75pt](V16) to (V17);
\draw[color=black!75, line width=0.75pt](V17) to (V15);
\draw[color=black!75, line width=0.75pt](V16) to (V15);
\draw[color=black!75, line width=0.75pt](V22) to (V17);
\draw[color=black!75, line width=0.75pt](V14) to (V15);

\draw[color=black!75, line width=0.75pt](Z2) to (V16);
\draw[color=black!75, line width=0.75pt](Z1) to (V22);
%%\draw[color=red, line width=1pt, dashed] (V14) to (V17); 

\end{tikzpicture}}
  \subcaptionbox{\footnotesize  base palette graph (dashed lines are missing edges) \label{fig:base-p}}%
  [.31\linewidth]{\definecolor {processblue}{cmyk}{0.96,0,0,0}
\tikzset{myptr/.style={decoration={markings,mark=at position 1 with %
    {\arrow[scale=3,>=stealth]{>}}},postaction={decorate}}}

\newcommand{\colorgrid}{}

\newcommand{\nodenum}[1]{{\fontsize{2pt}{2pt}\selectfont #1}}

\begin {tikzpicture}[ auto ,node distance =1cm and 1cm , on grid , semithick]

\node[circle, black, fill=white, draw]  (a1){\nodenum{2}};
\node[circle, black, fill=white, draw]  (a2)[below=0.75cm of a1]{\nodenum{3}};
\node[circle, black, fill=white, draw]  (a3)[below=0.75cm of a2]{\nodenum{5}};

\node[rectangle,minimum height=7.5pt, minimum width=20pt, black, fill=Yellow, draw, line width=0.25pt] (P1) [above right=0cm and 2.5cm of a1]{};
\node[rectangle,minimum height=7.5pt, minimum width=20pt, black, fill=red!50, draw, line width=0.25pt] (P2) [below=0.75cm of P1]{};
\node[rectangle,minimum height=7.5pt, minimum width=20pt, black, fill=ForestGreen!50, draw, line width=0.25pt] (P3) [below=0.75cm of P2]{};

\begin{scope}[on background layer]
\draw[color=black, line width=0.75pt]
(a1.center) to (P1.center)
(a1.center) to (P2.west)
(a1.center) to (P3.west)
(a2.center) to (P1.west)
(a2.center) to (P2.west)
(a3.center) to (P1.west)
(a3.center) to (P3.west);
\end{scope}

\begin{scope}[on background layer]
\draw[color=red, line width=1pt, dashed]
(a2.center) to (P3.west)
(a3.center) to (P2.west);
\end{scope}

\end{tikzpicture}} \hspace{0.25cm} 
  \subcaptionbox{\footnotesize  sampled palette graph (dotted lines are not-sampled edges) \label{fig:sampled-p}}%
  [.31\linewidth]{\definecolor {processblue}{cmyk}{0.96,0,0,0}
\tikzset{myptr/.style={decoration={markings,mark=at position 1 with %
    {\arrow[scale=3,>=stealth]{>}}},postaction={decorate}}}

\newcommand{\colorgrid}{}

\newcommand{\nodenum}[1]{{\fontsize{2pt}{2pt}\selectfont #1}}

\begin {tikzpicture}[ auto ,node distance =1cm and 1cm , on grid , semithick]

\node[circle, black, fill=white, draw]  (a1){\nodenum{2}};
\node[circle, black, fill=white, draw]  (a2)[below=0.75cm of a1]{\nodenum{3}};
\node[circle, black, fill=white, draw]  (a3)[below=0.75cm of a2]{\nodenum{5}};

\node[rectangle,minimum height=7.5pt, minimum width=20pt, black, fill=Yellow, draw, line width=0.25pt] (P1) [above right=0cm and 2.5cm of a1]{};
\node[rectangle,minimum height=7.5pt, minimum width=20pt, black, fill=red!50, draw, line width=0.25pt] (P2) [below=0.75cm of P1]{};
\node[rectangle,minimum height=7.5pt, minimum width=20pt, black, fill=ForestGreen!50, draw, line width=0.25pt] (P3) [below=0.75cm of P2]{};

\begin{scope}[on background layer]
\draw[color=black, line width=0.75pt]
(a1.center) to (P2.west)
(a1.center) to (P3.west)
(a2.center) to (P2.west)
(a3.center) to (P1.west);
\end{scope}

\begin{scope}[on background layer]
\draw[color=red, line width=1pt, dashed]
(a2.center) to (P3.west)
(a3.center) to (P2.west);
\end{scope}

\begin{scope}[on background layer]
\draw[color=blue, line width=1pt, dotted]
(a1.center) to (P1.center)
(a2.center) to (P1.west)
(a3.center) to (P3.west);
\end{scope}

\end{tikzpicture}}
\caption{An illustration of base palette graphs (\Cref{def:palette-graph}) and sampled palette graphs (\Cref{def:sampled-palette-graph}).}
\label{fig:palette-graph}
\end{figure}

Let us again fix a small almost-clique $K $ and the partial coloring~$\myC$. Consider the sampled palette graph $\Gsample$ of $K$, $\myC$, and $\SS := \set{S(v) \mid v \in K}$ for $S(v)$ of size, say, $\poly\!\log{(n)}$, chosen randomly from $[q]$. 
This is now a much sparser subgraph that is easier to maintain via a semi-streaming algorithm.  Similar to before, if we find an {$\cL$-perfect matching} in $\Gsample$ we will be done. The challenge now, however, is to argue that not only $\Gbase$, 
but even $\Gsample$ (that has much fewer edges) contains such a matching (with high probability over the choice of sampled lists). 

This challenge can be addressed  using ``random graph theory type'' arguments: we first establish several key properties of 
$\Gbase$ itself that ensure that it has an $\cL$-perfect matching; then, we consider $\Gsample$ which is a random subgraph of $\Gbase$ and use simple tools in the analysis of random graphs to prove
that $\Gsample$ also w.h.p.\ has an $\cL$-perfect matching.%
\footnote{Note that if $\Gbase$ is indeed a bipartite clique, then $\Gsample$ would become a standard random graph. However, in general, $\Gbase$ can be 
``sufficiently far'' from a bipartite clique, which requires a careful analysis of $\Gsample$ beyond known results in random graph theory. See~\Cref{lem:rgt} (and its proof in~\cite{AssadiCK19a}) for an example.}
The following lemma mentions one example of such a random graph theory type 
argument that played a key role in~\cite{AssadiCK19a} (we shall use this lemma and some news ones that we establish throughout our proofs in this paper).

\begin{lemma}[\cite{AssadiCK19a}]\label{lem:rgt}
	Let $H = (L,R,E)$ be any bipartite graph with the following properties: 
	\begin{enumerate}[label=$(\roman*)$]
		\item $m := \card{L}$ and $m \leq \card{R} \leq 2m$; 
		\item The minimum degree of vertices in $L$ is at least $(2/3) \cdot m$, i.e., $\min_{v \in L} \deg_{H}(v) \geq (2/3) \cdot m$; 
		\item For every set $A \subseteq L$ of size $\card{A} \geq m/2$, we have $\sum_{v \in A} \deg_{H}(v) \geq (\card{A} \cdot m) - m/4$.
	\end{enumerate}
	For any $\delta \in (0,1)$, a subgraph of $H$ obtained by sampling each edge independently with probability at least $(\frac{20}{m} \cdot (\log{m}+\log{(1/\delta)}))$ contains an $L$-perfect matching with probability at least $1-\delta$. 
\end{lemma}

As discussed earlier, the way one applies~\Cref{lem:rgt} is by setting $H$ to be the base palette graph of a given almost-clique, in which case, sampled palette graph has the same distribution
as specified in~\Cref{lem:rgt} and thus w.h.p. will have the desired matching. Also, while at this stage the properties of $H$ in this lemma may sound rather arbitrary, as we shall see later in our proofs, 
they appear naturally as properties of base graphs of (certain) almost-cliques. Finally, we note that even though the proof of this lemma is rather technical and so we do not repeat it here, it is not hard to verify that at least the graph $H$ itself has an $L$-perfect matching  using Hall's theorem (\Cref{fact:halls-theorem}), given the conditions imposed on it in~\Cref{lem:rgt}.

\newcommand{\PSalg}{\textnormal{\texttt{palette-sampling}}\xspace}
\newcommand{\SRalg}{\textnormal{\texttt{sparse-recovery}}\xspace}
\newcommand{\DECalg}{\textnormal{\texttt{find-decomposition}}\xspace}

\newcommand{\IndS}[1]{\ensuremath{I_{\textnormal{sample}}(#1)}}
\newcommand{\rmax}{\ensuremath{r_{\max}}}
\newcommand{\nondeg}[2]{\ensuremath{ \overline{\deg}_{#1}}({#2})}
\newcommand{\nonN}[2]{\ensuremath{ \overline{N}_{#1}}({#2})}

\newcommand{\Hplus}{\ensuremath{ H^{+} }}

\newcommand{\Kfriendly}{\ensuremath{\mathcal{K}_{\textnormal{friendly}}}}
\newcommand{\Klonely}{\ensuremath{\mathcal{K}_{\textnormal{lonely}}}}
\newcommand{\Kcritical}{\ensuremath{\mathcal{K}_{\textnormal{critical}}}}

\newcommand{\HR}{\ensuremath{H^+}}

\section{The Semi-Streaming Algorithm}\label{sec:alg}

In this section we describe the algorithm that collects necessary information
for the $\Delta$-coloring from the stream. This will then be used
in our main coloring procedure in~\Cref{sec:coloring} to prove~\Cref{res:main}.
The algorithm consists of the following three main parts:
\begin{itemize}
	\item The $\PSalg$ algorithm (\Cref{alg:PS}): An algorithm, quite similar to palette sparsification approach, that samples
        $\poly\!\log{(n)}$ potential colors for each vertex and store all possibly monochromatic edges during
        the stream.
	\item The $\DECalg$ algorithm (\Cref{alg:DEC}): An algorithm that recovers
        a sparse-dense decomposition of the input graph as specified
        in~\Cref{prop:decomposition} plus some extra useful information about
        the decomposition.
	\item The $\SRalg$ algorithm (\Cref{alg:SR}): An algorithm that uses
        sparse recovery techniques to extract further ``helper structures'' about
        the almost-cliques in the decomposition of the previous step.
\end{itemize}

We elaborate on each of these algorithms and their guarantees in the following
subsections. But we shall emphasize that they all run \emph{in parallel} in a
single pass over the stream. To continue, we start with setting up some parameters and key definitions.

\subsection{Parameters, Classification of Almost-Cliques, and Helper Structures}\label{sec:classification}

We use the following parameters for the design of our algorithms.
\begin{align}
	&\alpha = 10^{3}: \quad &&\text{a large constant used to
    simplify various concentration inequalities} \notag \\
	&\beta = 100 \cdot \log{n}: &&\text{used to bound the size of certain
    palettes in $\PSalg$} \notag \\
	&\eps = \frac{10^{-8}}{\log{n}}: &&\text{used as the parameter of
    sparse-dense decomposition of~\Cref{prop:decomposition}}.
    \label{eq:parameters}
\end{align}

We also assume%
\footnote{This is used in the proof of \Cref{lem:os-gap}.} that $\Delta = \Omega(\log^5{n})$
as otherwise we can simply store the graph entirely and solve the problem
offline, using any classical algorithm for Brooks' theorem.

Recall the notion of almost-cliques in~\Cref{def:almost-clique} used in our
sparse-dense decomposition.
In the coloring phase of the algorithm, we make further distinctions between almost-cliques based on their sizes as
defined below---the coloring algorithm will treat these classes separately and
our algorithms in this section provide further information about these
different classes.

\paragraph{Classification of almost-cliques.} We start with the following simple definition that partitions almost-cliques based on their size. 

\begin{definition}\label{def:scl-almost-clique}
	Let $K$ be an almost-clique in the sparse-dense decomposition. We say that $K$ is \textnormal{\textbf{small}} iff it has at most $\Delta$
	vertices, \textnormal{\textbf{critical}} iff it has exactly
            $\Delta+1$ vertices, and  \textnormal{\textbf{large}} otherwise.
\end{definition}
We have the following basic observation based on this definition. 
\begin{observation}\label{obs:critical-almost-clique}
	$(i)$ Any critical almost-clique contains at least one non-edge, and $(ii)$ any large almost-clique contains at least  $(\Delta+2)/2$ non-edges. 
\end{observation}
\begin{proof}
	Property $(i)$ holds because there are no $(\Delta+1)$ cliques in our input as otherwise the graph will not be $\Delta$-colorable, and $(ii)$ holds because maximum degree of vertices is $\Delta$
and thus every vertex has at least one non-edge in a large almost-clique.
\end{proof}

A property that fundamentally affects how we color an almost-clique is the
number of non-edges inside it. Intuitively, if an almost-clique is very
``clique-like'', that is, it has very few non-edges inside, then it is
more difficult to color. This motivates the following definition. 

\begin{definition}\label{def:holey-almost-clique}
	Let $K$ be an almost-clique in the sparse-dense decomposition. We say that~$K$ is \textnormal{\textbf{holey}} iff it has at least
            $10^7 \cdot \eps\Delta$ non-edges (or ``holes'') inside it. Otherwise, $K$ is \textnormal{\textbf{unholey}}. 
\end{definition}

Another key property that governs our ability to color an almost-clique is how
it is connected to the outside and in particular, what we can expect from
a coloring of its neighbors outside: can we see those colors as being
``random'' or are they ``adversarial''? In the latter case, can we recolor some
to make them ``less adversarial''?
This motivates the following two definitions. 

\begin{definition}\label{def:friend-stranger}
	Let $K$ be an almost-clique in the decomposition and $v$ be any vertex
    outside~$K$ that is neighbor to $K$. We say that:
	\begin{itemize}%[leftmargin=15pt]
		\item $v$ is a \textnormal{\textbf{friend}} of $K$ iff there are at
            least $2\Delta / \beta$ edges from $v$ to $K$, i.e.,
            $\card{N(v) \cap K} \geq 2\Delta/\beta$;
		\item $v$ is a \textnormal{\textbf{stranger}} to $K$ iff there are less
            than $\Delta / \beta$ edges from $v$ to $K$, i.e.,
            $\card{N(v) \cap K} < \Delta/\beta$;
	\end{itemize}
\end{definition}

We emphasize that there is a gap in the criteria between friend and stranger
vertices, and hence it is possible that a vertex is neither a friend of nor
a stranger to an almost-clique. Based on the notion of friend and stranger
vertices, we can further classify almost-cliques into these classes. 

\begin{definition}
    \label{def:fls-cliques}
	Let $K$ be an almost-clique in the decomposition. We say that $K$ is:
    \begin{itemize}%[leftmargin=15pt]
        \item \textnormal{\textbf{friendly}} iff $K$ has at least one neighbor
            outside $K$ that is a friend of $K$.
        \item \textnormal{\textbf{lonely}} iff all neighbors of $K$ outside
            $K$ are strangers.
        \item \textnormal{\textbf{social}} otherwise; that is, $K$
            has at least one neighbor outside $K$ that is \emph{not} a
            stranger, but at the same time, it has no friends.
    \end{itemize}
\end{definition}

We approach coloring friendly and lonely almost-cliques quite differently, but
both approaches can handle social almost-cliques.
This will be crucial as we can distinguish between friendly and lonely
almost-cliques but our tester may classify social almost-cliques in either of
these groups. 

\paragraph{Helper structures.} As stated earlier, the problematic almost-cliques to $\Delta$-color are unholey ones. 
Some of these unholey almost-cliques can be handled by a ``global'' argument that reason 
about the coloring of their neighbors outside. But for the rest, we may need to consider recoloring some of their neighbors outside and/or using some extra information about the graph. 
In the following, we define two ``helper structures'' that provide this extra information for our coloring approach. Our algorithms in this section then show how we can find these subgraphs in the stream. 

The first structure we have handles unholey almost-cliques which are critical. 

 \begin{definition}\label{def:helper-critical}
	Let $K$ be an unholey almost-clique which is critical. We define a \textnormal{\textbf{critical-helper structure}} for $K$ as a tuple $(u,v,N(v))$ with the following properties: 
	\begin{enumerate}[label=$(\roman*)$]
		\item $u,v$ are vertices of the graph and are both in $K$; 
		\item $u$ and $v$ are non-neighbor to each other;
		\item $N(v)$ is the neighborhood of $v$ in the graph. 
	\end{enumerate}
\end{definition}

At a very high level, if we have a critical-helper structure of $K$ at hand, we can color~$u$ and~$v$ the same (the crucial knowledge of $N(v)$ allows us to do this) which ``buys'' us an extra color which will be 
sufficient for us to color the entire almost-clique also.

The second structure we have handles unholey almost-cliques which are friendly or even social. 

\begin{definition}\label{def:helper-friendly}
	Let $K$ be an unholey almost-clique which is either friendly or social. We define a \textnormal{\textbf{friendly-helper structure}} for $K$ as a tuple $(u,v,w,N(v),N(w))$ with the following properties: 
	\begin{enumerate}[label=$(\roman*)$]
		\item $u,v,w$ are vertices of the graph such that $u \notin K$ and is \underline{not} a stranger to $K$ and $v,w \in K$;
		\item $u$ is neighbor to $v$ and non-neighbor to $w$, and $v$ and $w$ are themselves neighbors; 
		\item $N(v)$ and $N(w)$ are the neighborhoods of $v$ and $w$ in the graph, respectively. 
	\end{enumerate}
\end{definition}

Again, at a  high level, if we have a friendly-helper structure of $K$ at hand, we will be able to (re)color $u$ and $w$ the same, which ``buys'' us an extra color for $v$ and gives us the required flexibility for coloring the entire almost-clique (the knowledge
 of $N(v)$ and $N(w)$ crucially allows us to choose the colors for these vertices without creating a conflict with  their neighbors in the graph).

\begin{figure}[t!]
    \centering
    \subcaptionbox{A critical helper: The vertices $u$ and $v$ can receive the
    same color. \label{fig:helper-critical}}[.45\linewidth]{\pgfdeclarelayer{edges}       % hide edges behind ellipses
\pgfdeclarelayer{ellipses}    % hide edges behind vertices
\pgfsetlayers{edges,ellipses,main}
% set the order of the layers (main is the standard layer)

\begin{tikzpicture}
    [vertex/.style={circle, draw=blue!50, fill=blue!20, thick, inner sep=0pt,
     minimum size=2mm},
     red/.style={vertex, draw=red!50, fill=red!20},
     ghost/.style={inner sep=0pt, minimum size=0mm},
     scale=0.7]
    \node [vertex] (u) [label=above left:{\footnotesize$u$}] at (4, 4.5) {};
    \node [vertex] (v) [label=above left:{\footnotesize$v$}] at (4, 3) {};

    \draw (u) edge [dashed] (v);

    \node [ghost] (K) [label=right:{\footnotesize $K$}] at (4, 5) {};

    \draw (4, 3) ellipse (1.5cm and 3cm); % the almost-clique K
    \begin{pgfonlayer}{ellipses}
        \draw [fill=gray!20] (3, 1.5) ellipse (2cm and .5cm); % N(v)
    \end{pgfonlayer}

    %\node [ghost] (w1) [] at (1.5, 2.75) {};
    %\node [ghost] (w2) [] at (1.5, .75) {};

    \node [ghost] (v1) [] at (1, 1.5) {};
    \node [ghost] (v2) [] at (5, 1.5) {};

    \begin{pgfonlayer}{edges}
        \draw (v) edge (v1);
        \draw (v) edge (v2);
    \end{pgfonlayer}

    \node [ghost] (Nv) [] at (2, 1.5) {\footnotesize$N(v)$};
    % \coordinate[below=.7cm of v1] (v11);
    % \coordinate[below=.7cm of v2] (v22);
    % \draw [decorate,decoration={brace,amplitude=10pt,mirror}]
    % (v11) -- (v22) node [black, midway, xshift=-.5cm, yshift=-.75cm] (Nv) {$N(v)$};

    % \draw [help lines] (0,0) grid (7, 7);
\end{tikzpicture}}
    \hspace{0.4cm}
        \subcaptionbox{A friendly helper: The vertices $u$ and $w$ can receive the same color.
    \label{fig:helper-friendly}}%
     [.45\linewidth]{\pgfdeclarelayer{edges}       % hide edges behind ellipses
\pgfdeclarelayer{ellipses}    % hide edges behind vertices
\pgfsetlayers{edges,ellipses,main}
% set the order of the layers (main is the standard layer)

\begin{tikzpicture}
    [vertex/.style={circle, draw=blue!50, fill=blue!20, thick, inner sep=0pt,
     minimum size=2mm},
     red/.style={vertex, draw=red!50, fill=red!20},
     ghost/.style={inner sep=0pt, minimum size=0mm},
     scale=0.7]
    \node [red] (v) [label=above:{\footnotesize $ v$}] at (4, 4.5) {};
    \node [vertex] (w) [label=below:$w$] at (4, 3) {};

    \node [vertex] (u) [label=below:{\footnotesize $u$}] at (1.5, 4) {};

    \draw (u) edge [dashed] (w);
    \draw (u) edge [out=0, in=180] (v);

    \node [ghost] (K) [label=right:{\footnotesize $K$}] at (4, 1.5) {};

    \begin{pgfonlayer}{ellipses}
        \draw (4, 3) ellipse (1.5cm and 3cm); % the almost-clique K
        \draw [fill=gray!20] (1.5, 1.75) ellipse (0.5cm and 1cm); % N(w)
        \draw [fill=gray!20] (1.5, 4) ellipse (0.5cm and 1cm); % N(v)
    \end{pgfonlayer}

    \node [ghost] (w1) [] at (1.5, 2.75) {};
    \node [ghost] (w2) [] at (1.5, .75) {};

    \node [ghost] (v1) [] at (1.5, 5) {};
    \node [ghost] (v2) [] at (1.5, 3) {};

    \node [ghost] (u1) [] at (4, 5.5) {};
    \node [ghost] (u2) [] at (4, 3.5) {};

    \draw (u) edge [out = 0, in = 180] (u1);
    \draw (u) edge [out = 0, in = 180] (u2);

    \begin{pgfonlayer}{edges}
        \draw [fill=white] (4, 4.5) ellipse (.5cm and 1cm); % N(u)
        \draw (w) edge (w1);
        \draw (w) edge (w2);
        \draw (v) edge (v1);
        \draw (v) edge (v2);
    \end{pgfonlayer}

    % \coordinate[left=.4cm of v1] (v11);
    % \coordinate[left=.4cm of v2] (v22);
    % \draw [decorate,decoration={brace,amplitude=10pt, mirror}]
    % (v11) -- (v22) node [black, midway, xshift=-1.25cm] (Nv) {$N(v) \setminus K$};

    \node [ghost] (Nv) at (-0.25, 4) {\footnotesize $N(v) \setminus K$};

    % \coordinate[left=.4cm of w1] (w11);
    % \coordinate[left=.4cm of w2] (w22);
    % \draw [decorate,decoration={brace,amplitude=10pt, mirror}]
    % (w11) -- (w22) node [black, midway, xshift=-1.25cm] (Nw) {$N(w) \setminus K$};

    \node [ghost] (Nw) at (-0.25, 1.75) {\footnotesize $N(w) \setminus K$};

    \coordinate[right=.4cm of u1] (u11);
    \coordinate[right=.4cm of u2] (u22);
    \draw [decorate,decoration={brace,amplitude=10pt}]
    (u11) -- (u22) node [black, midway, xshift=7pt] (Nu) {};

    \node [ghost] (Nulabel) [right=1cm of Nu]
    {\footnotesize $\card{N(u) \cap K} \geq \frac{\Delta}{\beta}$};

    \draw [->] (Nu) edge (Nulabel);

    \draw (v) edge (w);

    %\draw [help lines] (0,0) grid (7, 7);
\end{tikzpicture}}
    \caption{The two types of helper structures in~\Cref{def:helper-critical} and~\Cref{def:helper-friendly}.}
    \label{fig:helper}
\end{figure}

\subsection{Palette Sampling}\label{sec:ps}

One key component of our algorithm is a color sampling procedure in the same spirit as the palette sparsification theorem of~\cite{AssadiCK19a}.

\begin{algorithm}
	\KwIn{Graph $G=(V,E)$ with known vertices $V$ and streaming edges $E$.}
	\begin{enumerate}[label=$(\roman*)$]
      \item
        For every vertex $v \in V$, sample the following lists of colors:
		\begin{itemize}%[leftmargin=-10pt]
            \item $L_1(v)$: Sample a single color chosen
                uniformly at random from $[\Delta]$.
            \item $L_2(v)$: Sample each color in $\bracket{\Delta}$
                independently with probability $\frac{\beta}{\Delta}$.
            \item $L_3(v)$: Sample each color in $\bracket{\Delta}$
                independently with probability
                $\frac{100 \cdot \alpha\cdot \log n}{\eps^2 \cdot \Delta}$.
			\item $L_4(v) := (L^*_4(v)\text{~and~} {L_{4, i}(v)\textnormal{: $i \in \bracket{\beta}$}})$:
                Independently, sample each color in $\bracket{\Delta}$ with probability $\frac{\beta}{\Delta}$ in $L_4^*(v)$ 
                and with $q~:=~\frac{1}{100\sqrt{\eps}\Delta}$ in $L_{4,i}(v)$ for $i \in [\beta]$. 
            \item $L_5(v)$: Sample each color in $\bracket{\Delta}$
                independently with probability $\frac{\beta}{\Delta}$.
            \item $L_6(v) :=
                ({L_{6, i}(v)\textnormal{: $i \in \bracket{2 \beta}$}})$:
                Sample each color in $\bracket{\Delta}$ independently with
                probability $\beta^2 / \Delta$.
		\end{itemize}
              \item
                Let $L(v) := \cup_{j \in [6]} L_j(v)$. Store any
			 edge $({u, v})$ if $L(u) \cap L(v) \neq \emptyset$
			and let $H$ be the subgraph of $G$ on these stored edges,
			referred to as the \textbf{conflict graph}. 
                      \end{enumerate}
                      	\caption{The $\PSalg$ algorithm.}\label{alg:PS}
\end{algorithm}

The main difference of this algorithm with that of~\cite{AssadiCK19a} is that the number of sampled colors per vertex is larger here ($\poly\!\log{(n)}$ as opposed to $O(\log{n})$) and that we explicitly 
partition these samples into multiple lists instead of just one. 

At this point, the choices of the lists $L_1(v) \ldots L_6(v)$ may seem
arbitrary:%
\footnote{And indeed redundant; technically speaking, we could have just  sampled a single list of colors of proper size and postponed the partitioning to the analysis (as in fact done in~\cite{AssadiCK19a})---however, we find it more transparent to consider these lists explicitly separate from each other due to various dependency issues that this explicitly avoids.}
the (very) rough idea is that we will use different lists at
different stages of the coloring, and the sizes are chosen to allow each
stage to go through without causing too much dependency for the next stage.
In the rest of this subsection, we will bound the space complexity of $\PSalg$ (\Cref{alg:PS}).

Recall that in the conflict graph $G$, we only keep an edge $({u, v})$ if
$L(u) \cap L(v) \neq \emptyset$---which makes sense, since if we restrict ourselves to coloring vertices with a color from their lists $L(\cdot)$, then these are the only edges that can be
monochromatic. We have the following (standard) lemma.

\begin{lemma}\label{lem:ps-space-edges}
    With high probability,  conflict graph $H$ in $\PSalg$ has
    $O(n \cdot \log^6 n)$ edges.
\end{lemma}
\begin{proof}
    First, note that for any vertex $v \in V$, $L_1(v)$ has a single color.
    We want to establish that the sizes of the other lists
    $\card{L_2(v)}, \ldots , \card{L_6(v)}$ are bounded by $O(\log^3 n)$ with
    high probability.

    For $L_2(v)$: We have that the expected size of $L_2(v)$ is
    $\frac{\beta}{\Delta} \cdot \Delta = \beta$
    by linearity of expectation. Since each color is sampled into $L_2(v)$
    independently, we have via an application of Chernoff bound
    (\Cref{prop:chernoff}, with $\delta = 1$) that:
    \[
        \prob{ \card{L_2(v)} > 2\beta} \leq
        \exp\paren{- \frac{1^2 \cdot \beta}{3 + 1} } \leq n^{-25}.
    \]
    By union bound, we have that with high probability, $\card{L_2(v)}$ is
    bounded by $2\beta = O(\log n)$ for all $v \in V$.
    We can apply the same argument on all lists with expected size
    $\Omega(\log n)$ to show that their sizes are within a constant of their
    respective expected values with high probability. In particular,
    $L_3(v)$, $L_4^*(v)$, $L_5(v)$, all have expected sizes that are
    $\Omega(\log n)$ and also $O(\log^3 n)$ (see \Cref{eq:parameters}), and
    hence have size $O(\log^3 n)$ for each one with high probability.
    Further, since for each $i$, $L_{6, i}(v)$ has expected size $\beta^2$,
    $\card{L_6(v)}$ is $O(\log^3 n)$ as well, with high probability.

    This leaves us with the lists $L_{4, i}(v)$: Note that their expected size
    is $\frac{1}{100\sqrt{\eps}} = 100\sqrt{\log n}$. Since each color is
    sampled into $L_{4, i}(v)$ independently, we can use a Chernoff bound with
    $\delta = 10\sqrt{\log n}$ to get:
    \[
        \prob{ \card{L_{4, i}(v)} > 1000\log n} \leq
        \exp\paren{- \frac{100 \log n \cdot 100\sqrt{\log n}}{3 + 10\sqrt{\log n}}}
        \leq
        n^{-100}.
    \]
    And hence w.h.p., $\card{L_4(v)}$ is $O(\log^2 n)$.

    At this point, we have that there exist an absolute constant $\gamma > 0$
    such that with high probability,
    $\card{L(v)} < \gamma \cdot \log^3{n}$ for all $v \in V$.
    We condition on this event for the rest of this proof.

    Our strategy is to bound $\deg_H(u)$ for each $u \in V$.
    Recall that an edge $\cbrac{u, v}$ is in $H$ if $L(v)$ samples a color
    from $L(u)$. An arbitrary color $c \in \bracket{\Delta}$ is in
    $L(v)$ with probability at most~$\gamma\cdot\log^3(n) / \Delta$ (since each
    of the lists $L_i(v)$ is sampled uniformly).
    Then for any choice of $L(u)$ (of size at most~$\gamma\cdot \log^3 n$), we
    have by the union bound:
    \begin{align*}
        \prob{ L(u) \cap L(v) \neq \emptyset }
        = \prob{ \displaystyle\bigcup_{z \in L(u)} z \in L(v) }
        \leq \sum_{z \in L(u)} \prob{ z \in L(v) }
        = \frac{\gamma^2 \log^6(n)}{\Delta},
    \end{align*}
    where the randomness is over the choice of $L(v)$.

    Now we can bound the degree of an arbitrary vertex $u$ in the subgraph $H$
    formed by the edges stored in \Cref{alg:PS}. First, fix the list $L(u)$---we already conditioned on the event that it is small, but now we will
    ``give up'' the remaining randomness in the choice of $L(u)$, and proceed
    by assuming an arbitrary choice.
    Let $X_{u, v}$ be the
    indicator random variable that is $1$ iff the edge $\cbrac{u, v}$ is in
    $H$. Then the expected degree of $u$ is~$(\gamma^2\cdot \log^6 n)$ by the
    previous argument and linearity of expectation.
    Finally, we observe that for $v \neq w$, $X_{u, v}$ and $X_{u, w}$ are
    independent, and via another application of Chernoff bound
    (\Cref{prop:chernoff}, with $\delta = 1$), we have:
    \[
        \prob{ \deg_H(u) \geq 2\gamma^2 \log^6 n} \leq
        \exp\paren{-\frac {\gamma^2 \log^6 n}{4}}. 
    \]
    Thus by  union bound, each $u \in V$ has degree $O( \log^6 n )$ in $H$ with
    high probability. This immediately implies the bound on the number of edges
    stored by \Cref{alg:PS}.
\end{proof}

We can now bound the space used by $\PSalg$. We have done all of the
heavy-lifting already, by bounding the number of edges stored in $H$. The only new
observation in the following lemma is that edges and colors can be stored
in $O(\log n)$ bits.

\begin{lemma}\label{lem:ps-space-bits}
    With high probability, \PSalg uses $O(n \cdot \log^7 n)$ bits of space.
\end{lemma}
\begin{proof}
    We showed in the proof of \Cref{lem:ps-space-edges} that with high
    probability, $\card { L(v) } = O( \log^3 n)$ for all $v \in V$. Since
    each color is from $\bracket{\Delta}$, it can be represented by
    $\ceil{\log \Delta}$ bits, and hence storing the lists $L(v)$ for all $v$
    uses $O(n \log^4 n)$ bits.

    Further, each edge can be represented by $2\cdot\ceil{\log n}$ bits, so by
    invoking \Cref{lem:ps-space-edges}, we have that we need
    $O(n \cdot \log^7 n)$ bits to store the conflict graph $H$.
\end{proof}

\subsection{Finding the Decomposition}\label{sec:find-decompose}

We also work with the sparse-dense decomposition, but unlike~\cite{AssadiCK19a}, not only as an analytical tool but 
in fact algorithmically (as will become evident from the next subsection). We now describe an algorithm for finding the sparse-dense
decomposition of \Cref{prop:decomposition}. In particular, we only need to
provide the random edge and vertex samples required by the proposition. 
But, in addition to the samples required for the decomposition, we will also collect
independent random edge samples to allow us to distinguish friend vertices
from strangers (\Cref{def:friend-stranger}).

\begin{algorithm}
	\KwIn{Graph $G=(V,E)$ with known vertices $V$ and streaming edges
    $E$.}
	\begin{enumerate}[label=$(\roman*)$]
		\item Let $\gamma$ be the constant from the statement of
			\Cref{prop:decomposition}.
		\item \emph{Vertex samples:} For each vertex $v$, sample $v$
			into the set $\sample$ with probability $(\gamma \cdot \log n / \Delta)$
			independently. During the stream, for each vertex $v$ in $\sample$,
			store all edges incident on $v$.
		\item \emph{Edge samples:} For each vertex $v$, use reservoir
			sampling on edges of $v$ to pick a sample $\NS{v}$
			of size $\paren{\gamma \cdot \eps^{-2} \cdot \log n}$ from its
			neighborhood.
		\item \emph{Neighbor samples:} For each vertex $v$, store
			each neighbor of $v$ in $\IndS{v}$ with probability
			$\beta^2 / \Delta$.
                      \end{enumerate}
                      \caption{The $\DECalg$ algorithm.}\label{alg:DEC}
\end{algorithm}

Let us start by bounding the space used by \DECalg and then present the main
properties of the algorithm for our purpose. 

\begin{lemma}\label{lem:dec-space-bits}
    With high probability, \DECalg uses $O(n\log^4{n})$ bits of space.
\end{lemma}
\begin{proof}
    The set $\sample$ has size $\gamma \cdot n \log n / \Delta$ in expectation, and since
    each vertex $v$ is in $\sample$ independently, the size is at most (say)
    $5\gamma \cdot n \log n / \Delta$ with high probability by Chernoff bound
    (\Cref{prop:chernoff}).
    For each $v \in \sample$, we use upto $\Delta \cdot \ceil{\log n}$ bits
    of space to store all of its edges, and hence in total we use
    $O( n \log^2 n)$ bits to store $\sample$ and the neighborhood of vertices in it.

	The sets $\NS{v}$ have fixed sizes $(\gamma \cdot \eps^{-2} \cdot \log n)$
	each.  Storing a neighbor takes $\ceil{\log n}$ bits, and hence storing
	all the sets $\NS{v}$ takes $O( n \log^4 n)$ bits of space.

    The set $\IndS{v}$ is of size at most $2 \beta^2$ with high probability
    (again, the proof is the same as in \Cref{lem:ps-space-edges}).
    Hence storing $\IndS{v}$ for all $v$ uses $O( n \log^3 n)$ bits of space.
\end{proof}

We now establish the main properties we need from
$\DECalg$. 

\begin{lemma}\label{lem:decomposition}
    We can compute a sparse-dense decomposition
    (\Cref{prop:decomposition}) of the input graph $G = (V, E)$ with high
    probability using the samples collected by $\DECalg$.
\end{lemma}
\begin{proof}
    The proof is immediate---we collect $\sample$ and  $\NS{v}$  as
    needed by \Cref{prop:decomposition}, so we can use the algorithm in the proposition to compute the decomposition.
\end{proof}

We will also show that the independent random edge
samples $\IndS{v}$ are enough for a tester that can distinguish friends from
strangers for any almost-clique.

\begin{lemma}\label{lem:friend-stranger-tester}
	There exists an algorithm that given an almost-clique $K$, a vertex
    $v \notin K$, and the random neighbor samples $\IndS{v}$, with high probability can distinguish:
	\begin{itemize}
		\item $v$ has at most $\Delta / \beta$ neighbors in $K$, that is,
			$v$ is a stranger to $K$;
		\item $v$ has at least $2 \Delta / \beta$ neighbors in $K$, that is,
			$v$ is a friend of $K$.
	\end{itemize}
    The randomness in this lemma is only over the sample $\IndS{v}$.
\end{lemma}
	\begin{proof}
		Fix a vertex $v \notin K$. Let $u$ be any neighbor of $v$ in $K$. Then $u \in \IndS{v}$ with
		probability $\beta^2 / \Delta$. Let $X_u \in \set{0,1}$ be the indicator random variable which is $1$ iff $u \in \IndS{v}$. 
		Thus, $X := \sum_{u \in N(v) \cap K} X_u$ counts the size of intersection of $\IndS{v}$ with $K$. Firstly, 
		we have 
		\[
			\expect{X} = \sum_{u \in N(v) \cap K} \frac{\beta^2}{\Delta} = \card{N(v) \cap K} \cdot \frac{\beta^2}{\Delta}. 
		\]
		Thus, by~\Cref{def:friend-stranger}, $\expect{X}$ is at least $2\beta$ if $v$ is a friend of $K$ and at most $\beta$ if $v$ is a stranger. 
		Our tester can simply compute the value of $X$ and output \textit{friend} if $X$ is more than $\frac 32 \beta$, and \textit{stranger} otherwise. 
	To prove the correctness, suppose that $\card{N(v) \cap K} < \Delta / \beta$, i.e. $v$ is a stranger. Then, by an application of Chernoff bound (\Cref{prop:chernoff} with $\delta = 1/2$) we
        have:
		\[
			\prob{X> \frac 32 \cdot \beta} < \exp\paren{ -\frac{ 1/4 \cdot \beta}{3 + 1/2}} < \exp\paren{-\frac{\beta}{16}} < n^{-6},
		\]
		by the choice of $\beta$ in~\Cref{eq:parameters}. The other case can be proven symmetrically. 
		Hence if $v$ is a stranger (resp. friend), our tester also outputs stranger (resp. friend) with high
        probability.
	\end{proof}
An immediate consequence of \Cref{lem:friend-stranger-tester} is that we
have a tester that can distinguish friendly almost-cliques from lonely
almost-cliques (\Cref{def:fls-cliques}).

\begin{lemma}\label{lem:lonely-friendly-tester}
    There exists an algorithm that given an almost-clique $K$, and 
    to the random neighbor samples $\IndS{v}$ for every $v \in V$, with high probability can distinguish:
    \begin{itemize}
        \item $K$ is a friendly almost-clique;
        \item $K$ is a lonely almost-clique.
    \end{itemize}
   The randomness in this lemma is only over the samples $\set{\IndS{v} \mid v \in V}$.
\end{lemma}
\begin{proof}
    For each vertex $v \in V \setminus K$, run the tester from
    \Cref{lem:friend-stranger-tester}. If $K$ has even one friend, then 
    that friend is distinguished by the tester in~\Cref{lem:friend-stranger-tester} as a friend, and we
    can return that $K$ is friendly. Otherwise, if $K$ is lonely, it means that 
    every vertex $v$ we tested for $K$ will be considered stranger also with high probability and thus we can  correctly mark $K$ as lonely. 
\end{proof}

\subsection{Sparse Recovery for Almost-Cliques}\label{sec:sr-almost-clique}

Finally, we come to the most novel part of this section. 
Recall that as discussed  earlier, palette sparsification (and thus our own \PSalg) is doomed to fail for $\Delta$-coloring. To bypass this, 
we rely on the \emph{helper structures} defined in~\Cref{def:helper-friendly,def:helper-critical}, which, combined with the palette sparsification-type approach of \PSalg, 
allow us to color the graph.  

The first challenge here is 
that we obviously cannot afford to find the neighborhood of every vertex, and we do not know
\emph{during} the stream which vertices will satisfy the properties we need for these structures. 
We step around this by (crucially) using randomization to sample the ``right'' vertices. 
The second and main challenge is that for some almost-cliques (say, a critical almost-clique with only one non-edge), 
we may actually have to recover neighborhood of \emph{all} vertices in the almost-clique before finding the required helper structure; but doing 
this naively requires too much space.  Instead, we use the sparse recovery matrices of~\Cref{sec:sr}, in conjunction with the decomposition found by \DECalg, to 
recover these parts much more efficiently.

\begin{algorithm}
	\KwIn{Graph $G=(V,E)$ with known vertices $V$ and streaming edges
    $E$.}
	For every $r \in R = \cbrac{2^{i} \mid 0 \leq i \leq \ceil{\log \Delta}}$:%
			\footnote{Think of $r$ as a fixed sampling rate from a set of rates
			$R$.}
    \begin{enumerate}[label=$(\roman*)$]

            \item Sample each vertex $v \in V$ in a set $V_r$ independently with probability $\min\cbrac{1, \frac{\beta}{\eps \cdot r}}$.

            \item Construct the $(2r \times n)$ \emph{Vandermonde matrix} $\PhiV_r$
                (see \Cref{prop:sr}) and sample an $(\alpha \times n)$ \emph{random matrix} $\PhiR_r$ (see \Cref{prop:eq-test})
                over the field $\F_p$ where $p$ is a fixed prime larger than $n$ and smaller than, say, $n^2$, and $\alpha$ is the parameter in~\Cref{eq:parameters}. 

            \item For each vertex $v \in V_r$, define a vector $y(v) \in \F_p^{\,2r}$ and $z(v) \in \F_p^{\,\alpha}$ initially set to $0$. 
            For any incoming edge $\cbrac{u,v}$ in the stream, update
            \[
            y(v) \leftarrow y(v) + \PhiV_r  \cdot \mathbf{e}_u \qquad \text{and} \qquad z(v) \leftarrow z(v) + \PhiR_r \cdot  \mathbf{e}_u,
            \] 
            where $\mathbf{e}_u$ is the $n$-dimensional vector which is $1$ on coordinate $u$ and $0$ everywhere else. 
          \end{enumerate}
          	\caption{The $\SRalg$ algorithm.}\label{alg:SR}
\end{algorithm}

Let us start by bounding the space of this algorithm. 

\begin{lemma}\label{lem:sr-space-bits}
	\SRalg uses $O(n\log^4{n})$ bits of space.
\end{lemma}
\begin{proof}
     For each $r \in R$, each $v \in V$ is sampled into $V_r$ with
	probability (at most) $\frac {\beta}{\eps r}$, which means the expected size of
	$V_r$ is (at most) $\frac {\beta n}{\eps r}$. Since each sample is independent, we
	have by Chernoff bound (\Cref{prop:chernoff}) that, 
    \begin{align*}
        \prob{ \card{V_r} > 2 \cdot \frac{\beta n}{\eps r} } &\leq
        \exp\paren{ - \frac{ \beta n}{4 \eps r} } 
        \leq \exp\paren{ - \frac{ 100 n \log n}{4 \eps \cdot 2\Delta} }
        \leq n^{25/2 \cdot \eps} \ll 1/\poly{(n)},
    \end{align*}
    by the choice of $\eps = \Theta(\log^{-1}(n))$ in~\Cref{eq:parameters}. 
    
    Combining with the union bound over $O(\log{\Delta})$ choices of $r \in R$, we have that each $V_r$ is of size at most
    $\frac{2\beta n}{\eps r}$ with high probability.
	For each $v \in V_r$, \Cref{alg:SR} stores two vectors $y(v)$ and $z(v)$ that require $O(r \cdot \log{p})$ and $O(\alpha \cdot \log{p})$ bits, respectively, 
	where $p$ is the order of the field $\F_p$. This, together with the bound on $V_r$ and since $\alpha = \Theta(1)$ by~\Cref{eq:parameters} means that 
	the total number of bits needed to store these vectors is $O(\eps^{-1} \cdot \beta n \log{p} \cdot \log{\Delta}) = O(n\log^{4}{n})$ bits (where we used the fact that both $\eps^{-1}$ and $\log{p}$ are bounded by $O(\log{n})$). 
	
	Finally, the algorithm does not need to explicitly store the Vandermonde matrix for $\PhiV_r$ (as each of its entries can be easily generated in $O(\log{p})$ space at any point) 
	and can store each random matrix $\PhiV_r$ in $O(\alpha \cdot n\log{p}) = O(n\log{n})$ bits as $\alpha = \Theta(1)$ and $\log{p} = O(\log{n})$. Thus, for all $r \in R$, we also need to store $O(n\log^2{n})$ bits to 
	store the random matrices. Overall, the total space of the algorithm is still $O(n\log^4{(n)})$, concluding the proof. 
\end{proof}

We now switch to proving the main properties of $\SRalg$ for our purpose. Before getting into details however, we state a standard observation about linear transformations (namely, $\PhiV_r$ and $\PhiR_r$ in $\SRalg$) over a stream of updates. 

\begin{observation}\label{obs:lin-sketch}
	Fix any $r \in R$ and  $v \in V_r$ in $\SRalg$. At the end of the stream, 
	\[
		y(v) = \PhiV_r \cdot \chi(N(v)) \qquad \text{and} \qquad z(v) =  \PhiR_r \cdot \chi(N(v)),
	\]
	where $\chi(N(v)) \in \set{0,1}^V$ is the characteristic vector of $N(v)$.
\end{observation}
\begin{proof}
	We only prove the equation for $y(v)$; the one for $z(v)$ follows similarly.
	Initially, we set $y(v) = 0$. Then, during the stream, whenever we see the edge
	$\cbrac{u, v}$ for each $u \in N(v)$, we update $y(v)$ by adding $\PhiV_r  \cdot \mathbf{e}_u$ to it. As $\PhiV_r$ is a linear transformation, we have, 
	\[
		y(v) = \sum_{u \in N(v)} \PhiV_r \cdot \mathbf{e}_u
		= \PhiV_r \cdot \paren{\sum_{u \in N(v)} \mathbf{e}_u}
		= \PhiV_r \cdot \chi(N(v)),
	\]
	concluding the proof. 
\end{proof}

Recall from \Cref{prop:sr} that given $\PhiV_r \cdot x$ for an $r$-sparse vector
$x\in \F_p^n$, we can recover~$x$ in polynomial time, and by~\Cref{prop:eq-test}, we can test our recovered vector to make sure with high probability that it is indeed equal to $x$. 
We use this idea combined with the fact that at the end of the stream we know a sparse-dense decomposition of the graph 
to recover one helper structure for each almost-clique (that needs one). We start with the key part that handles the critical-helper structures (\Cref{def:helper-critical}). 
There will also be a simpler part that handles friendly-helper structures almost-cliques (\Cref{def:helper-friendly}).  

\paragraph{Finding Critical-Helper Structures.} We start by recovering a critical-helper structure for any critical almost-clique as defined in~\Cref{def:helper-critical}. Let $K$ be a critical almost-clique and 
$v$ be a vertex in $K$. Define the vector $x(v) := \chi( N(v) ) - \chi( K )$. For any coordinate $u \in [n]$ in $x(v)$, 
\begin{align}
	x(v)_u = \begin{cases}
		\phantom{-}0 &\textnormal{if } u \notin N(v) \cup K \textnormal{ or } u \in N(v) \cap K\\
		\phantom{-}1 &\textnormal{if } u \in N(v) \setminus K\\
		p-1 &\textnormal{if } u \in  K \setminus N(v)
	\end{cases}, \label{eq:x(v)-vector}
\end{align}
as we do the computation over $\F_p$. Notice that  since $v$ belongs to the almost-clique $K$, by~\Cref{def:almost-clique}, size of both $N(v) \setminus K$ and $K \setminus N(v)$ is at most $10\eps\Delta$. 
Thus, $x(v)$ is already considerably sparser than $\chi(N(v))$. In the following, we are going to take this idea to the next level to recover a critical-helper structure 
for $K$ using the vectors computed by \SRalg. 

\begin{lemma}\label{lem:helper-critical}
	There exists an algorithm that given
	a critical almost-clique $K$ (\Cref{def:scl-almost-clique}), with high probability, 
	finds a critical-helper structure $(v,u,N(v))$ of $K$ (\Cref{def:helper-critical}) using the information gathered by $\SRalg$. 
\end{lemma}
\begin{proof}
	(\Cref{fig:find-critical-helper} gives an illustration that might be helpful to refer to during the proof.) For any vertex $w \in K$, define: 
	\begin{itemize}
		\item $\nonN{K}{w} := K \setminus N(w)$: as the non-edge neighborhood of $w$ in $K$ and $\nondeg{K}{w} = \card{\nonN{K}{w}}$ as the non-edge degree of $w$. 
	\end{itemize}
	Define $w^*$ as the vertex that maximizes this non-edge degree, i.e., $w^* = \argmax_{w \in K} \,\,\nondeg{K}{w}$. By \Cref{obs:critical-almost-clique}, we have $\nondeg{K}{w^*} \geq 1$. 
	We first have the following simple claim that will be crucial in finding the neighborhood of at least one vertex in $\nonN{K}{w^*}$ using sparse recovery. 

    \begin{sidefigure}
        \centering
        \pgfdeclarelayer{edges}       % hide edges behind ellipses
\pgfdeclarelayer{ellipses}    % hide edges behind vertices
\pgfsetlayers{edges,ellipses,main}
% set the order of the layers (main is the standard layer)

\begin{tikzpicture}
    [vertex/.style={circle, draw=blue!50, fill=blue!20, thick, inner sep=0pt,
     minimum size=2mm},
     gray/.style={vertex, draw=gray!50, fill=gray!20},
     ghost/.style={inner sep=0pt, minimum size=0mm},
     scale=0.9]
    \node [vertex] (w*) [label=right:{\footnotesize $w^*$}] at (4.5, 4.5) {};

    \node [gray] (w1) [] at (3.5, 4.5) {};
    \node [gray] (w2) [] at (3.5, 4) {};
    \node [ghost] (wi) [] at (3.5, 3.5) {\footnotesize$\vdots$};

    \node [vertex] (w) [label=right:{\footnotesize$v$}] at (3.5, 2.75) {};

    \node [ghost] (K) [label=right:{\footnotesize$K$}] at (4, 5.5) {};

    \draw (4, 3) ellipse (1.5cm and 3cm); % the almost-clique K
    \begin{pgfonlayer}{ellipses}
        \draw [fill=gray!20] (1, 2.75) ellipse (.5cm and 1cm); % N(w) \ K
        \draw [fill=gray!20] (4, 1.5) ellipse (1cm and .5cm); % K \ N(w)
    \end{pgfonlayer}

    \node [ghost] (Nw1) [] at (1, 3.75) {};
    \node [ghost] (Nw2) [] at (1, 1.75) {};
    \node [ghost] (NNw1) [] at (3, 1.5) {};
    \node [ghost] (NNw2) [] at (5, 1.5) {};

    \draw (w*) edge [dashed] (w1);
    \draw (w*) edge [dashed] (w2);
    \draw (w*) edge [dashed] (w);
    \begin{pgfonlayer}{edges}
        \draw (w) edge (Nw1);
        \draw (w) edge (Nw2);
        \draw (w) edge [dashed] (NNw1);
        \draw (w) edge [dashed] (NNw2);
    \end{pgfonlayer}

    %\coordinate[left=.4cm of Nw1] (Nw11);
    %\coordinate[left=.4cm of Nw2] (Nw22);
    %\draw [decorate,decoration={brace,amplitude=10pt, mirror}]
    %(Nw11) -- (Nw22) node [black, midway, xshift=-2.5cm] (Nw)
    %{$\card{N(w) \setminus K} \leq \nondeg{K}{w^*}$};

    \node [ghost] (Nw) at (-1.95, 2.75)
    {\footnotesize $\card{N(v) \setminus K} \leq \nondeg{K}{w^*} \approx r$};

    \coordinate[below=.4cm of NNw1] (NNw11);
    \coordinate[below=.4cm of NNw2] (NNw22);
    \draw [decorate,decoration={brace,amplitude=10pt, mirror}]
    (NNw11) -- (NNw22) node [black, midway, yshift=-7pt] (NNw) {};

    \node [ghost] (NNwLabel) [left= of NNw]
    {\footnotesize $\card{K \setminus N(v)} \leq \nondeg{K}{w^*} \approx r$};

    \draw [->] (NNw) edge [out=270, in=270] (NNwLabel);

    % \draw [help lines] (0,0) grid (7, 7);
\end{tikzpicture}
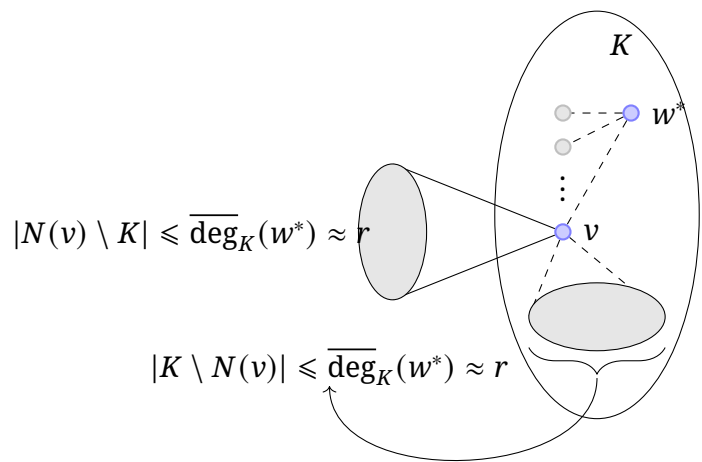
        \caption{Finding a critical-helper: With high probability, there is a
        vertex $v \in \nonN{K}{w^*}$ which is sampled in $V_r$.}
        \label{fig:find-critical-helper}
    \end{sidefigure}

	\begin{claim}\label{clm:picked-one}
		Let $r \in R$ be the smallest integer such that $r \geq 2\cdot\nondeg{K}{w^*}$. Then, 
		for every  $w \in \nonN{K}{w^*}$, the vector $x(w) := \chi(N(w)) - \chi(K)$ is $r$-sparse. 
	\end{claim}
	\begin{subproof}
		Fix any $w \in \nonN{K}{w^*}$. By the definition of $w^*$, we have $\nondeg{K}{w} \leq \nondeg{K}{w^*} \leq r/2$. At the same time, since $K$ is a critical almost-clique and thus has size $\Delta+1$, 
		this means that the number of neighbors of $w$ outside $K$ is also at most $\nondeg{K}{w^*} \leq r/2$. By~\Cref{eq:x(v)-vector}, this means that $x(w)$ has at most $r$ non-zero entries. %\Qed{clm:picked-one} 
	\end{subproof}
	
	Consider the parameter $r$ of~\Cref{clm:picked-one}. We have that $\nondeg{K}{w^*} \geq r/4$ as elements of $R$ are within a factor two of each other and by the value of $r$. Given that each vertex is chosen 
	in $r$ with probability $\min\set{1,\beta/(\eps \cdot r)}$, we have, 
	\[
		\Pr\paren{V_r \cap \nonN{K}{w^*} = \emptyset} \leq \paren{1-\frac{\beta}{\eps \cdot r}}^{r/4} \leq \exp\paren{-\frac{\beta}{4\eps}} \ll 1/\poly{(n)},
	\]
	by the choice of $\beta = \Theta(\log{n})$ and $\eps = \Theta(\log^{-1}(n))$ in~\Cref{eq:parameters}. In the following, we condition on the high probability event that a vertex $v \in \nonN{K}{w^*}$ is sampled 
	in $V_r$. For now, let us \emph{assume} that we know the identity of this vertex $v$ in $V_r$. 
	
	Firstly, by~\Cref{clm:picked-one}, we have that the vector $x(v)$ is $r$-sparse. Secondly, since $v \in V_r$, $\SRalg$ has computed $y(v) = \PhiV_r \cdot \chi(N(v))$ by~\Cref{obs:lin-sketch} and since we are given $K$, 
	we can also compute $\PhiV_r \cdot \chi(K)$. Thus, by linearity, we can compute $\PhiV_r \cdot x(v)$ this way. Finally, since $x(v)$ is $r$-sparse, by~\Cref{prop:sr}, 
	we can actually recover $x(v)$ from  $\PhiV_r \cdot x(v)$. But again, since we know $\chi(K)$, this gives us $\chi(N(v))$ and in turn $N(v)$ as well. Thus, the algorithm can 
	return the critical-helper structure $(v,u,N(v))$ for $u=w^*$ which satisfies all the properties (as $(v,w^*)$ is a non-edge). 
	
	It only remains the remove the assumption on the knowledge of identity of $v$ in $V_r$ (and possibly the value of $r$ itself). For this, we simply iterate over all vertices $w \in K$
	and for each one run $\PhiV_r \cdot x(w)$ and $\PhiR_r \cdot x(w)$ as described above for all values of $r \in R$ (by using $z(v)$ in place of $y(v)$ when computing the latter). As outlined in~\Cref{sec:sr}, 
	we can now apply~\Cref{prop:eq-test} to the outcome of each sparse recovery to get that with high probability, any vector $\chi(N(w))$ that we recover is correct. Since 
	we know that the vertex $v$ will \emph{not} output 'fail', we are guaranteed that with high probability we will return a valid critical-helper structure, concluding the proof. 
\end{proof}

\paragraph{Finding Friendly-Helper Structures.} We now switch to finding a
friendly-helper structure of~\Cref{def:helper-friendly} for each almost-clique
$K$ that is unholey and \emph{not} lonely. 

\begin{lemma}\label{lem:helper-friendly}
	There exists an algorithm that given
	an unholey and not lonely almost-clique~$K$ (\Cref{def:scl-almost-clique}) and a vertex $u \notin K$ which is not a stranger to $K$, 
	with high probability, 
	finds a friendly-helper structure $(u,v,w,N(v),N(w))$ of $K$ (\Cref{def:helper-friendly}) using the information gathered by $\SRalg$. 
\end{lemma}
\begin{proof}
	Define $\rmax = \max_{r \in R} r$. We have the following straightforward claim.%
\footnote{A careful reader may notice that in this claim, 
	we actually do not really need sparse recovery; we
could have simply stored all edges of vertices sampled in $V_{\rmax}$
explicitly during the stream.
However, given that we  indeed need sparse recovery for all other ranges of
$V_r$ for $r \in R$ in the previous part, we use a unified approach for $V_{\rmax}$ as well.}
	\begin{claim}\label{clm:full-nbr-prob}
	For any vertex $v \in V$, with probability at least
	$\frac{\beta}{2\eps \cdot \Delta}$, we can recover the set $N(v)$.  
\end{claim}
\begin{subproof}
	Note that $\rmax = 2^{\ceil{\log \Delta}}$ which is between
    $\Delta$ and $2\Delta$.  Recall that each $v$ is sampled into the set
    $V_{\rmax}$ with probability $\frac{\log n}{\eps \cdot \rmax}$,
	which is at least $\frac{\log n}{2 \eps \cdot \Delta}$. And then note that
	for each vertex $v$ in $V_{\rmax}$, $\SRalg$ stores $\PhiV_{\rmax} \cdot \chi(N(v))$ by~\Cref{obs:lin-sketch}. 
	By~\Cref{prop:sr}, we can recover $N(v)$ from $\PhiV_{\rmax}
        \cdot \chi(N(v))$ for every $v \in V_{\rmax}$, concluding the
        proof of the claim. 
%	\Qed{clm:full-nbr-prob}	 
\end{subproof}

	We  prove~\Cref{lem:helper-friendly} using this claim. Firstly, by property~\ref{dec:outside} of \Cref{def:almost-clique}, 
	there are also at least $10\eps \Delta$ vertices $w$ in $K$ that are not
	neighbors of $u$. By \Cref{clm:full-nbr-prob}, we recover $N(w)$ for any such choice of $w$ with probability at least ${\beta}/({2\eps \cdot \Delta})$.
	As such, we have, 
	\[
		\Pr\paren{\textnormal{$N(w)$ is not recovered for any $w \in K \setminus N(u)$}} \leq \paren{1-\frac{\beta}{2\eps \cdot \Delta}}^{\eps \cdot \Delta} \leq \exp\paren{-\frac{\beta}{2}} \leq n^{-50},
	\] 
	by the choice of $\beta = 100\log{n}$ in~\Cref{eq:parameters}. In the following, we further condition on the high probability event that for some $w \in K \setminus N(u)$, we 
	have recovered $N(w)$. Similar to the proof of~\Cref{lem:helper-critical}, let us \emph{assume} that we  know the identity of $w$. 
	
	Now consider $N(w) \cap N(u) \cap K$; since $\card{N(u) \cap K} \geq \Delta/\beta$ as $u$ is not a stranger to $K$ (\Cref{def:friend-stranger}), 
	and $\card{K \setminus N(w)} \leq 10\eps\Delta$ by property~\ref{dec:non-neighbors} of \Cref{def:almost-clique}, 
	we have that 
	\[
	\card{N(w) \cap N(u) \cap K} \geq \Delta/\beta - 10\eps\Delta > \Delta/2\beta
	\]
	 by the choice of parameters $\eps < 10^{-6} \cdot 1/\beta$ in~\Cref{eq:parameters}. 
	By the same argument as above, we have, 
	\[
		\Pr\paren{\textnormal{$N(v)$  not recovered for any $v \in N(w) \cap N(u) \cap K$}} \leq \paren{1-\frac{\beta}{2\eps \cdot \Delta}}^{\Delta/2\beta} \hspace{-10pt} \leq \exp\paren{-\frac{1}{2\eps}} \leq n^{-100},
	\] 
	by the choice of $\eps$ in~\Cref{eq:parameters}. We now have: $u$ is a vertex which is not a stranger to $K$, $w$ is a non-neighbor of $u$ in $K$ and we have $N(w)$, and $v$ is a neighbor of both $u$ and $w$ and we 
	have $N(v)$. Thus, we can return $(u,v,w,N(v),N(w))$ as a friendly-helper structure of~$K$. 
	
	Finally, removing the assumption on the knowledge of $v$ and $w$ is exactly as in the proof of~\Cref{lem:helper-critical}: we simply go over all choices of vertices in $K$ that we have sampled 
	in $V_{\rmax}$ and check whether any pairs of them satisfy the requirements of the structure or not---by the above argument, with high probability, we will find a pair. 
\end{proof}

\subsection{Listing the Information Collected by the Algorithm}\label{sec:info} 

For the ease of reference in the analysis, we now take stock of what all our algorithms collected about the graph from the stream. In particular, 
with high probability, we have the following information at the end of the stream: 

\begin{enumerate}[label=$\emph{\arabic*}.$]
	\item\label{p1} A list of sampled colors $L(v)$ for every vertex $v \in V$ as specified in~\Cref{alg:PS}. 
	
	\emph{Proof:} Follows  from the definition of $\PSalg$ in \Cref{alg:PS}. 
	
	\item\label{p2} The \textbf{conflict graph} $H$ consisting of every edge $(u,v)$ in the graph where $L(u) \cap L(v) \neq \emptyset$. 
	
	\emph{Proof}: Follows from the definition of $\PSalg$ in \Cref{alg:PS}. 
	
	\item\label{p3} A decomposition of $G$ into sparse vertices and almost-cliques as specified in~\Cref{prop:decomposition}. 
	
	\emph{Proof:} Follows from~\Cref{lem:decomposition} for $\DECalg$ in \Cref{alg:DEC}. 
	
	\item\label{p4} A collection $\Kfriendly$ of almost-cliques that contains \emph{all} friendly almost-cliques and \emph{no} lonely almost-clique, and for each $K \in \Kfriendly$, 
	one vertex $u \notin K$ which is \underline{not} a stranger to $K$. 
	
	\emph{Proof:} Follows from Part~\eqref{p3} and~\Cref{lem:lonely-friendly-tester} for $\DECalg$ in \Cref{alg:DEC}. 
	
	\item\label{p5} A collection $\Klonely$ of almost-cliques that contains \emph{all} lonely almost-cliques and \emph{no} friendly almost-clique. Moreover, $\Kfriendly \sqcup \Klonely$ partition \emph{all} almost-cliques, which also
	implies that every social almost-clique belongs to exactly one of these two collections. 
	
	\emph{Proof:} Follows from Parts~\eqref{p3},~\eqref{p4}, and~\Cref{lem:lonely-friendly-tester} for $\DECalg$ in \Cref{alg:DEC}. 
	
	\item\label{p6} A collection of $\Kcritical$ of critical almost-cliques and for each $K \in \Kcritical$, a \textbf{critical-helper structure} $(u,v,N(v))$ of~\Cref{def:helper-critical}. 
	
	\emph{Proof:} Follows from Part~\eqref{p3} and \Cref{lem:helper-critical} for $\SRalg$ in \Cref{alg:SR}.

	\item\label{p7} A collection of \textbf{friendly-helper structures} $\set{(u,v,w,N(v),N(w))}$ of~\Cref{def:helper-friendly}, one for each $K \in \Kfriendly$ such that $u$ is the vertex specified for $K \in \Kfriendly$ in Part~\eqref{p4}. 
	
	\emph{Proof:} Follows from Part~\eqref{p4} and \Cref{lem:helper-friendly} for $\SRalg$ in \Cref{alg:SR}. 
	
	\item\label{p8} The \textbf{recovery graph} $\HR$ consisting of all edges in the critical-helper structures of Part~\eqref{p6} and in the friendly-helper structures of Part~\eqref{p7}. 
	
	\emph{Proof:} Follows from Parts~\eqref{p6} and~\eqref{p7}. 
	
\end{enumerate}

We shall note that at this point, we covered all the process that is done by our algorithm \emph{during} the stream and what remains is to prove this information is useful, i.e., we can indeed
color the graph in the \emph{post-processing step} using this information. This is the content of the next section. 

Before moving on, we should note that, with high probability, the space complexity of $(i)$ $\PSalg$ is $O(n\log^7{n})$ bits by~\Cref{lem:ps-space-bits}, $(ii)$ $\DECalg$ is $O(n\log^4{n})$ bits by~\Cref{lem:dec-space-bits}, 
and $(iii)$ $\SRalg$ is $O(n\log^4{n})$ by~\Cref{lem:sr-space-bits}. Thus, our entire streaming algorithm takes $O(n\log^7{n})$ space. This adhere to the space complexity promised in~\Cref{res:main}.   

\begin{remark}
	As stated, the space complexity of our algorithm is bounded with high probability but not in the worst-case. This is standard to fix; simply run the algorithm as it is and whenever it attempted to use more than, say, $100$ times, 
	the space guaranteed by its expectation, terminate it and ``charge'' the failure probability to the error. 
\end{remark}

\begin{remark}[Removing Prior Knowledge of $\Delta$]
  \label{rem:no-know-delta}
  We observe that this semi-streaming algorithm does not really need to know
  $\Delta$ before the stream begins. In particular, we can run $O(\log n)$
  independent copies of the algorithm, each with a difference ``guess'' of
  $\Delta \in \set{2^k \mid 2^k \leq 2n}$.
  At the same time, we can compute $\Delta$ at the end of the stream by simply
  counting for each vertex the number of edges incident to it in $O(n \log n)$
  space.

  An overestimate of $\Delta$ does not hurt us in terms of space
  usage, but an underestimate can (for example, if we guess $\Delta = 1$, and
  the input includes a clique on $n - 1$ vertices, the algorithm stores the
  entire graph).
  Hence, if at any point (for any guess $\Delta$), if a vertex has degree
  larger than $2\Delta$, we stop that run of the algorithm.
  Now, at the end of the stream we will have:
  \begin{itemize}
    \item The actual maximum degree $\Delta$.
    \item The output of the algorithm for $\Delta'$ and $2\Delta'$ such that
      $\Delta' \leq \Delta \leq 2\Delta'$.
  \end{itemize}
  But now we can get the desired samples by ``resampling'' the outputs from
  \Cref{alg:PS,alg:DEC,alg:SR}.
  In particular, for each vertex $v \in V$, and each color $c \in L(v)$
  from the run of \Cref{alg:PS} with guess $\Delta'$, we keep $c$ with
  probability $\Delta' / \Delta$.
  Since we are only dropping colors from the palettes, this process only
  removes some edges from the conflict graph.
  The samples of \Cref{alg:DEC} are adapted in a similar manner. Finally, the
  set of sampling rates $R$ in \Cref{alg:SR} for $2\Delta'$ is a superset of
  that for a (hypothetical) run with the correct guess of $\Delta$, so we can
  just ignore the vectors corresponding to unused sample rates.%
  \footnote{Technically, there is no need to even run \Cref{alg:SR} separately
  for different guesses of $\Delta$.}

  Hence at the end of the stream we know $\Delta$ and can adapt the samples
  as required, so the coloring procedure in the next section can proceed as
  normal.
\end{remark}

\newcommand{\gap}{\ensuremath{ \textsf{gap} }}
\newcommand{\assign}{\ensuremath{ \textsf{assign} }}
\newcommand{\lose}{\ensuremath{ \textsf{lose} }}
\newcommand{\Present}{\ensuremath{ \textsf{Present} }}
\newcommand{\present}{\ensuremath{ \textsf{present} }}
\newcommand{\Blocked}{\ensuremath{ \textsf{Blocked} }}

\newcommand{\Bin}[2]{\ensuremath{ \mathcal{B}\paren{#1, #2} }}

\newcommand{\ColMalg}{\textnormal{\texttt{colorful-matching}}\xspace}
\newcommand{\ColUHCC}{\textnormal{\texttt{unholey-critical-coloring}}\xspace}
\newcommand{\ColUHFC}{\textnormal{\texttt{unholey-friendly-coloring}}\xspace}

\section{The Coloring Procedure}\label{sec:coloring}

We now describe the coloring procedure that we use to find a $\Delta$-coloring of the graph. 
This procedure is agnostic to the input graph, in the sense that we run it after
processing the stream, and it only uses the information we gathered in the previous section, listed in~\Cref{sec:info} (throughout, we condition on the high probability event that the correct information is collected from the stream).

The general framework in our coloring procedure is the following:
We will maintain a proper partial $\Delta$-coloring $\myC \colon V \to \bracket{\Delta} \cup \cbrac{\perp}$ (as defined in~\Cref{sec:prelim}). 
We then go through different \textbf{phases} in the coloring algorithm and each phase updates $\myC$ by coloring certain subsets of vertices, say, (a subset of) sparse vertices, or certain almost-cliques. 
These new colorings are typically going to be extensions of $\myC$  but in certain cases, we crucially have to go back and ``edit'' this partial coloring, i.e., come up with a new proper partial coloring 
which is no longer an extension the  current one. Eventually, we will color all the vertices of the graph and end up with a proper $\Delta$-coloring. 

In the following, we present the order of the phases of our coloring procedure and the task we expect from each one. 
Each phase shall use a different list of colors $L_1(\cdot),\cdots,L_6(\cdot)$ computed by $\PSalg$ (\Cref{alg:PS}) when updating the partial coloring. We also note that the order 
of running these phases is crucial as some of them present further guarantees for subsequent phases, and some of them need to assume certain properties of the current partial coloring which will no longer remain true
if we change the order of phases. 

On the high level, the coloring procedure is as follows (\Cref{tab:phases} give a summary of which combination of almost-cliques are handled in which phase). 

\begin{itemize}%[leftmargin=10pt]
	\item \textbf{Phase 1 -- One-Shot Coloring (\Cref{sec:one-shot}):} We use the single color sampled in $L_1(v)$ for every vertex $v \in V$ to color a large fraction of vertices. The effect of this coloring is that it ``sparsifies'' the graph for 
	sparse vertices. We note that this part is standard and appears in many other coloring results starting from, to our knowledge,~\cite{MolloyR97}; see~\cite{AssadiCK19a} for more details.   	
	
	\item \textbf{Phase 2 -- Lonely (or Social) Small Almost-Cliques (\Cref{sec:lonely}):} Recall that from Part~\eqref{p5} of~\Cref{sec:info}, we have a list of $\Klonely$ of almost-cliques that contains all lonely almost-cliques and potentially some social ones. 
	We can easily also identify which of these almost-cliques are small (\Cref{def:scl-almost-clique}) based on their size. 
	
	We will color all these small almost-cliques in $\Klonely$ by colors in $L_2(\cdot)$. This requires a novel argument that uses the facts that: $(i)$ these almost-
	cliques are ``loosely connected'' to outside (no ``high degree'' neighbor, formally friend vertices, in their neighborhood), and $(ii)$ the coloring outside only used a limited set of colors, namely, is sampled from $L_1(\cdot)$ and $L_2(\cdot)$
	so far and is thus not ``too adversarial'' (recall the discussion we had in~\Cref{lesson3} regarding necessity of such arguments). Moreover, the coloring in this phase is an extension of the last one. 
	
	\item \textbf{Phase 3 -- Sparse Vertices (\Cref{sec:sparse}):} We then conclude the coloring of all sparse vertices using the sampled lists $L_3(v)$ for every sparse vertex $v \in V$ (by Part~\eqref{p3} of~\Cref{sec:info}, we know these vertices). 
	This part is also a standard argument as a continuation of Phase 1. But to apply this standard argument, we use the fact that even though we interleaved the standard approach with Phase 2 in the middle, since that coloring 
	was only an extension of Phase 1 (meaning it did not \emph{recolor} any vertex colored in Phase 1), the argument still easily goes through. 
	
	The coloring in this phase is also an extension of the last one. However, now that all sparse vertices are colored, we go ahead and remove the color of any vertex which is not sparse and nor is colored by Phase 2 
	(these are remnants of one-shot coloring in Phase 1 and we no longer need
	them now that all sparse vertices are colored). This is just to simplify the analysis for later parts.

	It is worth mentioning that this interleaving of Phase 2 in the middle of Phase 1 and 3 is crucial for our arguments (this is the chicken-and-egg problem mentioned in~\Cref{sec:tech-alg}): the lists $L_3(\cdot)$ used in Phase 3 are much larger than the rest and thus the ``not-too-adversarial'' property of coloring of 
	outside vertices in Phase 2 will no longer be guaranteed had we changed the order of Phase 2 and 3; at the same time, changing the order of Phase 1 and 2 will also destroy the ``sparsification'' guarantee provided by Phase 1 for sparse vertices.

 \def\arraystretch{1.5}
  
 \begin{table}[t!]
        \centering
        \begin{tabular}{|l|c|c|c|c|c|c|c|}
            \cline{1-7}
          \multicolumn{1}{|c|}{\multirow{2}{*}{\textbf{Type}}} &\multicolumn{3}{c|}{\emph{Holey}} & \multicolumn{3}{c|}{\emph{Unholey}}\\
            \cline{2-7}
             & \emph{Friendly} & \emph{Social} & \emph{Lonely} & \emph{Friendly} & \emph{Social} & \emph{Lonely} \\
            \cline{1-7}
            \emph{Small} & Phase 4 & Phases 2 or 4 & Phase 2 & Phase 6 & Phases 2 or 6 & Phase 2  \\ 
              \cline{1-7}
            \emph{Critical} & Phase 4 & Phase 4 & Phase 4 & Phase 5 & Phase 5 & Phase 5  \\ 
              \cline{1-7}
            \emph{Large} & Phase 4 & Phase 4 & Phase 4 & -- & -- & --  \\ 
       \hline
        \end{tabular}
          \caption{A list of all combination of different almost-cliques together with the phase of our coloring procedure that is responsible for handling them. Note that by~\Cref{obs:critical-almost-clique}, 
          there are no unholey large almost-cliques. The sparse vertices are handled in Phase 1 and Phase 3. Moreover, Phase 1 may color some vertices of lonely (or social) small almost-cliques that we are \emph{not} allowed to recolor (we can
          recolor all the other remnants of Phase 1 after Phase 3). 
          \label{tab:phases}}

    \end{table}
    	
	\item \textbf{Phase 4 -- Holey Almost-Cliques (\Cref{sec:holey}):} The next step is to color holey almost-cliques (\Cref{def:holey-almost-clique}), i.e., the ones with $\Omega(\eps\Delta)$ non-edges inside them, using colors sampled in $L_4(\cdot)$. 
	 We note that we actually do \emph{not} know which almost-cliques are holey and which ones are not.%
\footnote{Technically, we could have designed a semi-streaming
	algorithm that also recovers this information about the decomposition (at least approximately). However, as we explain next, this is not needed.} Instead, we simply run this phase over all remaining almost-cliques and argue 
	that all the holey ones (and possibly some other ones) will get fully colored as desired. 
	
	The proof of this phase is a simple generalization of a similar proof used in the palette sparsification theorem of~\cite{AssadiCK19a}, which even though
	quite technical, does not involve much novelty from us in this work. The coloring in this phase is an extension of the last one. 
	
	\item \textbf{Phase 5 -- Unholey Critical Almost-Cliques (\Cref{sec:unholey-critical}):} By~\Cref{obs:critical-almost-clique}, all large almost-cliques are holey. Thus, the largest remaining almost-cliques at this point are unholey critical almost-cliques. 
	We know these almost-cliques in $\Kcritical$ by Part~\eqref{p6} of~\Cref{sec:info} (the holey ones are already colored and it is possible, yet unlikely, that even some of unholey ones are also colored in Phase 4). 
	We color the remainder of $\Kcritical$ now. 
	
	In the previous phases, we solely colored vertices from lists $L(\cdot)$ sampled  in \PSalg. But we already know that such an approach is just not going to work for unholey critical almost-cliques (recall the
	 example in~\Cref{fig:ps-fail} discussed in~\Cref{sec:tech-lower}). This is the first time we deviate from this approach (and thus deviate from palette sparsification-type arguments). 
	 
	We now will use the critical-helper structures (\Cref{def:helper-critical})---which our 
	streaming algorithm collected in Part~\eqref{p6} of~\Cref{sec:info}---and a new ``out of palette''  coloring argument, wherein we color one of the vertices of the almost-clique using a color not sampled for it, so that 
	\emph{two} vertices of the almost-clique are colored the \emph{same}.  We then show that this already buys us enough flexibility to color the remaining vertices using lists $L_5(\cdot)$ of vertices similar to Phase 4. 
	The coloring in this phase is also an extension of the last one.

	\item \textbf{Phase 6 -- Unholey Friendly (or Social) Small Almost-Cliques (\Cref{sec:unholey-friendly}):} It can be verified, after a moment of thought or better yet by consulting~\Cref{tab:phases}, that the only almost-cliques remained to color 
	are the ones that are unholey, small, and also not lonely. They are perhaps the ``most problematic'' ones and are handled last.%
\footnote{It is quite natural to ask if these almost-cliques are the ``hardest'' to color, why do we wait to color them after
	everything else, at which time, our hands might be too tied? There are two closely related answers: $(i)$ they may just be connected to each other (or rather the graph can only consists of these types of almost-cliques) and 
	thus we anyway have to deal with at least one of them after having colored the rest of the graph; and $(ii)$ even though they are ``hard'' to color, they are somewhat ``more robust'' also, compared to say Phase 2 almost-cliques, in that we can 
	color them even when their outside neighbors are colored adversarially by using a key recoloring step.} These almost-cliques also require the ``out of palette'' coloring argument used in Phase 5, but even this is 
	not enough for them (we already discussed this regarding~\Cref{fig:hard} in~\Cref{sec:tech-alg}). In particular, unlike Phase 4 and 5 that allowed for coloring of the almost-cliques even in the presence of adversarial coloring of outside vertices, 
	this simply cannot be true for this phase (as shown in~\Cref{fig:hard}); at the same time, unlike Phase 2 almost-cliques, we cannot hope for  a ``random'' coloring of outside vertices.  
	
	To handle these almost-cliques, we rely on our friendly-helper structures (\Cref{def:helper-friendly}) combined with a \textbf{recoloring step}: in particular, we recolor one vertex outside of the almost-clique 
	using the sampled lists $L_6(\cdot)$ and show that this recoloring, plus another out of palette coloring argument, again buys us enough flexibility to color these almost-cliques also from lists $L_6(\cdot)$ (we note that this 
	out of palette coloring argument is in fact different from the one used in Phase 5). Finally, due to the recoloring step, the coloring in this phase is no longer an extension of the last one. 
\end{itemize}

After all these phases, we have finished coloring all the vertices.
In other words, we now have $\Delta$-coloring of the entire graph as desired.
This will then conclude the proof of~\Cref{res:main}.

In the rest of this section, we go over each of these phases in details and present the algorithm and 
analysis for each one (postponing the less novel ones to~\Cref{app:missing-proofs}). We again emphasize that to find the final $\Delta$-coloring, this phases must be executed \emph{in this particular order}.

\subsection{Phase 1: One-Shot Coloring}\label{sec:one-shot}

We start with the standard \emph{one-shot coloring} algorithm used extensively in the coloring literature (to the best of our knowledge, this idea has appeared first in~\cite{MolloyR97}). 
The purpose of this algorithm is to color many \emph{pairs} of vertices in the neighborhood of sparse vertices using the \emph{same} color (recall that neighborhood of sparse vertices contains many non-edges which can potentially be colored the same). This then effectively turn the sparse vertices into 
``low degree'' ones and reduces the problem from $\Delta$-coloring to $O(\Delta)$-coloring which is much simpler. 

\begin{algorithm}
	\KwIn{The vertex set $V$, the conflict-graph $H$, and the list $L_1(v)$ for every vertex $v \in V$.} 
    \begin{enumerate}[label=$(\roman*)$]
        \item For every vertex $v \in V$:
        \begin{itemize}
            \item \emph{Activate} $v$ independently with probability $1 / \alpha$ for parameter $\alpha = \Theta(1)$ in~\Cref{eq:parameters}. 
            \item If $v$ is activated, set $x(v)$ to be the only color in
            $L_1(v)$, otherwise set $x(v) = \perp$.
        \end{itemize}
        \item For every vertex $v \in V$, set $\myC_1(v) = x(v)$ if for all $u \in N_H(v)$, $x(v) \neq x(u)$; otherwise, set $\myC_1(v) = \perp$. In words, 
       any activated vertex $v$ keeps its color $x(v)$ iff it is not used anywhere in its neighborhood.
    \end{enumerate}
\caption{The $\OneShot$ algorithm.}\label{alg:OS}
\end{algorithm}

We have the following basic observation about the correctness of $\OneShot$. 

\begin{observation}\label{obs:one-shot}
	The partial coloring $\myC_1$ computed by $\OneShot$ is a proper partial $\Delta$-coloring in $G$. 
\end{observation}
\begin{proof}
	It is immediate to verify that $\myC_1$ is a proper partial coloring in $H$ simply because we remove both colors of any monochromatic edge. 
	To see this also holds in $G$, note that for any edge $(u,v) \in G$, if $\myC_1(u) = \myC_1(v) \neq \perp$, then it means that $x(u) = x(v)$ 
	which in particular also means $L_1(u) \cap L_1(v) \neq \emptyset$. Thus, 
	the conflict graph $H$ contains the edge $(u,v)$ also, a contradiction with $\myC_1$ being a proper partial coloring of $H$. 
\end{proof}

We now get to the main property of $\OneShot$. The effect of this partial coloring is that the neighborhood of every \emph{sparse}
vertex $v \in V$ becomes {abundant} with available colors (compared to the remaining degree of $v$). 
In particular, recall the definition of $\cn{v}{\myC_1}$ as the colored
degree of a vertex $v$ and $\avail{v}{\myC_1}$ as the number of colors 
available to $v$ with respect to a partial coloring $\myC_1$ (defined in~\Cref{sec:prelim}).
Then we have the following guarantee for the $\OneShot$ algorithm:

\begin{lemma}\label{lem:os-gap}
    For every sparse vertex $v \in \Vsparse$, in the partial coloring $\myC_1$
    of $\OneShot$, 
    \[
        \avail{v}{\myC_1} > (\deg{(v)} - \cn{v}{\myC_1}) +
        \frac{\eps^2 \cdot \Delta}{2\alpha}
    \]
    with high probability, where the randomness is only over the choice of the
    lists $L_1(v)$.
\end{lemma}

The proof is postponed to \Cref{sec:app-one-shot}.

\subsection{Phase 2: Lonely (or Social) Small Almost-Cliques}\label{sec:lonely}

In this section, we will describe an algorithm that extends the partial
coloring $\myC_1$ of Phase 1 to all small almost-cliques that are in $\Klonely$ of Part~\eqref{p5} of~\Cref{sec:info} (which in particular, contains all lonely almost-cliques and no friendly almost-clique). 
We are going to work with palette graphs introduced in~\Cref{sec:palette-graph} in this phase. The algorithm is simply as follows. 

\begin{algorithm}
	\KwIn{A proper partial $\Delta$-coloring $\myC$, a small almost-clique $K \in \Klonely$, the conflict-graph $H$, and the list $L_2(v)$ for every vertex $v \in K$.}
	
    \begin{enumerate}[label=$(\roman*)$]
        \item Construct the sampled palette graph $\Gsample = (\cL,\cR,\Esample)$ of the almost-clique $K$, $\myC$, and $\SS := \set{L_2(v) \mid v \in K}$ (according to \Cref{def:sampled-palette-graph}).  
        \item Find an $\cL$-perfect matching $\MM$ in $\Gsample$ and output `fail' if it does not exists. Otherwise, update $\myC(v) = \MM(v)$ where $\MM(v)$ denotes the color corresponding to the color-node matched to the vertex-node $v$ by the matching $\MM$. 
        \end{enumerate}
        \caption{The algorithm of Phase 2 for coloring each small almost-clique in $\Klonely$. }\label{alg:phase2}
\end{algorithm}

In Phase 2, we start by setting $\myC$ to be equal to $\myC_1$ of Phase 1 and then successively run~\Cref{alg:phase2} on each small almost-clique $K \in \Klonely$ while updating $\myC$ as described by the algorithm. 
At the end, we let $\myC_2$ denote the final partial $\Delta$-coloring.

\begin{lemma}\label{lem:phase2}
    With high probability, $\myC_2$ computed by Phase 2 is a proper partial $\Delta$-coloring in $G$ that is an extension of $\myC_1$ and colors all small almost-cliques in $\Klonely$. 
  
     The randomness in this lemma is only over the randomness of $\OneShot$ (activation probabilities) and choice of the lists $L_1(v)$ and $L_2(v)$ for all $v \in V$.
\end{lemma}

We prove~\Cref{lem:phase2} in the rest of this subsection. 
	The fact that $\myC_2$ is an extension of~$\myC_1$ follows immediately from the definition of the algorithm as for every $K \in \Klonely$, the corresponding 
	$\Gsample$ only contains vertices uncolored by $\myC_1$ and we never change the color of any colored vertex. Moreover, 
	the fact that $\myC_2$ is a proper $\Delta$-coloring follows from the definition of $\Gsample$ as described  in~\Cref{sec:palette-graph} (note that we only need edges in $H$ and not all of $G$ to construct $\Gsample$). 

The main part of the proof in this phase is to show that we actually succeed in coloring all small almost-cliques in $\Klonely$ in this phase, i.e., w.h.p.,~\Cref{alg:phase2} does not  ever return `fail'.

Fix a small almost-clique $K \in \Klonely$. Consider the base palette graph $\Gbase = (\cL,\cR,\Ebase)$ of $K$ and $\myC$ (\Cref{def:palette-graph}) where $\myC$ denotes the partial coloring passed to~\Cref{alg:phase2} when coloring $K$. 
First, note that $\card{\cL} \leq \card{\cR}$  since $\card{K} \leq \Delta$ as $K$ is small (\Cref{def:scl-almost-clique}), and we remove
at most one color-node in $\cR$ per each vertex-node in $\cL$ removed from $K$. Thus, having an $\cL$-perfect matching in $\Gbase$ and $\Gsample$ is not \emph{entirely} out of the question. 
We now establish two other  properties of $\Gbase$  that will allow us to argue that $\Gbase$  has an $\cL$-perfect matching. We will then build on 
these properties to prove the same for $\Gsample$ as well. 

\begin{claim}\label{clm:2-deg-L}
    Every vertex in $\cL$ in $\Gbase$ has degree at least $3\Delta/4$, with high probability.
 \end{claim}
    \begin{subproof}
		Fix a vertex-node $v \in \cL$ (and $v \in K$). By property~\ref{dec:neighbors} of \Cref{def:almost-clique}, $v$ has
		at most $10\eps \Delta$ neighbors outside $K$, each of which can rule out at most one color for $v$. 
		As for neighbors inside $K$, recall that each vertex $u \in K$ activates in \Cref{alg:OS} with probability
	$1 / \alpha$. This implies that in expectation, at most $\Delta / \alpha$
	of them can receive a color in $\myC_1$ (and hence $\myC$). By an application
	of Chernoff bound (\Cref{prop:chernoff}), at most $2 \Delta / \alpha < \Delta / 100$ vertices
	of $K$ are colored by $\myC_1$ with high probability (for the choice of $\alpha$ in~\Cref{eq:parameters}). 
	
	Hence, in total,
	only $10\eps\Delta + \Delta / 100$ colors are ruled out for $v$ by $\myC$ with high probability. Given the value of $\eps$ in~\Cref{eq:parameters}, 
	this means $\deg_{\Gbase}(v) > 3\Delta/4$ with high probability. %\Qed{clm:2-deg-L}
	    \end{subproof}

The following claim---albeit  simple to prove after having setup the process carefully---is the heart of the argument in this phase. Roughly speaking, this claim allows us to treat the coloring outside 
the almost-clique as ``not too adversarial''.

\begin{claim}\label{clm:2-deg-R}
     Every vertex in $\cR$ in $\Gbase$ has degree at least $3\Delta/4$, with high probability.
 \end{claim}
 \begin{proof}
	To lower bound the degree of a color-node $c \in \cR$, we have to work a little harder.
    The main idea is this: for a color-node $c$ to lose its edges to $\Omega(\Delta)$ vertex-nodes
    in $\cL$, the color $c$ itself must have been used to color $\Omega(\beta)$ vertices outside $K$ by
    $\myC$; this is because of the crucial condition that $K$ is \emph{not} friendly
    and hence each vertex receiving the color $c$ outside of $K$ only rules $c$ out for less than
    $2\Delta / \beta$ vertex-nodes of $\cL$. It is generally very hard to keep track of which colors are assigned by $\myC$ so far in the neighborhood of $K$, but fortunately
    we have a loose but simple proxy for that: as $\myC$ is using only the colors in $L_1(\cdot)$ and $L_2(\cdot)$  at this point, we can simply consider which colors are sampled in these lists 
    in the neighborhood of $K$ instead. We formalize this in the following.

    Fix a color-node $c \in \cR$. For every vertex $v \notin K$, we define $m(v, K)$ as the
    number of edges from $v$ to $K$ in $G$. We define the random variable $X_{c,v}$ which is equal to $m(v,K)$ iff $c \in L_1(v) \cup L_2(v)$ and otherwise $X_{c,v} = 0$. 
    Notice that $X_{c,v}$ is a (potentially loose) upper bound on the \emph{reduction} in the degree of color-node $c \in \cR$ because of any assignment of a color the vertices outside of $K$. 
    In other words, we have that,
    \begin{align}
    	\deg_{\Gbase}{(c)} \geq \card{\cL} - \sum_{v \notin K} X_{c,v} \,. \label{eq:2-reduction}
    \end{align}
    We use the variable $X_{c,v}$, instead of the actual color assignment of $v$ by $\myC$, for two reasons: One, it is easy to compute its probability, and
    second (and more importantly) these variables for different $v$'s are \emph{independent} (while the actual color of vertices will be correlated). 
    We would like to show that random variable $X_c := \sum_{u \notin K} X_{c, u}$ is sufficiently small.

    Recall that $c \in L_1(v)$ with probability $1 / \Delta$, and it is in
    $L_2(v)$ with probability $\beta / \Delta$ by the choice of lists in $\PSalg$ (\Cref{alg:PS}). 
    Thus, it is in the union of the two lists with probability at most $2 \beta / \Delta$. Hence, 
    \begin{align*}
        \expect{ X_c } &\leq \sum_{v \notin K} m(v, K) \cdot \frac{2\beta}{\Delta} = \sum_{u \in K} \card{N(u) \setminus K} \cdot \frac{2 \beta}{\Delta} \tag{by a simple double counting argument for edges between $K$ and its neighbors} \\
        &\leq \card{K} \cdot 10\eps\Delta \cdot  \frac{2 \beta}{\Delta} \tag{by property~\ref{dec:neighbors} of almost-cliques in~\Cref{def:almost-clique}} \\
        &\leq \Delta \cdot 20\eps \cdot \beta \tag{as $\card{K} \leq \Delta$ since $K$ is small by~\Cref{def:scl-almost-clique}} \\
        &\leq \Delta/100 \tag{by the choice of $\eps,\beta$ in~\Cref{eq:parameters}}. 
    \end{align*}

   To prove a concentration bound for $X_c$, note that it is a sum of independent random variables in the range $[0,2\Delta/\beta]$, as each vertex $v \notin K$ has 
   less than $2\Delta/\beta$ neighbors in $K$ (as $K$ is not a friendly almost-clique and thus has no friend neighbors---see~\Cref{def:friend-stranger}). Thus, 
   by Chernoff bound (\Cref{prop:chernoff} for $b=2\Delta/\beta$), we have,  
    \begin{align*}
        \prob{ X_c > (1 + 20) \cdot \frac {\Delta}{100} } &\leq \exp\paren{ -\frac{20^2 \cdot \Delta / 100}{(3 + 20) \cdot 2\Delta/\beta} } = \exp\paren{ -\frac{400\log{n}}{46} } < n^{-8}, 
    \end{align*}
    by the choice of $\beta$ in~\Cref{eq:parameters}. A union bound over all choices of $v \notin K$ and $c \in [\Delta]$, combined with~\Cref{eq:2-reduction}, implies that with high probability, 
    $\deg_{\Gbase}{(c)} \geq \card{\cL} - 21\Delta/100$. 
    
    Finally, as already proven in~\Cref{clm:2-deg-L}, $\card{\cL} \geq \card{K} - \Delta/100$ with high probability as at most $\Delta/100$ vertices of $K$ are colored by $\myC_1$. 
    Given that size of $K$ is also at least $(1-5\eps)\Delta$ by property~\ref{dec:size} of almost-cliques in~\Cref{def:almost-clique}, and by the choice of $\eps$ in~\Cref{eq:parameters}, 
    we get that with high probability $\card{\cL} \geq \Delta - 2\Delta/100$. Combined with the above bound, we have 
    \[
    \deg_{\Gbase}{(c)} \geq \Delta - 2\Delta/100 - 21\Delta/100 > 3\Delta/4
    \]
    as desired, concluding the proof.  
 \end{proof}

 Given that size of $\cL$ in $\Gbase$ is at most $\Delta$, it is now easy to use~\Cref{clm:2-deg-L,clm:2-deg-R}, combined with Hall's theorem (\Cref{fact:halls-theorem}) 
 to prove that $\Gbase$ has an $\cL$-perfect matching. But, our goal is to prove that $\Gsample$,  not only $\Gbase$, has such a matching; this is the content of the next claim. 
 
 \begin{claim}\label{clm:2-matching}
    The subgraph $\Gsample$ has an $\cL$-perfect matching with high probability.
 \end{claim}

\begin{proof}
	We  condition on the high probability events that $\Gbase$ has the properties in~\Cref{clm:2-deg-L,clm:2-deg-R}. An important observation is in order: the properties of $\Gbase$ 
	depend on the choice of $L_1(v)$ for $v \in V$ but only $L_2(v)$ for $v \notin K$ (as we only need to visit the coloring of $\myC$ with $L_2(\cdot)$ \emph{outside} of $K$). As a result, 
	even conditioned on these properties, the choice of $L_2(v)$ for $v \in K$ is independent and from its original distribution in $\PSalg$. 
	
	At this point, $\Gsample$ is a subgraph of $\Gbase$ obtained by sampling each edge independently and with probability $\beta/\Delta$ by the choice of $L_2(\cdot)$ in $\PSalg$ (\Cref{alg:PS}). 
	We use this to prove that $\Gsample$ should also have an $\cL$-perfect matching with high probability. The argument follows standard ideas in random graph theory (even though $\Gsample$ is not exactly a random graph). 
	
	By Hall's theorem (\Cref{fact:halls-theorem}), for $\Gsample$ to \emph{not} have an $\cL$-perfect matching, there should exist a set $A \subseteq \cL$ such that $\card{N_{\Gsample}(A)} < \card{A}$. 
	But for this to happen, there should exist a pair $(S,T)$ of subsets of $\cL$ and $\cR$, respectively, such that $\card{T} = \card{S}-1$ and no edge between $S$ and $\cR \setminus T$ is sampled in $\Gsample$ 
	(simply take $A=S$ and notice that $N(A) \subseteq T$ which has size less than $A$). We refer to any such pair $(S,T)$ as a \textbf{witness pair}. 
	We bound the probability that any witness pair exists in $\Gsample$. 
	
	\paragraph{Case 1: when $\card{S} \leq 2\Delta/3$.} Consider any choice of the set $T$ with $\card{T} = \card{S}-1$ from $\cR$. By~\Cref{clm:2-deg-L}, degree 
	of every vertex-node in $S$ is at least $3\Delta/4$ in $\Gbase$. This means the number of edges from $S$ to $\cR \setminus T$ is at least $\card{S} \cdot (3\Delta/4 - 2\Delta/3) = \card{S} \cdot \Delta/12$. 
	As such, 
	\[
		\Pr\paren{\text{$(S,T)$ is a witness pair}} \leq \paren{1-\frac{\beta}{\Delta}}^{\card{S} \cdot \Delta/12} \hspace{-10pt} \leq \exp\paren{-\frac{100}{12} \cdot \card{S} \cdot \log{n}} < n^{-8\card{S}},
	\]
	by the choice of $\beta$ in~\Cref{eq:parameters}. A union bound over all ${{\card{\cR}} \choose {\card{S}-1}} < n^{\card{S}-1}$ choices for $T$ then implies that for any such $S$, 
	\[
		\Pr\paren{\text{there is a set $T$ so that $(S,T)$ is a witness pair}} \leq n^{\card{S}-1} \cdot n^{-8\card{S}} < n^{-7\card{S}}. 
	\]
	Finally, a union bound all choices for the set $S$, partitioned based on their size, implies that, 
	\[
		\Pr\paren{\text{there is a witness pair $(S,T)$ with $\card{S} \leq 2\Delta/3$}} \leq \sum_{s=1}^{2\Delta/3} {{\card{\cL}}\choose{s}} \cdot n^{-7\card{s}} < n^{-5}. 
	\]
	\paragraph{Case 2: when $\card{S} > 2\Delta/3$.} Again, fix any choice of the set $T$ with $\card{T} = \card{S}-1$ from $\cR$. This time, by~\Cref{clm:2-deg-R}, 
	degree of every color-node \emph{not} in $T$ is at least $3\Delta/4$ in $\Gbase$. This means that neighborhood of each such color-node intersects with $S$ in at least $3\Delta/4 - (\Delta - 2\Delta/3) = 5\Delta/12$ (as $\card{\cL} \leq \Delta$) 
	vertex-nodes. In other words, there are at least $\card{\cR \setminus T} \cdot 5\Delta/12$ edges between $S$ and $\cR \setminus T$ in $\Gbase$. Thus, 
	\[
		\Pr\paren{\text{$(S,T)$ is a witness pair}} \leq \paren{1-\frac{\beta}{\Delta}}^{\card{\cR \setminus T} \cdot 5\Delta/12} \hspace{-10pt} \leq \exp\paren{-\frac{500}{12} \cdot \card{\cR \setminus T} \cdot \log{n}} < n^{-40\card{\cR \setminus T}},
	\]
	by the choice of $\beta$ in~\Cref{eq:parameters}. A union bound over all ${{\card{\cR}} \choose {\card{T}}} = {{\card{\cR}}\choose{\card{\cR \setminus T}}} < n^{\card{\cR \setminus T}}$ choices for $T$ then implies that for any such $S$, 
	\[
		\Pr\paren{\text{there is a set $T$ so that $(S,T)$ is a witness pair}} \leq n^{\card{\cR \setminus T}} \cdot n^{-40\card{\cR \setminus T}}  \leq n^{-39 \cdot (\card{\cR} - \card{S}+1)}, 
	\]
	where we used the fact that $\card{T} = \card{S}-1$. 	Now note that the number of choices for the set $S$ of a fixed size is
	\[
		{{\card{\cL}}\choose{\card{S}}} = {{\card{\cL}}\choose{\card{\cL}-\card{S}}} \leq n^{\card{\cL}-\card{S}} \leq n^{\card{\cR} - \card{S}}, 
	\]
	where the last inequality uses the fact that $\card{\cR} \geq \card{\cL}$. As a result, 
	\[
		\Pr\paren{\text{there is a witness pair $(S,T)$ with $\card{S} > 2\Delta/3$}} \leq \sum_{s=2\Delta/3}^{\card{\cL}} n^{\card{\cR} - \card{S}} \cdot n^{-39 \cdot (\card{\cR}-\card{S}+1)} < n^{-39}. 
	\]
	
	\smallskip
	
	Finally, by combining Case 1 and 2 above, we have that with high probability, there is no witness set $(S,T)$ in $\Gsample$. This implies that for every $A \subseteq \cL$, we have $\card{N_{\Gsample}(A)} \geq \card{A}$, 
	which, by Hall's theorem (\Cref{fact:halls-theorem}), implies that $\Gsample$ has an $\cL$-perfect matching.
\end{proof}

\Cref{lem:phase2} now follows immediately from \Cref{clm:2-deg-L,clm:2-deg-R,clm:2-matching} as described earlier. 

 \begin{remark}\label{rem:why-interleave}
 	Before moving on from this subsection, let us mention why our coloring procedure attempts to color lonely small almost-cliques (Phase 2) before the remaining sparse vertices. In~\Cref{clm:2-deg-R}, we 
	crucially used the fact that we can use the lists used to color $\myC$ so far as a proxy for approximating the event $\myC(v) = c$ with $c \in L_1(v) \cup L_2(v)$, instead. This was okay because these lists are of relatively small size to make the argument go through. However, coloring sparse vertices requires us to use the lists $L_3(\cdot)$ which are much larger and thus would break this claim entirely. 
	
	Concretely, in~\Cref{clm:2-deg-R}, 	we had $O(\eps\Delta^2)$ edges going out of the almost-clique and each was responsible for blocking a fixed color on a vertex with probability $p = O(\beta/\Delta)$ which is governed by sizes of $L_1(\cdot),L_2(\cdot)$. This meant that each color was blocked for $O(\eps \cdot \beta \cdot \Delta)$ vertices which can be made $o(\Delta)$ by taking $\eps$ sufficiently smaller than $\beta$.  Nevertheless, 
	had we also included lists $L_3(\cdot)$, then the right probability parameter $p$ would have become $O(\beta/(\eps^2 \Delta))$ which is crucial for coloring $\eps$-sparse vertices; but then, it meant 
	that the bound we got on the number of blocked vertices for a color is actually $O(\eps \Delta^2 \cdot \beta/(\eps^2 \Delta)) = O(\Delta/\eps)$ which is $>\Delta$ no matter the tuning of parameters. 
 \end{remark}

\subsection{Phase 3: Sparse Vertices}\label{sec:sparse}
In this phase, we will describe an algorithm to extend the partial coloring
$\myC_2$ to all vertices of $\Vsparse$. The key observation is that for any
extension of $\myC_1$ (i.e. $\myC_2$, and every intermediate coloring in this
phase), the gap between available colors and remaining
degree created by $\myC_1$ for each sparse vertex (see \Cref{lem:os-gap}) does
not shrink.
This is because if $\myC$ is an extension of~$\myC_1$, each additional neighbor
of some sparse vertex $v$ that $\myC$ colors, increases $\cn{v}{\myC}$ by $1$, and
decreases $\avail{v}{\myC}$ by at most $1$, keeping the gap intact.
We have the following  lemma for this phase:

\begin{lemma}\label{lem:color-sparse}
    With high probability, there exists a proper partial $\Delta$-coloring $\myC$ that is an extension of $\myC_2$ and colors 
    all remaining vertices in $\Vsparse$ using only the colors in the lists $L_3(v)$ for sparse vertices $v \in \Vsparse$ (and thus the randomness is also only over these lists).
\end{lemma}

Once again, since the lemma is not new, its proof is postponed to \Cref{sec:app-sparse}.

We note that $\myC$ is \emph{not} the final coloring we obtain in this phase. Instead, we are going to update $\myC$ to a proper $\Delta$-coloring $\myC_3$
that colors all vertices in $\Vsparse$ as well as all small almost-cliques in $\Klonely$ handled by Phase 2; however, we shall remove the color of every other vertex $v$, i.e., set $\myC_3(v) = \perp$ for them. Such vertices 
are solely colored by the $\OneShot$ algorithm and we no longer need their guarantees as we are done coloring sparse vertices. Thus, to summarize:
\begin{itemize}
\item $\myC_3$ is a proper partial $\Delta$-coloring of all vertices in $\Vsparse$ as well as small almost-cliques in $\Klonely$ and does not color any other vertex (we also require no further properties from $\myC_3$ and it might as well 
be considered adversarially chosen from now on). 
\end{itemize}

\subsection{Phase 4: Holey Almost-Cliques}\label{sec:holey}

In this phase, we will extend the partial coloring $\myC_3$ to all vertices in
holey almost-cliques.

\begin{lemma}\label{lem:holey}
    There exists a proper partial $\Delta$-coloring $\myC_4$ that is an
	extension of $\myC_3$ and assigns a color to every vertex $v$ in each
	holey almost-clique using a color from $L_4(v)$. The randomness in this
	lemma is only over the lists $L_4(\cdot)$ of all vertices.
\end{lemma}

As discussed in the overview of the coloring algorithm, we will iterate over
all almost-cliques in our decomposition, and attempt to color them assuming
that they are holey. The main tool to show that this succeeds on
holey~almost-cliques is the following lemma, which shows that any coloring outside a
holey~almost-clique can be extended to it while using only the lists
$L_4(\cdot)$ on vertices inside it.

\begin{lemma}\label{lem:color-holey-one}
	For a holey almost-clique $K$, and any partial $\Delta$-coloring $\myC$ outside
	$K$, there exists, with high probability, a coloring $\myC'$ which extends
	$\myC$ to $K$ such that $\myC'(v) \in L_4(v)$ for all $v \in K$.
\end{lemma}

The proof is almost verbatim from~\cite{AssadiCK19a}, except that we have to go through 
every step carefully to make sure it works for $\Delta$-coloring---hence we provide it in \Cref{sec:app-holey} for completeness.

\Cref{lem:holey} now follows immediately from~\Cref{lem:color-holey-one} by going over all uncolored almost-cliques at this point one by one, and apply this lemma with $\myC$ being the current coloring, and $\myC'$ being the one
we can update this coloring to. Thus, at the end, we obtain the desired $\myC_4$.

\subsection{Phase 5: Unholey Critical Almost-Cliques}
\label{sec:unholey-critical}

In this phase, we will color unholey critical almost-cliques. In particular,
we have a set $\Kcritical$ of almost-cliques, and for each $K \in \Kcritical$,
a critical-helper structure $(u, v, N(v))$ (\Cref{def:helper-critical}). The
main lemma of this section is:

\begin{lemma}\label{lem:phase5}
    With high probability, there exists a proper partial $\Delta$-coloring
    $\myC_5$ that is an extension of $\myC_4$ and satisfies the following
    properties: $(i)$ it colors vertices of all almost-cliques in $\Kcritical$; 
     and $(ii)$ for any vertex $v$, if $\myC_5(v) \notin L(v)$,
			then $N_{\Hplus}(v) = N_{G}(v)$; that is, we can only color~$v$ with a
            color not from $L(v)$  if we know its entire neighborhood
			(via the critical helper structure).

   The randomness in this lemma is only over the lists $L_5(v)$ for vertices
    $v$ in $\Kcritical$. 
\end{lemma}

This is the first phase in which we use our out-of-palette-coloring idea -- we do not require that $\myC_5(v) \in L(v)$ always holds in this lemma. 
In particular, for the vertex $v$ in the critical-helper $(u, v, N(v))$, we are going to use a color out of its sampled palette. 
Since we know the entire neighborhood of $v$,
we at least have enough information to avoid an improper coloring.

The proof of the lemma is algorithmic. We start with a brief overview.
The plan (as always) is to iterate over the almost-cliques of $\Kcritical$ in
arbitrary order, and extend the coloring $\myC_4$ to eventually color all of
them. For a particular almost-clique $K \in \Kcritical$, we will use the
critical-helper structure $(u, v, N(v))$ to assign the \emph{same} color to both $u$ and $v$.
The reason we can do this is that we know all the neighbors of $v$, and hence
can pick a color in the list $L_5(u)$ that can be assigned to both of them -- existence of such a color in the first place is because both $u$ and $v$ belong 
to an almost-clique and thus have at most $O(\eps\Delta)$ edges to outside; thus, as long as we have sampled a color out of these many, which will happen with high probability, 
we can find such a color. Having done that, the rest of $K$ can be colored by palette sparsification:
the imbalance we create by giving two vertices the same color is just enough
for it to succeed with high probability.

\begin{algorithm}
	\KwIn{critical almost-cliques $\Kcritical$, a critical-helper structure}
	for each $K \in \Kcritical$, and the partial coloring $\myC_4$.
	\begin{enumerate}[label=$(\roman*)$]
		\item Initialize $\myC \gets \myC_4$.  For each $K \in \Kset$:
            \begin{itemize}
				\item Let $(u, v, N(v))$ be the critical-helper structure for
                    $K$ of Part~\eqref{p6} of~\Cref{sec:info}. 
				\item Find a color $c \in L_5(u)$ that is not used
					in $N_{\Hplus}(u) \cup N_{\Hplus}(v)$ by $\myC$, and set
					\[
					\myC(u) \gets c \quad \text{and} \quad \myC(v) \gets c.
					\]
                \item Extend $\myC$ to color $K \setminus \cbrac{u, v}$ with
                    the color of each vertex $w$ chosen from $L_5(w)$, by
					constructing the sampled palette graph of $K$ with respect
					to $\myC$ and sampled lists $\SS := \set{L_5(w) \mid w \in K}$ 
                    exactly as in~\Cref{alg:phase2}
					(see \Cref{clm:palette-phase5} for details).
            \end{itemize}
        \item Return $\myC_5 \gets \myC$ as the output coloring. 
	\end{enumerate}
        	\caption{The $\ColUHCC$ algorithm.}\label{alg:Col-UCAC}
\end{algorithm}

This algorithm claims to color all of the almost-cliques in $\Kcritical$, but
it is not obvious at all that each of its steps is possible. We will prove
this in the following series of claims.

\begin{claim}
    With high probability,
	there exists a color $c \in L_5(u)$ such that $c$ is not used by $\myC$
	in $N_G(u) \cup N_G(v)$. Further, there is an algorithm that can find
	this color if it exists.
\end{claim}
\begin{proof}
    Recall that $u$ and $v$ have at most $10\eps \Delta$ neighbors each outside
    $K$, and hence only at most $20\eps \Delta$ neighbors in total which are
	colored by $\myC$ (see \Cref{fig:use-critical-helper}).
	\begin{sidefigure}[h!]
	\centering
	\pgfdeclarelayer{edges}       % hide edges behind ellipses
\pgfdeclarelayer{ellipses}    % hide edges behind vertices
\pgfsetlayers{edges,ellipses,main}
% set the order of the layers (main is the standard layer)

\newcommand{\colorgridbr}{\tikz{
\draw[xstep=0.3cm,ystep=0.1cm] (0,0)  grid (0.3,0.2); 
\filldraw[fill=blue!50] (0,0) rectangle (0.3,0.1);
\filldraw[fill=red!50] (0,0.1) rectangle (0.3,0.2);
}}

\begin{tikzpicture}
    [vertex/.style={circle, draw=blue!50, fill=blue!20, thick, inner sep=0pt,
     minimum size=2mm},
     red/.style={vertex, draw=red!50, fill=red!20},
     ghost/.style={inner sep=0pt, minimum size=0mm},
     scale=1]
    \node [vertex] (u) [label=above left:{\footnotesize$u$}] at (4, 4.5) {};
	\node [vertex] (v) [label=above left:{\footnotesize$v$}] at (4, 1.5) {};

    \draw (u) edge [dashed] (v);

    \node [ghost] (K) [label=right:{\footnotesize $K$}] at (4, 1) {};

    \draw (4, 3) ellipse (1.5cm and 3cm); % the almost-clique K
    \begin{pgfonlayer}{ellipses}
        \draw [fill=gray!20] (1.5, 4.5) ellipse (0.5cm and 1cm); % N(u)
        \draw [fill=gray!20] (1.5, 1.5) ellipse (0.5cm and 1cm); % N(v)
    \end{pgfonlayer}

    %\node [ghost] (w1) [] at (1.5, 2.75) {};
    %\node [ghost] (w2) [] at (1.5, .75) {};

	\node [ghost] (v1) [] at (1.5, 2.5) {};
	\node [ghost] (v2) [] at (1.5, 0.5) {};

	\node [ghost] (u1) [] at (1.5, 5.5) {};
	\node [ghost] (u2) [] at (1.5, 3.5) {};

    \begin{pgfonlayer}{edges}
        \draw (v) edge (v1);
        \draw (v) edge (v2);
        \draw (u) edge (u1);
        \draw (u) edge (u2);
    \end{pgfonlayer}

	\node [ghost, font=\footnotesize, align=left] (Nv) [] at (-1, 1.5)
	{$\leq 10\eps\Delta$ colored\\neighbors};

	\node [ghost, font=\footnotesize, align=left] (Nu) [] at (-1, 4.5)
	{$\leq 10\eps\Delta$ colored\\neighbors};
    % \coordinate[below=.7cm of v1] (v11);
    % \coordinate[below=.7cm of v2] (v22);
    % \draw [decorate,decoration={brace,amplitude=10pt,mirror}]
    % (v11) -- (v22) node [black, midway, xshift=-.5cm, yshift=-.75cm] (Nv) {$N(v)$};

    % \draw [help lines] (0,0) grid (7, 7);

	\node (L5u) [right = 3pt of u]{\colorgridbr};
\end{tikzpicture}
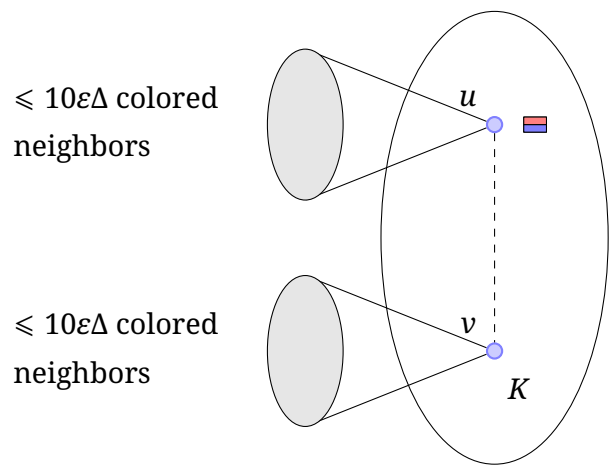
	\caption{$L_5(u)$ has a color not used by $\myC$ over $N(u) \cup N(v)$.}
	\label{fig:use-critical-helper}
\end{sidefigure}
	
	Let $\Avail{u, v}{\myC}$ be the set of colors in
	$\bracket{\Delta}$ that are not used in the neighborhood of $u$ or $v$ by
	the coloring $\myC$. Then
	\[\card{\Avail{u, v}{\myC}} \geq \Delta - 20\eps \Delta \geq \Delta / 2.\]
	We would like to show that
	$\prob{ L_5(u) \cap \Avail{u, v}{\myC} = \emptyset}$ is small.
    Since each color from $\Avail{u, v}{\myC}$ is sampled into $L_5(u)$
    independently with probability $\beta / \Delta$, the probability that none
    of them are in $L_5(u)$ is at most:
    \[
        \paren{1 - \frac{\beta}{\Delta}}^{\Delta / 2}
        \leq \exp\paren{ -\frac{\beta \cdot \Delta}{2\Delta} }
        = n^{-50}.
    \]
    Hence a ``good'' color exists with high probability.

	To find this color, we note that the critical-helper structure contains
    $N(v)$ and while we do not necessarily know all the neighbors of $u$ in
    $G$, we do know the ones that:
    \begin{itemize}
        \item Have a color from $L_5(u)$ in their own palette.
        \item Were assigned a color \emph{outside} their palette -- since for such
            vertices we know all of their neighbors by the invariant maintained in~\Cref{lem:phase5}. 
    \end{itemize}
    This means that we can iterate over the colors in $[{\Delta}]$, and
    check for each one whether it is used by
    $N_{\Hplus}(u) \cup N_{\Hplus}(v)$, which serves as a proxy for checking
    $N_G(u) \cup N_G(v)$, and we are done.
\end{proof}

We now have to perform a somewhat daunting task---to extend the coloring $\myC$ to the rest of the almost-clique $K$.
It turns out that the small imbalance we create in the previous line
(by coloring $2$ vertices with $1$ color) is enough for palette sparsification
to come to our rescue.

\begin{claim}\label{clm:palette-phase5}
    With high probability,
    $K \setminus \cbrac{u, v}$ can be colored as an extension of $\myC$, while
    coloring each vertex $w \in K \setminus \cbrac{u, v}$ with a color from
    $L_5(w)$.
\end{claim}
\begin{proof}
    Consider the base and the sampled palette graphs $\Gbase$ and $\Gsample$
	(\Cref{def:palette-graph,def:sampled-palette-graph})
	corresponding to the partial coloring $\myC$, with:
	\begin{itemize}
		\item $\cL := K \setminus \cbrac{u, w}$,
		\item $\cR := \bracket{\Delta} \setminus \cbrac{c}$, where $c$ is
			the color assigned to $u$ and $v$, and
		\item $S(v) := L_5(v)$ for all $v \in \cL$.
	\end{itemize}
	Recall that an $\cL$-perfect matching in $\Gsample$ implies a
	coloring of remaining vertices in $K$ from their $L_5(\cdot)$ lists which is further an extension of $\myC$. % (which extends $\myC$) of $\cL$.
	We are going to use~\Cref{lem:rgt} to obtain that  $\Gsample$ has a perfect matching. 
	
	For this, we need to establish the required properties of $\Gbase$. 
	Proceeding in the same order as the lemma (recall that $m$ denotes  $\card{\cL}$ in this lemma):
	\begin{enumerate}[label=$(\roman*)$]
        \item $m \leq \card{\cR} \leq 2m$: 
			In this instance, $\card{L} = \card{R}$, so both inequalities follow
			trivially.
        \item The minimum degree of any vertex in $\cL$ is at least $2m / 3$:
            Note that any vertex $v \in \cL$ is in $K$, and hence has only
            at most $10\eps\Delta$ edges going out of $K$ (in $G$, the input
            graph). Hence there are at most $10\eps\Delta$ colors that are
            blocked for $v$, and
			$\deg_{\Gbase}(v) \geq \card{\cR} - 10\eps\Delta \geq (2/3)\cdot m$.

		\item For any subset $A \subseteq \cL$ such that $\card{A} \geq m / 2$,
			\[
				\sum_{v\in A} \deg_{\Gbase}(v) \geq \card{A} \cdot m - m/4.
			\]
			Let $t$ denote the number of non-edges in $K$, and recall that
			since $K$ is unholey, $t \leq 10^7 \cdot \eps \Delta < \Delta / 10$.
            Note that any vertex $v \in K$ has at most
            $\Delta - (\card{K} - 1 - \nondeg{K}{v})$ neighbors (in $G$) outside
            $K$, because $\deg_G(v) \leq \Delta$, and $v$ has exactly
            $\card{K} - 1 - \nondeg{K}{v}$ neighbors inside $K$. But since each
            non-edge between $v \in \cL$ and a color in $\cR$ corresponds to an
            edge from $v$ to outside $K$, we have:
            \[
                \deg_{\Gbase}(v) \geq \card{R} - (\Delta + 1) + \card{K} -
                \nondeg{K}{v}.
            \]
			By summing this inequality for all $v \in A$, we get:
			\begin{align*}
				\sum_{v\in A} \deg_{\Gbase}(v) &\geq
				\sum_{v\in A}
				\card{R} - (\Delta + 1) + \card{K} - \nondeg{K}{v} \\
				&=
				\sum_{v\in A}
				\underbrace{\Delta - 1}_{\card{R}} - (\Delta + 1)
				+ \underbrace{\card{\cL} + 2}_{\card{K}} - \nondeg{K}{v} \\
				&=
				\sum_{v\in A}
				\card{\cL} - \nondeg{K}{v}
				=
				\card{A} \cdot m - \underbrace{\sum_{v \in A} \nondeg{K}{v}}_{\leq 2t}  \\
				&\geq
				\card{A}\cdot m - m / 4.
			\end{align*}
    \end{enumerate}
    Each edge of $\Gbase$ is sampled into $\Gsample$ independently with
    probability
    \[
        \beta / \Delta \geq 99 \log n / m \geq
        20/m \cdot (\log n + 3 \cdot \log n).
    \]
    By \Cref{lem:rgt} $\Gsample$ has a perfect matching with probability 
    at least $(1 - n^{-3})$.
\end{proof}

This concludes the proof of~\Cref{lem:phase5}. 

\subsection{Phase 6: Unholey Friendly (or Social) Small Almost-Cliques}
\label{sec:unholey-friendly}
In this (final!) phase, we describe an algorithm that takes the partial coloring $\myC_5$ and from it compute a proper $\Delta$-coloring of the entire graph by coloring the remaining vertices. 
This step however is the one that includes our \emph{recoloring} ideas and thus the coloring we obtain is no longer an extension of $\myC_5$.

We are left with a
set $\Kfriendly$ of unholey small almost-cliques, that may be friendly or
social. For each $K \in \Kfriendly$, we have a friendly-helper
$(u, v, w, N(v), N(w))$ (\Cref{def:helper-friendly}).
Note in particular that we will actually edit the color assigned to some
vertices, so the coloring obtained by this algorithm is not necessarily an
extension of $\myC_5$. We show the following lemma:

\begin{lemma}\label{lem:phase6}
    Given the coloring $C_5$ from the previous section,
    with high probability there exists a proper coloring $\myC_6$ that colors
    the entire graph such that:
    For any vertex $v$, if $\myC_6(v) \notin  L(v)$,
    $N_G(v) = N_{\Hplus}(v)$.
    That is, we can color $v$ with a
    color not from $L(v)$ only if we know its entire neighborhood
    (via the friendly-helper structure).

    The randomness in this lemma is over the lists $L_{6, i}(v)$ for \underline{all}
    vertices $v \in V$ and all $i \in [2\beta]$ (even including vertices that
	were colored before this phase). %, and the lists
 %   $L_5(v)$ for vertices $v$ in an almost-clique in $\Kfriendly$.
\end{lemma}

Once again, we will prove the lemma with an algorithm. The idea is the same,
to extend the coloring $C_5$ to each almost-clique in $\Kfriendly$ one-by-one,
with a key difference being that we will go back and edit the color of one
vertex per almost-clique we color. 

For an almost-clique $K$, we will
use its friendly-helper structure $(u, v, w, N(v), N(w))$ to assign the
same color to $u$ and $w$, and then invoke palette sparsification to color
the rest of $K$.

\begin{algorithm}
	\KwIn{The set of almost-cliques $\Kfriendly$, a friendly-helper
    structure for each $K \in \Kfriendly$, and the partial coloring $C_5$.}
	\begin{enumerate}[label=$(\roman*)$]
        \item Initialize $\myC = \myC_5$, and an index $i_u \gets 0$ for each
			vertex $u \in V(G)$ (this keeps track of how many times we recolored $u$).
        \item For each $K \in \Kfriendly$:
            \begin{itemize}
                \item 
					Let $(u, v, w, N(v), N(w))$ be the friendly-helper
                    structure of $K$.
                \item Find a color $c \in L_{6, i_u}(u)$ such that
					$c \notin \myC( N_{\Hplus}(u) \cup N_{\Hplus}(w) )$ and set 
					\[
					\myC(u) \gets c \quad \text{and} \quad \myC(w) \gets c.
					\]
			Update $i_u \gets i_u + 1$ (we emphasize that $u$ is not part of $K$ but rather a neighbor to it). 		
                \item Extend $\myC$ to color $K \setminus \cbrac{v, w}$ with
                    each remaining vertex $x$ getting a color from $L_{6,2\beta}(x)$
                    exactly as in~\Cref{alg:Col-UCAC} (see \Cref{clm:palette-phase6} for details).
				\item Extend $\myC$ to color $v$ by finding a color that does not appear in $N_{H+}(v)$. %with a color from $[\Delta]$.
            \end{itemize}
        \item Return $\myC_6 \gets \myC$ as the output.
	\end{enumerate}
        	\caption{The $\ColUHFC$ algorithm.}\label{alg:Col-UFSAC}
\end{algorithm}

As before,
while the algorithm claims to color all the almost-cliques in $\Kfriendly$,
it is far from obvious that each step it performs is possible. We show that
this is indeed the case.

\begin{claim}
    With high probability,
    there exists a color in $L_{6, i_u}(u)$ that does not appear in
	$\myC$ over $N(u) \cup N(w)$. Further, there is an algorithm that can
	find this color if it exists.
\end{claim}
\begin{proof}
    The crucial observation is that since $u$ is not a stranger to $K$, it
    has at least $\Delta / \beta$ uncolored neighbors in the coloring
    $\myC$. Further, since $w \in K$ -- an almost-clique -- it has at most
    $10\eps\Delta$ neighbors outside $K$, and hence only at most
    $10\eps\Delta$ neighbors that receive colors in $\myC$. Hence there
    are at least $\Delta / \beta - 10\eps\Delta \geq \Delta / 2\beta$
    colors that are not used by $\myC$ in $N(u) \cup N(w)$
	(see \Cref{fig:use-friendly-helper}).

\begin{sidefigure}[h!]
	\centering
	\pgfdeclarelayer{edges}       % hide edges behind ellipses
\pgfdeclarelayer{ellipses}    % hide edges behind vertices
\pgfsetlayers{edges,ellipses,main}
% set the order of the layers (main is the standard layer)

\newcommand{\colorgridbr}{\tikz{
\draw[xstep=0.3cm,ystep=0.1cm] (0,0)  grid (0.3,0.2); 
\filldraw[fill=blue!50] (0,0) rectangle (0.3,0.1);
\filldraw[fill=red!50] (0,0.1) rectangle (0.3,0.2);
}}

\begin{tikzpicture}
    [vertex/.style={circle, draw=blue!50, fill=blue!20, thick, inner sep=0pt,
     minimum size=2mm},
     red/.style={vertex, draw=red!50, fill=red!20},
     ghost/.style={inner sep=0pt, minimum size=0mm},
     scale=1]
%    \node [red] (v) [label=above:{\footnotesize $ v$}] at (4, 4.5) {};
    \node [vertex] (w) [label=below:$w$] at (4, 3) {};

    \node [vertex] (u) [label=below:{\footnotesize $u$}] at (1.5, 4) {};
	%\node (L6u) [above left = 1pt and 3pt of u]{\colorgridbr};
	\node (L6u) [left = 3pt of u]{\colorgridbr};

    \draw (u) edge [dashed] (w);

    \node [ghost] (K) [label=right:{\footnotesize $K$}] at (4, 1.5) {};

    \begin{pgfonlayer}{ellipses}
        \draw (4, 3) ellipse (1.5cm and 3cm); % the almost-clique K
        \draw [fill=gray!20] (1.5, 1.75) ellipse (0.5cm and 1cm); % N(w)
%        \draw [fill=gray!20] (1.5, 4) ellipse (0.5cm and 1cm); % N(v)
    \end{pgfonlayer}

    \node [ghost] (w1) [] at (1.5, 2.75) {};
    \node [ghost] (w2) [] at (1.5, .75) {};

    \node [ghost] (v1) [] at (1.5, 5) {};
    \node [ghost] (v2) [] at (1.5, 3) {};

    \node [ghost] (u1) [] at (4, 5.5) {};
    \node [ghost] (u2) [] at (4, 3.5) {};

    \draw (u) edge [out = 0, in = 180] (u1);
    \draw (u) edge [out = 0, in = 180] (u2);

    \begin{pgfonlayer}{edges}
        \draw [fill=white] (4, 4.5) ellipse (.5cm and 1cm); % N(u)
        \draw (w) edge (w1);
        \draw (w) edge (w2);
%        \draw (v) edge (v1);
%        \draw (v) edge (v2);
    \end{pgfonlayer}

    % \coordinate[left=.4cm of v1] (v11);
    % \coordinate[left=.4cm of v2] (v22);
    % \draw [decorate,decoration={brace,amplitude=10pt, mirror}]
    % (v11) -- (v22) node [black, midway, xshift=-1.25cm] (Nv) {$N(v) \setminus K$};

%    \node [ghost] (Nv) at (-0.25, 4) {\footnotesize $N(v) \setminus K$};

    % \coordinate[left=.4cm of w1] (w11);
    % \coordinate[left=.4cm of w2] (w22);
    % \draw [decorate,decoration={brace,amplitude=10pt, mirror}]
    % (w11) -- (w22) node [black, midway, xshift=-1.25cm] (Nw) {$N(w) \setminus K$};

	\node [ghost, align=left, font=\footnotesize] (Nw) at (-0.75, 1.75)
	{$\leq 10\eps\Delta$ colored\\ neighbors};

    \coordinate[right=.4cm of u1] (u11);
    \coordinate[right=.4cm of u2] (u22);
    \draw [decorate,decoration={brace,amplitude=10pt}]
    (u11) -- (u22) node [black, midway, xshift=7pt] (Nu) {};

	\node [ghost, align=left, font=\footnotesize] (Nulabel) [right=1cm of Nu]
    {$\geq \frac{\Delta}{\beta}$ uncolored\\ neighbors};

    \draw [->] (Nu) edge (Nulabel);

%    \draw (v) edge (w);

    %\draw [help lines] (0,0) grid (7, 7);
\end{tikzpicture}
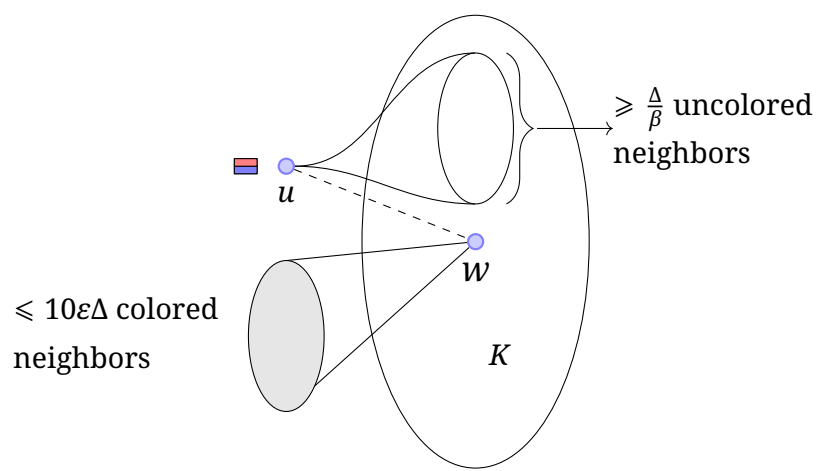
	\caption{$L_{6, i_u}(u)$ has a color not used by
		$\myC$ over $N(u) \cup N(w)$.}
	\label{fig:use-friendly-helper}
\end{sidefigure}

	Let $\Avail{u, w}{\myC}$ denote this set of colors.
    We are interested in showing that the probability
    $ \prob{ L_{6, i_u}(u) \cap \Avail{u, w}{\myC} = \emptyset} $ is small.
	We start with the technical note that this is the first time (any part of)
	the coloring algorithm is looking at $L_{6, i_u}(u)$, and hence this
	particular list is independent of $\myC$ entirely.
	Recall that each color from $\Avail{u, w}{\myC}$ is in $L_{6, i_u}(u)$
	independently with
	probability $\beta^2 / \Delta$.
    Hence the probability that none of them is in $L_{6, i_u}(u)$ is at most:
    \[
        \paren{1 - \frac{\beta^2}{\Delta}}^{\Delta / 2\beta}
        \leq \exp\paren{ -\frac{\beta^2 \cdot \Delta}{2\beta \cdot \Delta} }
        = n^{-50}.
    \]
	So there is a color $c$ that satisfies our requirements with high
	probability. To find $c$ we can iterate over the colors of $L_{6, i_u}(u)$
	and check whether $c$ is used by $\myC$ in
	$N_{\Hplus}(u) \cup N_{\Hplus}(w)$. As before, this is a proxy
    for checking that $c$ is used in $N_G(u) \cup N_G(w)$, and it works because
    $c \in L(u)$, $N_{\Hplus}(w) = N_{G}(w)$, and if any neighbor of $u$ picked
    a color out of its palette, it will be in $N_{\Hplus}(u)$.

	Note that we increment $i_v$ in the final line to get a ``fresh'' list for
	the next time we have to color $v$---this happens at most $\beta$ times
	for a vertex $v$ (since it must be a non-stranger to some almost clique $K$
	to be recolored), and $L_6$ has $2\beta$ lists for each $v$, so there are
	always enough lists to go around.
\end{proof}

Next, we have the palette  sparsification analogue of this phase.

\begin{claim}\label{clm:palette-phase6}
    With high probability,
    $K \setminus \cbrac{v, w}$ can be colored as an extension of $\myC$,
    where the color of each vertex $x \in K\setminus \cbrac{v, w}$ is from
    $L_{6,2\beta}(x)$.
\end{claim}
\begin{proof}
    Consider the base and the sampled palette graphs $\Gbase$ and $\Gsample$
	(Definitions~\ref{def:palette-graph} and~\ref{def:sampled-palette-graph})
	corresponding to the partial coloring $\myC$, with:
	\begin{itemize}
		\item $\cL = K \setminus \cbrac{v, w}$,
		\item $\cR = \bracket{\Delta} \setminus \cbrac{c}$, where $c$ is
			the color assigned to $u$ and $w$, and
		\item $S(x) = L_{6,2\beta}(x)$ for all $x \in \cL$.
	\end{itemize}
	Recall that an $\cL$-perfect matching in $\Gsample$ gives a coloring which
	extends $C$ to $\cL$ using only colors from the list $L_{6, 2\beta}$.
	Hence our focus shifts to showing that $\Gbase$ has the properties required
	by \Cref{lem:rgt} to obtain that $\Gsample$ has a perfect matching.
	Proceeding in the same order as the lemma:

	\begin{enumerate}[label=$(\roman*)$]
        \item $m \leq \card{\cR} \leq 2m$: The first inequality follows from
            the fact that $\card{\cR} = \Delta - 1$, and $m \leq \Delta - 2$
			(since $K$ is a \emph{small} almost-clique, and we removed two
			vertices from it).
			For the second one, note that
			$2m = 2\card{K} - 4 \geq 3/2 \cdot \Delta$.
        \item The minimum degree of any vertex in $\cL$ is at least $2m / 3$.
            Note that any vertex $v \in \cL$ is in $K$, and hence has only
            at most $10\eps\Delta$ edges going out of $K$ (in $G$, the input
            graph). Hence there are at most $10\eps\Delta$ colors that are
            blocked for $v$, and
            $\deg_{\Gbase}(v) \geq \card{\cR} - 10\eps\Delta \geq 2m/3$
            (the last inequality is from combining parts (i) and (ii) above).

		\item For any subset $S \subset \cL$ such that $\card{S} \geq m / 2$,
			\[
				\sum_{v\in S} \deg_{\Gbase}(v) \geq \paren{\card{S} \cdot m} -
				m/4.
			\]
			Since the only facts we use are that $K$ is unholey, and points
			$(i)$ and $(ii)$ above, the proof is exactly the same as in
			\Cref{clm:palette-phase5}.
    \end{enumerate}
	And once again, by invoking \Cref{lem:rgt} we are done.
\end{proof}

To finish up, note that $v$ has two neighbors ($u$ and $w$) with the same
color, and we know its entire neighborhood. Hence we can find a color to
assign to it, and we are done.

This concludes the proof of~\Cref{lem:phase6}. As at this point, all vertices of the graph are colored, we obtain a proper $\Delta$-coloring of $G$. This in turn concludes the proof of~\Cref{res:main}. 

\subsection*{Acknowledgements} 

Sepehr Assadi would like to thank Soheil Behnezhad, Amit Chakrabarti, Prantar Ghosh, and Merav Parter for helpful conversations. 
We thank the organizers of DIMACS REU in Summer 2020, in particular Lazaros Gallos, for initiating this collaboration and all their help and encouragements along the way.   
We are also thankful to the anonymous reviewers of STOC 2022 and TheoretiCS for helpful comments and suggestions on the presentation of the paper.

\printbibliography

\newcommand{\vecx}{\ensuremath{ \mathbf{x} }}
\newcommand{\vecy}{\ensuremath{ \mathbf{y} }}
\newcommand{\A}{\ensuremath{\mathcal{A}}}

\newcommand{\stcomp}[1]{\ensuremath{ \overline{#1} }}

\appendix
\part*{Appendix}

\section{Proofs of the Impossibility Results}\label{sec:lower}
We present the formal proofs of our impossibility results, alluded to~\Cref{sec:tech-lower}, in this appendix. 
The first is in the query model, and shows that it is essentially impossible to do better
than the trivial algorithm that learns the entire graph. 
The second is in the streaming setting, and shows that it is essentially impossible to do
better than to store the entire graph in a slightly non-standard model where edges of the input can appear more than once in the stream. 

\paragraph{Notation:} We will use a boldface $\vecx$ to denote a vector
(or a bit string), and $x_i$ to index it.

\subsection{Sublinear-Time Algorithms}\label{sec:sub-time}
In this section, we show that there is no algorithm in the general query model
that solves $\Delta$-coloring in $o(n\Delta)$ queries, via a reduction from
the $\AndORone_{t, m}$ problem.

We assume that the vertices of the graph $G = (V, E)$ are \emph{known}, as is
the maximum degree $\Delta$. The general query model supports the following
queries on $G$:
\begin{itemize}
    \item Degree queries: Given a vertex $v \in V$, output $\deg(v)$.
    \item Neighbor queries: Given a vertex $v$, and an index
        $i \in \bracket{\Delta}$, output the $i$-th neighbor of $v$, or~$\perp$ if $i > \deg(v)$.
    \item Pair queries: Given two vertices $u$ and $v \in V$, output whether
        $\cbrac{u, v}$ is an edge in $G$ or not.
\end{itemize}

The $\OR$ problem on $N$ bits is: Given query access to a bit-string
$\vecx \in \cbrac{0, 1}^N$, determine whether there exists an $i$ such that
$x_i = 1$. By query access, we mean that the algorithm can ask for any
$i \in \cbrac{N}$ whether $x_i$ is $0$ or $1$. It is well known that
the randomized query complexity $R(\OR_N)$ is $\Omega(N)$
(see~\cite{BuhrmanW02}).
We will reduce the following promise version of the problem to
$\Delta$-coloring (which is also known to have randomized query complexity
$\Omega(N)$):

\begin{problem}[$\ORone_N$]
    Given query access to a string $\vecx \in \cbrac{0, 1}^N$ such that the
    Hamming weight of
    $\vecx$ is at most $1$, determine whether there is an index $i$ such that
    $x_i = 1$.
\end{problem}

Let $N = \binom{n}{2}$.
The reduction from $\ORone_N$ to $\Delta$-coloring is the following:
For a bit-string $\vecx \in \cbrac{0, 1}^N$, we will define a graph $G$ on the
vertex set $U = \cbrac{u_1, \ldots , u_n} \cup V = \cbrac{v_1, \ldots , v_n}$.
We index $\vecx$ as $x_{i, j}$ for $1 \leq i < j \leq n$, and add edges to $G$
as follows:
\begin{itemize}
    \item If $x_{i, j} = 0$, add the internal edges $\cbrac{u_i, u_j}$ and
        $\cbrac{v_i, v_j}$.
    \item If $x_{i, j} = 1$, add the crossing edges $\cbrac{u_i, v_j}$ and
        $\cbrac{v_i, u_j}$.
\end{itemize}

Note that the graph $G$ is $(n - 1)$-regular (so $\Delta = n - 1$).
Now, if $\vecx = 0^N$, then the sets $U$ and~$V$ have no edges between them, and
$G$ is just two copies of $K_n$, and hence not $\Delta$-colorable. If, on the
other hand, $x_{i, j} = 1$ for one pair $(i, j)$ then $G$ is two copies of
$K_n$ minus an edge, connected by a pair of cross edges, and by Brooks' Theorem
$G$ is $\Delta$-colorable. To finish the reduction, we need to show that we
can simulate each of the queries of the general query model on $G$ with at most
a single query on $\vecx$:
\begin{itemize}
    \item Degree queries: We always return $n - 1$, without looking at any
        bit of $\vecx$.
    \item Neighbor queries: To get the $j$-th neighbor of $u_i$ (or $v_i$),
        we need to look at the bit:
        \[
            \begin{cases}
                x_{j, i} &\text{if } j < i\\
                x_{i, j + 1} &\text{if } i \leq j \leq n - 1 \\
            \end{cases}
        \]
    \item Pair queries: Assume without loss of generality that the query is
        for the pair $(u_i, v_j)$ such that $i < j$. Then we can answer it by
        just looking at the bit $x_{i, j}$.
\end{itemize}
And hence, a $o(n^2)$ query algorithm for $\Delta$-coloring a graph on $2n$
vertices implies a $o(N)$ query algorithm for the $\ORone$ problem on
$N = \binom{n}{2}$ bits. Note that we can pad an arbitrary instance of
$\ORone_N$ to an instance we can reduce with at most a constant blowup in size,
since $(n + 1)^2 \leq 2n^2$ for all $n$ large enough.

\begin{lemma}\label{lem:lb-query}
    The randomized query complexity of $\Delta$-coloring a graph on $n$
    vertices is $\Omega(n^2)$ for some choice of $\Delta = \Theta(n)$. 
\end{lemma}

We can take this idea further: Suppose instead that we have $\Theta(n / \Delta)$
instances of $\ORone$ on $\binom{\Delta + 1}{2}$ bits. In particular, define
the $\AndORone$ problem:

\begin{problem}[$\AndORone_{t, m}$]
    Given query access to a set of $t$ strings
    $\vecx_1, \ldots , \vecx_t \in \cbrac{0, 1}^{m}$, compute:
    \[
        \bigwedge_{i = 1}^t \ORone_{m}(\vecx_i).
    \]
\end{problem}

Then it is well known that the randomized query complexity of
$\AndORone_{t, m}$ is $\Omega(tm)$ (see~\cite{BunT15},~\cite{Sherstov13}).
We will take an instance of $\AndORone$ with $m = \binom{\Delta + 1}{2}$ and
$t \in \IN$---setting $n = t \cdot 2m$---and reduce it to
$\Delta$-coloring a graph on $2n$ vertices. In particular, the graph will just
be the (disjoint) union of the graphs formed by reducing each of the
$\ORone_m$ instances in the fashion described above. Note that the graph will
be $\Delta$-regular.
It is immediate that a $o(n\Delta)$-query algorithm for $\Delta$-coloring
implies a $o(t\cdot \Delta^2) = o(tm)$-query algorithm for $\AndORone$, 
and we have a contradiction.

\begin{lemma}
    The randomized query complexity of $\Delta$-coloring a graph on $n$
    vertices is $\Omega(n\Delta)$ for all choices of $100 \leq \Delta < n/100$. 
\end{lemma}

\subsection{Streaming Algorithms on Repeated-Edge Streams}\label{sec:or-streams}
In this section, we will show that any $O(1)$ pass algorithm for
$\Delta$-coloring on graph streams with repeated edges needs $\Omega(n\Delta)$
space to color a graph on $n$ vertices and maximum degree $\Delta$. In
particular, we will reduce the $\tribes$ problem from communication complexity
to the $\clique_{\Delta + 1}$ problem. We start by defining the
$\tribes_{m, n}$ problem.

\begin{problem}[$\tribes_{m, n}$]
    In the $\tribes_{m, n}$ problem, Alice and Bob receive $m$ vectors
    $\vecx_1, \ldots , \vecx_m \in \cbrac{0, 1}^n$ and
    $\vecy_1, \ldots , \vecy_m \in \cbrac{0, 1}^n$ respectively.
    They want to compute the function:
    \[
        \bigwedge_{k = 1}^m \disjoint_n(\vecx_k, \vecy_k).
    \]
    Where $\disjoint_n$ is the standard disjoint function that is true iff
    its inputs differ in all of their bits.
\end{problem}

In words, the problem is just to solve $m$ instances of $\disjoint_n$, and
return true only if all of them are disjoint.
By a result of~\cite{JayramKS03}, the randomized communication complexity of
$\tribes_{m, n}$ is $\Omega(mn)$.
We will show a low-communication protocol
for $\tribes_{m, n}$ assuming a small-space algorithm for
$\clique_{\Delta + 1}$, which is defined as follows:

\begin{problem}[$\clique_{\Delta + 1}$]
    Given a graph $G$ as a stream, determine whether or not $G$ contains a
    $(\Delta + 1)$-clique.
\end{problem}

First, we will do a warm-up reduction, from $\disjoint_N$ to $\clique$, which
simply asks if the input graph is a clique.
Suppose that $N = \binom{n}{2}$ for some $n$,%
\footnote{This is easily achieved by blowing up the universe by at most a
constant factor.}
and we have an instance $\vecx, \vecy \in \cbrac{0, 1}^N$ of $\disjoint_N$.
Then we can relabel the indices to $(u, v)$ such that $1 \leq u < v \leq n$.
Let $A$ be the set $\cbrac{ (i, j) \mid x_{i, j} = 1 }$, and $B$ be the
same for $\vecy$.
Define a graph $G$ on the vertex set $V = \cbrac{v_1, \ldots , v_n}$ as
follows:
Add the edge $\cbrac{v_i, v_j}$ to $G$ iff $\cbrac{i, j}$ is in
$\stcomp{A} \cup \stcomp{B}$. Note that $A$ and $B$ are disjoint if and only
if $\stcomp{A} \cup \stcomp{B}$ is the entire universe, which is equivalent
to $G$ being a clique.

Assume that we have a constant-pass $o(n^2)$ space algorithm for
$\Delta$-coloring (and hence for $\clique$). Then here is a low-communication
protocol for $\disjoint_N$:
Alice will form ``half'' the stream by taking the
edge set $\stcomp{A}$, and run the streaming algorithm for $\clique$ on it,
and communicate the state of the algorithm (using $o(n^2)$ bits) to Bob. Bob
will then continue running the algorithm on his ``half'' of the stream
(the edge set $\stcomp{B}$), and hence finish one pass of the algorithm over
the edges of $G$. If
there are additional passes required, Bob will communicate the state of the
algorithm to Alice (again, using $o(n^2)$ bits) and the process will repeat.
After a constant number of passes, the streaming algorithm will decide whether
$G$ is a clique, and hence if $A$ and $B$ are disjoint, having used
$o(n^2) \cdot O(1)$ communication.
And hence we have shown the following lemma:

\begin{lemma}
    There is no constant pass $o(n^2)$ space streaming algorithm for $\clique$.
\end{lemma}

We can use the same idea when starting with an instance of $\tribes_{m, k}$
where $m = \binom{\Delta + 1}{2}$, and $k \in \IN$%
\footnote{And setting $n = k \cdot (\Delta + 1)$.}
to rule out a $o(n\Delta)$~space algorithm for $\clique_{\Delta + 1}$
(and hence $\Delta$-coloring).
In particular, define the $k$-partite graph $G$ on the vertex sets
$V_1, \ldots, V_k$, such that each $G(V_i, E_i)$ is the graph that is a clique
iff $\vecx_i$ and $\vecy_i$ are disjoint.

Then supposing we have a constant~pass $o(n\Delta)$~space algorithm for
$\clique_{\Delta + 1}$, we can use exactly the same low communication protocol
as before: Alice and Bob each construct half of the graph stream, and to
complete a pass exchange the entire state of the algorithm twice. This implies
a $o(n \Delta) = o(k \Delta^2) = o(mk)$ communication protocol for
$\tribes_{m, k}$, and we have the contradiction we desire.

\begin{lemma}
    There is no constant pass $o(n\Delta)$ space streaming algorithm for
    $\Delta$-coloring.
\end{lemma}

The crucial observation is that in the stream created in the communication
protocol, the edges of $G$ can \emph{repeat}. In particular, if the bit
$(i, j)$ is zero in both $\vecx$ and $\vecy$,
both Alice and Bob add the edge $\cbrac{v_i, v_j}$ to
their halves of the stream. Suppose that our algorithm for $\Delta$-coloring
(and hence $\clique$) is only required to work on graph streams with no
repeated edges. Then the low communication protocol breaks down completely, and
hence it is actually possible to solve the $\Delta$-coloring problem in
$o(n\Delta)$ space.

\paragraph{A technical remark:}
Note that both of the lower-bounds we prove are for the problem which tests
whether a graph is $\Delta$-colorable, but they also apply to the promise
version of the problem where the graph is guaranteed to be $\Delta$-colorable,
and the task is to output a coloring. In particular, suppose we have a
$\Delta$-coloring algorithm $\A$ on a graph stream, then we can use it for
the reduction from $\disjoint_N$ as follows:
\begin{itemize}
    \item Run $\A$ on the input stream for the graph $G$ with edge set
        $\stcomp{A} \cup \stcomp{B}$, communicating between Alice and Bob as
        before.
    \item If $\A$ fails to output a coloring, then the sets $A$ and $B$ are
        disjoint (w.h.p.).
    \item If $\A$ outputs a coloring, then we need to \emph{test} it.
        In particular, if the coloring is improper, there exists an edge in
        $G$ that is monochromatic. Alice and Bob can independently test their
        halves of the edges for such an edge, and if they don't find one,
        the coloring is valid, and hence $A$ and $B$ are disjoint.
\end{itemize}
A similar argument shows that the query lower-bound applies to the promise
version too.

\section{Missing Proofs from Section~\ref{sec:coloring}}\label{app:missing-proofs}

In this appendix, we show all the results whose proofs we skipped in
\Cref{sec:coloring}. The common factor between these results is that they are
all small modifications of previous work and are presented here only
for completeness. 

\paragraph{Preliminaries:} We use a standard form of Talagrand's inequality~\cite{Talagrand95} as
specified in~\cite{MolloyR13}.
A function $f(x_1,\ldots,x_n)$ is called \textbf{$c$-Lipschitz} iff changing
any  $x_i$ can affect the value of $f$ by at most $c$.
Additionally, $f$ is called \textbf{$r$-certifiable} iff whenever
$f(x_1,\ldots,x_n) \geq s$, there exist at most $r \cdot s$ variables
$x_{i_1},\ldots,x_{i_{r \cdot s}}$ so that knowing the values of these
variables certifies $f \geq s$.

\begin{proposition}[Talagrand's inequality; cf.~\cite{MolloyR13}]\label{prop:talagrand}
  Let $X_1,\ldots,X_m$ be $m$ independent random variables and
  $f(X_1,\ldots,X_m)$ be a $c$-Lipschitz function; then for any $t \geq 1$,
  \begin{align*}
    \Pr\paren{\card{f - \expect{f}} > t } \leq 2 \exp\paren{-\frac{t^2}{2c^2 \cdot m}}.
  \end{align*}
  Moreover, if $f$ is additionally $r$-certifiable, then for any $b \geq 1$,
  \begin{align*}
    \Pr\paren{\card{f - \expect{f}}
    > b + 30c \sqrt{r \cdot \expect{f}}} \leq 4 \exp\paren{-\frac{b^2}{8c^2 r \expect{f}}}.
  \end{align*}
\end{proposition}

\subsection{From Section~\ref{sec:one-shot}}\label{sec:app-one-shot}

\begin{lemma}[Re-statement of \Cref{lem:os-gap}]
    For every sparse vertex $v \in \Vsparse$, in the partial coloring $\myC_1$
    of $\OneShot$, 
    \[
        \avail{v}{\myC_1} > (\deg{(v)} - \cn{v}{\myC_1}) +
        \frac{\eps^2 \cdot \Delta}{2\alpha}
    \]
    with high probability, where the randomness is only over the choice of the
    lists $L_1(v)$.
\end{lemma}

\begin{proof}
    We consider the random variable $\gap$ that counts the
    colors assigned (by $x$) to \emph{at least} two neighbors of $v$, and retained
    (in $\myC_1$) by \emph{all} of them. Note that if $\gap$ is large, we
    are in good shape: since each color it counts increases the number of
    colored neighbors of $v$ by at least $2$, while decreasing the number of
    colors available to $v$ by $1$. And hence our aim is to show that $\gap$ is
	large with high probability.

    We do this in a roundabout manner. First, we lower bound the expectation
    of $\gap$ with that of the random variable $\gap'$, which counts the number
    of colors assigned to \emph{exactly} two vertices in $N(v)$, and retained
    by both of them. The main reason for this is that $\expect{\gap'}$ is
    easy to calculate, and large enough for our purposes. In particular,
    let $F$ be the set of non-edges between the neighbors of $v$, that is
    \[F = \cbrac*{ \cbrac{u, w} \subset N(v) \mid \cbrac{u, w} \notin E(G) }.\]
    Then since $v$ is sparse, $\card{F} \geq \eps^2 \cdot \Delta^2 / 2$
	(by \Cref{def:sparse}).

    For a color $c \in \bracket{\Delta}$ and a non-edge
	$f = \cbrac{u, w} \in F$, let $\gap'_{c, f}$ indicate the event:
    \begin{itemize}
        \item $x(u) = x(w) = c$,
        \item no other vertex in $N(v) \cup \cbrac{v}$ receives the color $c$
            from $x$, and
        \item no vertex in $N(u) \cup N(w)$ receives $c$.
    \end{itemize}
    By definition, $\gap' = \sum_{c, f} \gap'_{c, f}$. We have:
    \[
        \prob{\gap'_{c, f}} \geq \frac {1}{\Delta^2} \cdot
            \paren{ 1 - \frac{1}{\Delta}}^{3\Delta}
            \geq \frac{1}{\Delta^2} \cdot
            \exp\paren{ - \frac{2}{\Delta} \cdot 3\Delta}
            = \frac{1}{e^6 \cdot \Delta^2}.
    \]
    Where the first inequality follows because
    $\card{N(u) \cup N(w) \cup N(v)} \leq 3\Delta$, and the second because
    \[
        \exp\paren{-\frac{2}{\Delta}}
        \leq 1 - \frac{2}{\Delta} + \frac{2}{\Delta^2}
        \leq 1 - \frac{1}{\Delta}.
    \]
    By linearity of expectation,
    \[
        \expect{\gap} \geq \expect{\gap'} \geq
        \underbrace{%
            \vphantom{\frac{\eps^2 \cdot \Delta^2}{2}}
            \Delta}_{\text{choose } c} \cdot
        \underbrace{\frac{\eps^2 \cdot \Delta^2}{2}}_{\text{choose } f} \cdot
        \frac{1}{e^6 \cdot \Delta^2}
        = \frac {\eps^2 \cdot \Delta}{2e^6}.
    \]

    Next, we want to show that $\gap$ is concentrated around its expectation,
    and for once, Chernoff does not suffice. Define the random variable
    $\assign$ which counts the number of colors \emph{assigned} (by $x$) to
    at least two vertices in $N(v)$, and the random variable $\lose$ which
    counts the number of colors assigned to at least two vertices in $N(v)$,
    and lost by \emph{any} of them. Then clearly
    $\gap = \assign - \lose$. We will show that $\assign$ and $\lose$ are both
	concentrated around their means, and this implies that $\gap$ is too.

	First, note that $\assign$ depends only on the assignment $x(w)$ for all
	$w$ in the neighborhood of $v$, that is:
	\[ \assign \colon \prod_{w \in N(v)} \bracket{\Delta} \to \IN. \]
	Further, it is $1$-Lipschitz -- changing $x(w)$ from $c$ to $c'$
	can:
	\begin{itemize}
		\item Make it so $c$ occurs only once (instead of twice) in $N(v)$,
			decreasing $\assign$ by $1$.
		\item Make it so $c'$ occurs twice (instead of once) in $N(v)$,
			increasing $\assign$ by $1$.
	\end{itemize}
	And hence the net change to $\assign$ from changing $x(w)$ is at most $1$
	in absolute value. Then by the first part of Talagrand's Inequality
	(\Cref{prop:talagrand}, with $m = \Delta$, $c = 1$, and
	$t = \frac{\eps^2 \cdot \Delta}{20e^6}$)
	we have:
	\[
		\prob{ \card{\assign - \expect{\assign}} \geq
		\frac{\eps^2 \cdot \Delta}{20e^6}}
%%		\leq 2\cdot \exp\paren{ -\paren{\frac{\eps^2 \cdot \Delta}{20e^6}}^2 \cdot \frac {1}{\Delta}}
		\leq 2\cdot
		\exp\paren{ -\frac{2\eps^4 \cdot \Delta^2}{800e^{12}\cdot \Delta} }.
	\]
	Which for $\Delta = \Omega(\log^5 n)$ (the $\Omega$ hides a monstrous
	constant), is $1 / \poly(n)$. Note that our choice of $t$ is $1/10$-th of
	the lower bound on the expected value of $\gap$.

	The same argument does not work for $\lose$---the random variable depends
	on the $2$-hop neighborhood of $v$, which has size roughly $\Delta^2$, and
	hence the bound we get above is too weak.
	However, note that in addition to being
	$2$-Lipschitz, $\lose$ is also $3$-certifiable. More concretely, let $W$
	denote the $2$-hop neighborhood of $v$, then:
	\begin{itemize}
		\item Changing $x(w)$ for some $w$ in $W$ can change the contribution
			of at most $2$ colors to $\lose$: the old and the new color
			assigned by $x$ to $w$. Hence $\lose$ is $2$-Lipschitz.
		\item For any $s$, to get $\lose \geq s$, we need to set $x(\cdot)$ for
			three vertices (two neighbors of $v$, and one of their common
			neighbors) to $i$, for $i \in \bracket{s}$.
	\end{itemize}

	Then by second part of Talagrand's Inequality (\Cref{prop:talagrand},
	with $c = 2$, $r = 3$, $b = \frac{\eps^2 \cdot \Delta}{20e^6}$):
	\[
		\prob{\card{ \lose - \expect{\lose} } >
		\frac{\eps^2 \cdot \Delta}{20e^6} + 60\sqrt{3\cdot \expect{\lose}}}
		\leq
		4\cdot \exp\paren{ - \paren{\frac{\eps^2 \cdot \Delta}{20e^6}}^2
		\cdot \frac{1}{48\cdot \expect{\lose}}}.
	\]
	Crudely, we upper bound $\lose$ (and hence $\expect{\lose}$) by
	$\Delta / 2$, to give us:
	\[
		\prob{\card{ \lose - \expect{\lose} } >
		\frac{\eps^2 \cdot \Delta}{20e^6} + 60\sqrt{3/2 \cdot \Delta}}
		\leq
		4\cdot \exp\paren{ - \frac{\eps^4}{9600e^{12}} \Delta }.
	\]
	Which with $\Delta = \Omega(\log^5 n)$ (and $\Omega$ doing an even braver
	job) is $1 / \poly(n)$.

	Taking the union of the bad events (i.e. either $\assign$ or $\lose$
	deviates too much from its mean), and applying the triangle inequality, we
	have that with high probability:
	\[
		\gap > \expect{\gap} - \frac{\eps^2\cdot \Delta}{10e^6} - O(\sqrt{\Delta}).
	\]
	And hence combining with the lower bound on $\expect{\gap}$ we found
	earlier, $\gap > \frac{\eps^2\cdot \Delta}{2000}$ with high probability.
\end{proof}

\subsection{From Section~\ref{sec:sparse}}\label{sec:app-sparse}

\begin{lemma}[Re-statement of \Cref{lem:color-sparse}]
    With high probability, there exists a proper partial $\Delta$-coloring $\myC$ that is an extension of $\myC_2$ and colors 
    all remaining vertices in $\Vsparse$ using only the colors in the lists $L_3(v)$ for sparse vertices $v \in \Vsparse$ (and thus the randomness is also only over these lists).
\end{lemma}

\begin{proof}
	Let $\myC$ initially be  the coloring $\myC_2$. 
    We will color the vertices of $\Vsparse$ greedily by updating $\myC$. That is, we iterate over
    $\Vsparse$ in arbitrary order, and if a vertex $v$ is uncolored in
    $\myC$, we pick a color in $L_3(v)$ which does not
    conflict with any of its neighbors in $H$, and set $\myC(v)$ to it.
    Using the randomness of $L_3(v)$, and conditioning on the high probability
    event of \Cref{lem:os-gap}:

    \begin{claim}\label{clm:L3-has-color}
        With high probability,
        for each vertex $v \in \Vsparse$, and any partial coloring $\myC$ that
        extends $\myC_1$, there exists a color $c \in L_3(v)$ that is not used
        in $N_H(v)$ by $\myC$.
    \end{claim}
    \begin{subproof}
        We would like to show that $\Avail{v}{\myC}$ and $L_3(v)$ have a
        nonzero intersection with high probability.
        Since $v$ is sparse, by \Cref{lem:os-gap} we have:
        \[
            \avail{v}{\myC} > (\deg{(v)} - \cn{v}{\myC_1}) +
            \frac{\eps^2 \cdot \Delta}{2\alpha}
        \]
        Which implies that even if the entire neighborhood of $v$ is colored
        by $\myC$, there are still $\frac{\eps^2}{2\alpha}\cdot \Delta$ colors
        available for $v$ to use.

        Since each color from $\Avail{v}{\myC}$ is sampled into $L_3(v)$
        independently with probability
        $\frac{100 \alpha\cdot \log n}{\eps^2 \cdot \Delta}$, 
        \[
        	\prob{{\Avail{v}{\myC} \cap L_3(v) = \emptyset}} \leq \paren{1-\frac{100 \alpha\cdot \log n}{\eps^2 \cdot \Delta}}^{  \frac{\eps^2 \cdot \Delta}{2\alpha}} \leq 
	\exp\paren{-\frac{100 \alpha\cdot \log n}{\eps^2 \cdot \Delta} \cdot \frac{\eps^2 \cdot \Delta}{2\alpha}} < n^{-50},
        \]
        concluding the proof. 
%    \Qed{clm:L3-has-color}
    \end{subproof}
    And hence the greedy algorithm can color $\Vsparse$. Note that this
    algorithm can be implemented efficiently using the information we gathered
    in \Cref{sec:alg}. In particular, since $\myC$ only assigns colors
    from $L(v)$ to $v$, it is enough to check for conflicts in $H$ while
    coloring from $L_3(v)$. 
%  \Qed{lem:color-sparse}
\end{proof}

\subsection{From Section~\ref{sec:holey}}\label{sec:app-holey}

\begin{lemma}[Re-statement of \Cref{lem:color-holey-one}]
	For a holey almost-clique $K$, and any partial $\Delta$-coloring $\myC$ outside
	$K$, there exists, with high probability, a coloring $\myC'$ which extends
	$\myC$ to $K$ such that $\myC'(v) \in L_4(v)$ for all
	$v \in K$.
\end{lemma}

The main thing we want to exploit is that $K$ has a lot of non-edges, since
we can assign the same color to both endpoints of a non-edge.
One easy way to do this for many non-edges simultaneously is to generate a
``matching'' of non-edges, which we do via the following algorithm:

\begin{algorithm}
    \KwIn{The input graph $G$, an almost-clique $K$,
    the partial coloring $\myC$, and the list $L_{4, i}(v)$ for each $v$.}
	\begin{enumerate}[label=$(\roman*)$]
		\item Initialize $\myC' \gets \myC$.
        \item Let $F = \cbrac{f_1, \ldots , f_t}$ be the set of non-edges in
            $K$.
        \item For each color $c \in \bracket{\Delta}$:
            If there is a non-edge $f = \cbrac{u, v} \in F$ such that
            $c \in L_{4, i}(u) \cap L_{4, i}(v)$, and $c$ is not used in the
            neighborhood of $u$ or $v$, then:
            \begin{itemize}
                \item Assign $\myC'(u) \gets c$ and $\myC'(v) \gets c$.
                \item Add $f$ to $M$.
                \item Remove all non-edges incident on $f$ from $F$.
            \end{itemize}
	\end{enumerate}
	\caption{The $\ColMalg$ algorithm.}\label{alg:Col-Matching}
\end{algorithm}
It follows immediately that $M$ is a matching; we will show that for holey
cliques, $M$ is \emph{large} with constant probability.

\begin{lemma}\label{lem:big-matching}
    If $K$ is a holey clique, \Cref{alg:Col-Matching} produces a matching of
	size $\ell~=~\frac{t}{10^6 \eps\Delta}$ with probability at least $1/2$.
\end{lemma}
\begin{proof}
    We define:
    \begin{itemize}
        \item $\Present(c)$ as the set of non-edges present in $F$ when
            we encounter color $c$ in line (ii) of \Cref{alg:Col-Matching}.
            Let $\present(c)$ denote $\card{\Present(c)}$.
        \item A color $c$ is \emph{successful} if we add assign it to a
            non-edge during \Cref{alg:Col-Matching}.
    \end{itemize}
    Note that the number of successful colors is exactly the size of $M$,
    and hence we would like to show that many colors succeed.
    To do so, we show that $\present(\cdot)$ is large for many colors.
    In particular, we say that a color $c$ is \emph{heavy} if it has
    $\present(c) \geq t/2$, and have the following claim:

    \begin{claim}\label{clm:heavy-vertices}
        There are at least $\Delta / 2$ heavy colors.
    \end{claim}
    \begin{proof}
        For a non-edge $f = \cbrac{u, v}$, define $\Blocked(f)$ to be the set
		of colors used by $\myC$ to color the neighbors of $u$ or $v$ outside
        $K$. By property~\cref{dec:neighbors} of \Cref{def:almost-clique},
        the number of neighbors of $u$ or $v$ (and hence the number of colors
        used by $\myC$ to color them) is at most $20\eps\Delta$. As a
        result, at the beginning of the algorithm:
        \[
            \sum_{f \in F} \card{ \Blocked(f) } \leq t \cdot 20\eps\Delta.
        \]
        This means that on average, each color occurs in $\Blocked(f)$ for at
        most
        \[
            t\cdot 20\eps\Delta \cdot 1/\Delta = 20\eps t
        \]
        non-edges $f$.
%		(the inequality follows from $q \geq \Delta / 2$).
        By Markov's Inequality, there are at most $\Delta / 2$ colors $c$ which
        occur in $\Blocked(f)$ for more than $40 \eps t$ non-edges $f$.
        Hence there are at least $\Delta / 2$ colors $c$ which are
		\emph{not blocked}
        for at least $t - 40\eps t \geq 9/10 \cdot t$ non-edges in $F$---at the beginning of the algorithm.

        How many of these non-edges remain in $\Present(c)$ when we look at
        $c$? Upon adding the non-edge $\cbrac{u, v}$ to $M$, we remove all
        non-edges incident on $u$ or $v$ from $F$. Since each $u \in K$ has
        at most $10\eps\Delta$ non-neighbors in $K$
        (property~\ref{dec:non-neighbors} of \Cref{def:almost-clique}), each
        non-edge added to $M$ removes at most $20\eps\Delta$ edges from
        $\Present(c)$. Because the algorithm has already succeeded if
		$\card{M}$ becomes larger than $\ell$, the number of edges removed from
		$\Present(c)$ is at most $\ell \cdot 20\eps\Delta$ in total, which is
		$< t/10$.
        Hence there are at least $\Delta / 2$ colors $c$ with
        $\present(c) \geq 8/10 \cdot t \geq t/2$.
%        \Qed{clm:heavy-vertices}
    \end{proof}

    Regrouping, we have shown that there are many colors in $\bracket{\Delta}$
	that can be assigned to many non-edges in $F$. Next, we would like to show
	that for each of these heavy colors $c$, the probability that $c$ is in
	sampled by both endpoints of a non-edge in $\Present(c)$ is high.

    \begin{claim}\label{clm:success}
        For a heavy color $c$, $\prob{c \text{ is successful}} \geq 10^{-5}
        \cdot t / \eps\Delta^2$.
    \end{claim}
    \begin{subproof}
        The color $c$ is successful if there is least one non-edge
        $\cbrac{u, v} \in \Present(c)$ such that
        $c \in L_{4, i}(u) \cap L_{4, i}(v)$.
        Using the inclusion exclusion principle, and cutting out terms
        of ``order'' more than $2$:
        \begin{align*}
            \prob{ c \text{ is successful}} &\geq
            \sum_{\cbrac{u, v} \in \Present(c)}
            \prob{c \in L_{4, i}(u) \cap L_{4, i}(v)}\\ &-
			\sum_{f, g\, \in\, \Present(c)}
%            \prob{c \in L_{4, i}(u) \cap  L_{4, i}(v) \cap
%                  L_{4, i}(u') \cap L_{4, i}(v')}
			\prob{c \in \bigcap_{w \in f \cup g} L_{4, i}(w)}
			.
        \end{align*}

        The first term is easy to compute exactly: $c$ belongs to both lists
        with probability $q^2$, so
        \[
            \sum_{\cbrac{u, v} \in \Present(c)}
            \prob{c \in L_{4, i}(u) \cap L_{4, i}(v)} = \present(c) \cdot q^2.
        \]
        For the second term, we have to consider two cases:
        \begin{itemize}
            \item $\card{f \cup g} = 3$: The probability that $c$
                belongs to $L_{4, i}(w)$ for 3 vertices $w$ is $q^3$.
				After picking $f = \cbrac{u, v}$ from $\Present(c)$, there are
                at most $20\eps\Delta$ choices for the third vertex, and hence
                the total contribution of terms of this type is at most
                \[ \present(c) \cdot 20\eps\Delta \cdot q^3. \]
            \item $\card{f \cup g} = 4$: The probability that $c$
                belongs to the list of $4$ vertices is $q^4$. And there are
                at most $\present(c)^2$ ways to pick a $2$-set of non-edges.
                So the total contribution of these terms is at most
                \[ \present(c)^2 \cdot q^4. \]
        \end{itemize}
		Adding everything up, we get:
		\begin{align*}
			\prob{c \text{ is successful}} &\geq
			\present(c) \cdot q^2 - \present(c) \cdot 10\eps\Delta \cdot q^3
			- \present(c)^2 \cdot q^4\\
			&\geq 9/10\cdot \present(c) \cdot q^2 - \present(c)^2 \cdot q^4
			\intertext{\raggedleft
				(since $q = \frac{1}{100\sqrt{\eps}{\Delta}}$ and thus
				$10\eps\Delta \cdot q < 1/10$ for $\eps < 1$)}
			&\geq \present(c) \cdot q^2 \cdot
			\paren{ 9/10 - 20\eps\Delta^2 \cdot q^2 }\\
			\intertext{\raggedleft (since
				$\present(c) \leq \card{F} \leq 20\eps\Delta^2$ by
				property~\ref{dec:non-neighbors} of \Cref{def:almost-clique})}
			&\geq 8/10 \cdot \present(c) \cdot q^2
			\text{\ \ \ \ \ \ \ \ \ \ \ \ \ \ \ \ \ \ \ \ \ \   (since
				$q = \frac{1}{100\sqrt{\eps}\Delta}$,
				$20\eps\Delta^2\cdot q^2 = 1/200$)}\\
			&\geq \frac{t}{10^5 \cdot \eps\Delta^2},
		\end{align*}
	\nopagebreak	concluding the proof.  %  \Qed{clm:success}
    \end{subproof}

	Finally, we are ready to prove the lemma itself. Let
	$\theta = 10^{-5} \cdot t/\eps\Delta^2$ (the RHS of \Cref{clm:success});
	note that $\theta < 1$ because $t \leq \card{F} \leq 20\eps\Delta^2$.
	Let $Z$ be a random variable with the binomial distribution
	$\Bin{\Delta / 2}{\theta}$. Note that we can couple each Bernoulli trial
	used to determine $Z$ with a heavy color succeeding -- since there are
	at least $\Delta / 2$ of them, they succeed with probability at least
	$\theta$, and two different colors are independent of each other.
	Hence $Z$ is a lower bound for~$\card{M}$. We follow our usual formula:
	\[ \expect{Z} = \Delta / 2 \cdot \theta = 10^{-5} \cdot t / 2\eps\Delta.\]
	And with an application of Chernoff Bound (\Cref{prop:chernoff}, with
	$\delta = 1/2$):
	\[
		\prob{Z < 1/2 \cdot \expect{Z}} \leq
		2\exp\paren{-\frac{(1/2)^2\cdot t}{10^5 \cdot 2\eps\Delta \cdot (2 + 1/2)}}
		\leq
		2\exp\paren{-\frac{10^7 \cdot \eps\Delta}{10^6 \cdot 2\eps\Delta}}
		\leq 2e^{-5} < 1/2.
	\]
	And hence with probability at least $1/2$,
	\[ \card{M} \geq Z \geq 10^{-5} \cdot t/8\eps\Delta >
	10^{-6} \cdot t / \eps\Delta = \ell.
	\]
This concludes the proof. 
%\Qed{lem:big-matching}
\end{proof}

Now, since we run \Cref{alg:Col-Matching} for $\beta$ independent sets lists
$\cbrac{ L_{4, i}(v) \mid v \in V}$, we get a non-edge matching $M$ of size
$\ell = t / \eps\Delta$ from one of the runs with high probability. We keep
the coloring assigned to this largest non-edge edge matching, and hence have
the following lemma:

\begin{lemma}\label{lem:big-matching-whp}
	Suppose $K$ is a holey almost-clique, with $t \geq 10^7\cdot \eps\Delta$
	non-edges inside it.
	Then for any coloring $C$ outside $K$, with high probability there is an
	extension $C'$ of $C$ which:
	\begin{itemize}
		\item Colors $2 \cdot \frac{t}{10^6 \cdot \eps\Delta}$ vertices of
			$K$.
		\item Uses only $\frac {t}{10^6 \cdot \eps\Delta}$ colors inside $K$,
			and further for each $v \in K$ that it colors, uses a color from
			$L_4(v)$.
	\end{itemize}
\end{lemma}

Let us pause for a moment, and consider why we did all this work.
\Cref{lem:big-matching-whp} tells us that \Cref{alg:Col-Matching} colors 
\emph{some} number of vertices in $K$, using only half that many unique colors.
As in the case of sparse vertices, this creates enough of a gap between
available colors and the remaining degree of each vertex in $K$ such that
an available color is sampled in $L_4(v)$ with high probability. This 
is exactly what we need to prove \Cref{lem:color-holey-one}.
%%\parth{Why \emph{did} we do all this work?}

\begin{proof}[Proof of Lemma~\ref{lem:color-holey-one}] % \Cref does
                                % not work here
	Let $C'$ be the partial coloring obtained from \Cref{lem:big-matching-whp}.
	Then we define the base and sampled palette graphs $\Gbase$ and $\Gsample$
	(Definitions~\ref{def:palette-graph}~and~\ref{def:sampled-palette-graph})
	with:
	\begin{itemize}
		\item $\cL$ as the set of vertices of $K$ \emph{not} colored by $\myC'$.
		\item $\cR$ as the set of colors \emph{not} used by $\myC'$ in $K$.
		\item $S(v) = L_4^*(v)$ for each $v \in \cL$.
	\end{itemize}
	We will show that $\Gbase$ satsifies the conditions in \Cref{lem:rgt}, and
	hence $\Gsample$ has an $\cL$-perfect matching with high probability, and
	hence $\myC'$ can be extended to $(L_4 \cup L_4^*)$-color all of $K$. In
	the same order as the lemma:

\begin{enumerate}[label=$(\roman*)$]
	\item For $m := \cL$, $m \leq \card{\cR} \leq 2m$:
		To get the first inequality, we rewrite $\card{\cL}$ and $\card{\cR}$
		in terms of $K$ and the vertices and colors removed by $\myC'$:
		\[
			\card{\cL} = \card{K} - \frac{2t}{10^6\cdot \eps\Delta},
			\card{\cR} = \Delta - \frac{t}{10^6 \cdot \eps\Delta}.
		\]
		Note that if $\card{K} \geq \Delta + 1 + k$, the number of non-edges
		inside $K$ (that is, $t$) is at least $k \cdot \Delta$. Then the number
		of vertices removed far outstrips the number of colors, in particular:
		\[
			\card{\cL} = \Delta + 1 + k - \frac{2k}{10^6\cdot \eps} \leq
			\Delta - k,
			\card{\cR} = \Delta - \frac{k}{10^6 \cdot \eps}.
		\]
		Which makes $\card{\cL} \leq \card{\cR}$ for $k \geq 1$, because
		$\frac{1}{10^6\cdot \eps} > 2$. On the other hand, if
		$K = \Delta + 1$, since $t \geq 10^7 \cdot \eps \Delta$,
		\[
			\card{\cR} - \card{\cL} =
			\frac{10^7 \cdot \eps\Delta}{10^6 \cdot \eps\Delta} - 1
			\geq 9.
		\]

		To get the second inequality, first note that $m \geq 2/3 \cdot \Delta$.
		This is because the maximum number of non-edges in $K$ is
		$(2\Delta) \cdot (10\eps\Delta) = 20\eps\Delta^2$ (by
		property~\ref{dec:non-neighbors} of \Cref{def:almost-clique}), and hence the
		number of vertices in $\cL$ is at least
		\[
			(1 - 5\eps)\Delta - \frac{40\eps\Delta^2}{10^6\cdot \eps\Delta}
			\geq 2/3 \cdot \Delta.
		\]
		Then since $\card{R} \leq \Delta$ we are done.
	\item Each vertex $v \in \cL$ has $\deg_{\Gbase}(v) \geq 2/3 \cdot m$:
		Each vertex $v$ in $\cL$ may have up to $10\eps\Delta$ edges out of $K$
		(in $G$), hence blocking $10\eps\Delta$ colors in $\cR$ for $v$. Since
		all the remaining colors are available to $v$,
		\[
		\deg_{\Gbase}(v) \geq \card{R} - 10\eps\Delta \geq m - 15\eps \cdot m
		\geq 2/3 \cdot m.
		\]
		Where the second inequality follows from part (i).
	\item For every set $A \subset \cL$ of size $\card{A} \geq m / 2$, we
		have $\sum_{v \in A} \deg_{\Gbase}(v) \geq \card{A}\cdot m - m/4$;
        recall that for any $v \in \cL$:
        \[
            \deg_{\Gbase}(v) \geq \card{R} - (\Delta + 1) + \card{K} -
            \nondeg{K}{v}.
        \]
		Let $T := \frac{t}{10^6 \cdot \eps\Delta} \geq 10$.
		Summing the inequality above over $v \in A$, we get:
		\begin{align*}
			\sum_{v\in A} \deg_{\Gbase}(v) &\geq
			\sum_{v\in A}
			\paren{
			\card{R} - (\Delta + 1) + \card{K} - \nondeg{K}{v}} \\
			&=
			\sum_{v\in A}
			(
			\underbrace{\Delta - T}_{\card{R}} - (\Delta + 1)
			+ \underbrace{\card{\cL} + 2T}_{\card{K}} - \nondeg{K}{v}) \\
			&=
			\sum_{v\in A} \paren{\card{L} - \nondeg{K}{v} + T - 1}\\
			&=
			\card{A} \cdot (m + T - 1) -
			\underbrace{\sum_{v\in A}\nondeg{K}{v}}_{\leq 2t}
			> \card{A}\cdot m.
		\end{align*}
		Where the last inequality follows from
		\[
			(T - 1) \cdot \card{A} \geq T/2 \cdot \card{A} \geq
			\frac{tm}{10^6 \cdot 4\eps\Delta}
			\geq \frac{t\Delta}{10^6 \cdot 6\eps\Delta} \geq 5t.
		\]
\end{enumerate}
And now by the promised application of \Cref{lem:rgt}, $\Gsample$ has an
$\cL$-perfect matching, and we are done.
%\Qed{lem:color-holey-one}
\end{proof}

\end{document}